\PassOptionsToPackage{table,usenames}{xcolor}
\documentclass[nonblindrev]{informs3customR}
\usepackage{appendix}
\usepackage{algorithm}
\usepackage{algpseudocode}

\usepackage{amsmath}
\usepackage[frak=mma]{mathalfa}
\usepackage{graphicx}
\usepackage{amsfonts}
\usepackage{amssymb}
\usepackage{xcolor}
\usepackage{booktabs} 
\usepackage{mathrsfs}
\usepackage{natbib}
\usepackage[shortlabels]{enumitem}
\usepackage{setspace}

\usepackage[normalem]{ulem}
\usepackage[pagebackref=true, colorlinks=true, linkcolor=blue, citecolor=blue, urlcolor=blue, anchorcolor=blue]{hyperref}

\usepackage{bold-extra} 

\usepackage{url}

\usepackage[capitalise,noabbrev]{cleveref}
\usepackage{physics}

    \usepackage{tikz}
\usetikzlibrary{arrows.meta, positioning, shapes.geometric, fit, calc}

\usepackage{pgfplots}
\usepgfplotslibrary{groupplots}
\pgfplotsset{compat=1.18}
\usetikzlibrary{fillbetween}

\usepackage{hhline}
\usepackage{latexsym}
\usepackage{ifthen}
\usepackage{lscape}
\usepackage{chngpage}
\usepackage{bibunits}
\usepackage{relsize}
\usepackage{bm} 
 
 \usepackage{chngcntr}
\counterwithout{table}{section}

\newtheorem{prop}{Proposition}

\newtheorem{lem}{Lemma}

\newtheorem{rem}{Remark}

\newtheorem{exm}{Example}

\def\E{\mathbb{E}}

\def\eps{\epsilon}

\usepackage[english]{babel}

\newcommand{\e}{\mathbb{E}}

\newcommand{\D}{\mbox{\rm d}}
\newcommand{\Var}{\mathbb{V}\mbox{\rm ar}}
\newcommand{\Pro}{\mathbb{P}}

\newcommand{\LB}{\mbox{\tiny \rm L}}
\newcommand{\UB}{\mbox{\tiny \rm U}}
\newcommand{\ba}{\bar{\alpha}}

\newcommand{\C}[1]{\mathcal{#1}}

\DefineNamedColor{named}{OliveGreen}    {cmyk}{0.64,0,0.95,0.40}
\DefineNamedColor{named}{BrickRed}      {cmyk}{0,0.89,0.94,0.28}
\DefineNamedColor{named}{SeaGreen}      {cmyk}{0.69,0,0.50,0}
\DefineNamedColor{named}{BlueGreen}     {cmyk}{0.85,0,0.33,0}
\DefineNamedColor{named}{BlueViolet}    {cmyk}{0.86,0.91,0,0.04}
\DefineNamedColor{named}{NavyBlue}      {cmyk}{0.94,0.54,0,0}

\renewcommand{\thefootnote}{\fnsymbol{footnote}}

\newtheorem{dfn}{Definition}

\newtheorem{cor}{Corollary}

\newtheorem{assum}{Assumption}

\newcommand{\ti}[1]{\scalebox{0.45}{\mbox{\rm #1}}} 

\DefineNamedColor{named}{NavyBlue}{cmyk}{0.94,0.54,0,0}

\usepackage[pagewise]{lineno}

\makeatletter
\renewcommand\footnoterule{%
  \kern -3pt
  \hrule width \columnwidth height 0.4pt
  \kern 2.6pt
}
\makeatother

\begin{document}

\RUNAUTHOR{Caldentey, Xie}
\RUNTITLE{Intertemporal Demand Allocation for Inventory Control}
\TITLE{\bf   Intertemporal Demand Allocation for \\  Inventory Control in Online Marketplaces\\}
\ARTICLEAUTHORS{%
\AUTHOR{{\large Ren\'e Caldentey ~~~and~~~ Tong Xie}\\}
\AFF{{\normalsize Booth School of Business, The University of Chicago, Chicago, IL 60637} \\
\EMAIL{rene.caldentey@chicagobooth.edu, txie0@chicagobooth.edu}}
}%

\date{}
\ABSTRACT{\vspace{0.5cm}

Online marketplaces increasingly do more than simply match buyers and sellers: they route orders across competing sellers and, in many categories, offer additional ancillary fulfillment services that make seller inventory a source of platform revenue. We investigate how a platform can use intertemporal demand allocation to influence sellers’ inventory choices without directly controlling stock. We develop a model in which the platform observes aggregate demand, allocates orders across sellers over time, and sellers choose between two fulfillment options, \emph{fulfill-by-merchant} (\textup{FBM}) and \emph{fulfill-by-platform} (\textup{FBP}), while replenishing inventory under state-dependent base-stock policies. The key mechanism we study is informational: by changing the predictability of each seller’s sales stream, the platform changes sellers’ safety-stock needs even when average demand shares remain unchanged. We focus on nondiscriminatory allocation policies that give sellers the same demand share and forecast risk. Within this class, uniform splitting minimizes forecast uncertainty, whereas any higher level of uncertainty can be implemented using simple low-memory allocation rules. Moreover, increasing uncertainty above the uniform benchmark requires routing rules that prevent sellers from inferring aggregate demand from their own sales histories. These results reduce the platform’s problem to choosing a level of forecast uncertainty that trades off adoption of platform fulfillment against the inventory held by adopters. Our analysis identifies demand allocation as a powerful operational and informational design lever in digital marketplaces.

\smallskip

}
\KEYWORDS{Inventory control; Dynamic demand allocation; Online marketplace operations; Time-series forecasting.\vspace{0.3cm}}

\maketitle

\parindent 0em
~

\vspace{-3cm}

\section{Introduction}\label{sec:introduction}

Digital marketplaces increasingly operate not only as matching intermediaries, but also as market makers that control how demand is allocated across third-party sellers over time. In many categories, multiple sellers offer identical or nearly identical products, and platform algorithms determine which offer is made salient, and hence which seller is most likely to receive the next order. For example, on Amazon, the ``buy box'' (``featured offer'') selects the default merchant for one-click purchasing, thereby routing marginal order flow among competing sellers \citep{LeeMusolff2025PlatformGuidedSearch,Raval2023SteeringOneClick,Markup2021BuyBox}.\footnote{This control is economically consequential because third-party sellers account for a large share of marketplace activity; in recent quarters, they represented more than 60\% of Amazon's worldwide paid units \citep{AmazonQ4Results2025}.}\smallskip

Motivated by this reality, the central question we investigate in this paper is: ``{\em How can a platform shape sellers' inventory decisions without directly controlling their stock levels?}'' Platforms often have several levers at their disposal to influence inventory choices, including restock limits, capacity reservations, minimum performance thresholds, and automated replenishment recommendations. In this paper, however, we focus on a different and more indirect mechanism that, to the best of our knowledge, has not been studied in the literature. The mechanism we exploit builds on the informational asymmetry between the platform and sellers: because sellers replenish inventory based on their own realized sales histories, platform control over demand allocation affects not only a seller's average demand, but also the information revealed through its order stream and, hence, the forecastability of its future demand. Because forecastability determines safety-stock requirements under standard inventory policies, it becomes a key operational lever through which the platform can influence decentralized inventory decisions.
\smallskip

This mechanism is especially salient when the platform also provides fulfillment services. In programs such as Fulfillment by Amazon (FBA) and Walmart Fulfillment Services (WFS), sellers retain ownership of inventory while placing stock inside the platform's logistics network, where the platform provides storage together with downstream fulfillment services such as pick--pack--ship, returns, and customer service \citep{AmazonAnnualReport2024,LaiLiuXiaoZhao2022FBA,LiShenLiaoCaiChen2024FulfillmentService,WalmartCorp2020WFS,WalmartMarketplaceGrowth2024}. We use \textup{FBP} as a generic label for such arrangements and \textup{FBM} for seller-managed fulfillment.\footnote{Comparable programs include Walmart Fulfillment Services (WFS) for Walmart Marketplace sellers \citep{WalmartCorp2020WFS,SupplyChainDive2020WFS,WalmartMarketplaceGrowth2024,WalmartWFSStrategies2024}, Mercado Libre's managed fulfillment offering \citep{Reuters2024MELIBrazilInvestment,Processexcellence2025MELIFull}, and Zalando Fulfillment Solutions (ZFS) in Europe \citep{ZalandoCMD2021PlatformPartners}.} 
These programs are often framed as service-quality upgrades for consumers, but they also create a second operational relationship between the platform and sellers: the platform becomes an \emph{inventory carrier} that can monetize storage and fulfillment activity \citep{ShiYuDong2021WarehousingRM,AmazonAnnualReport2024,amazon_fba_fees_guide}. This dimension has become economically important in its own right. For Amazon, ``third-party seller services'' net sales---which include commissions and related fulfillment and shipping fees---grew from \$117.7B in 2022 to \$156.1B in 2024 \citep{AmazonAnnualReport2024}. Similar growth in fulfillment services has been observed for Walmart and Alibaba\footnote{Walmart reports continued rapid marketplace growth and notes that more than 60\% of its top sellers are enrolled in WFS \citep{WalmartMarketplaceGrowth2024,WalmartWFSStrategies2024}. In Alibaba's ecosystem, Cainiao's segment revenue reached RMB~99.0B in fiscal year 2024, up 28\% from the previous year \citep{AlibabaFY2024Results2024}.}. Although inventory remains a decentralized seller decision, \textup{FBP} gives the platform a direct economic stake in inventory through storage and fulfillment fees.\footnote{Fulfillment services can account for a substantial share of the platform's effective take rate. For example, while Amazon's referral fees are around 15\%, FBA fees for storage, packing, and delivery often consume 20--35\% of seller revenue, making fulfillment one of the largest cost components for many third-party sellers \citep{kaziukenas2023amazon50cut}.} The platform therefore influences inventory not by dictating stock levels, but by shaping the realized demand process on which sellers base their replenishment decisions. \smallskip


We formalize this idea in a stylized model where aggregate demand is allocated across sellers who manage inventory using base-stock policies. Because safety stock depends on forecast errors, the platform’s allocation policy affects inventory by altering the predictability of seller-level demand, even when long-run demand shares remain unchanged. 
This mechanism generates an extensive–intensive trade-off. Increasing forecast uncertainty raises sellers’ safety stock and inventory levels but can reduce participation in the platform’s fulfillment program if higher inventory costs make it less attractive. The platform therefore trades off higher inventory among participating sellers against lower adoption.\smallskip

Our results characterize both the power and the limits of this lever. The main takeaway is that the platform can materially influence on-platform inventory, and thus potentially service levels and inventory-related revenues, using simple history-dependent allocation mechanisms. First, we characterize the levels of seller-level short-horizon forecast uncertainty (root MSFE) that are implementable when the platform acts in a {\em neutral} manner, meaning that sellers are treated symmetrically in a sense we make precise later, and subject to seller participation constraints.  Second, we show that any implementable increase relative to the uniform-split benchmark, under which demand is allocated evenly across sellers, can be generated by low-order moving-average modifications of the uniform allocation rule. Third, we identify a structural role for informational asymmetry in the form of \emph{non-invertibility}: to induce strictly higher safety stocks under neutrality, the platform must exploit its informational advantage in how aggregate demand innovations are revealed through demand allocation.\smallskip

Taken together, our results show that platform allocation is not only a demand-routing tool but also an inventory-design tool. A platform that cannot directly control seller stock levels can still shape on-platform inventory by shaping what sellers are able to forecast. \smallskip

\paragraph{Roadmap.} The rest of the paper is organized as follows. \cref{sec:literature} reviews the related literature and positions our work within it. \cref{sec:motivating-example} introduces the model, including the platform's demand-allocation problem, sellers' inventory decisions, the \textup{FBP}/\textup{FBM} choice, and the neutrality requirement. \cref{sec:preliminaries} develops the forecasting framework that links allocation dynamics to seller-level forecastability. \cref{sec:optpolicy} characterizes optimal neutral allocation policies and clarifies the role of non-invertibility. \cref{sec:IllustrativeExample} presents an illustrative numerical example. \cref{sec:extension} discusses extensions, and \cref{sec:conclusions} provides concluding remarks.

\section{Related Literature}\label{sec:literature}

This paper lies at the intersection of three research streams. The first studies how digital marketplaces operating
as two-sided platforms shape market outcomes by controlling the \emph{allocation of demand} (order flow) across competing
sellers, and how these allocation choices interact with seller operations---especially inventory---when the platform also
offers fulfillment and inventory-carrying services. The second concerns neutrality and non-discrimination constraints in
platform design: beyond operational objectives, platforms may face reputational, legal, and political pressures that limit
overt discrimination among comparable sellers. The third stream is methodological and concerns forecasting
stationary time series and how forecast uncertainty propagates into inventory decisions through base-stock policies,
safety-stock levels, and information-sharing (or information-design) frictions. In what follows, we review these streams in tandem and position our contribution relative to them.

\subsection{Demand allocation and inventory in online marketplaces}\label{subsec:lit_platforms}

A growing literature in operations and economics emphasizes that many digital marketplaces do not merely match buyers and sellers; they actively \emph{shape} market outcomes by controlling ranking and order-routing rules. When multiple sellers offer close substitutes, the platform effectively implements a \emph{demand-allocation mechanism} that shifts marginal orders and therefore affects competition, entry incentives, and surplus. Recent work formalizes how platform-guided search, algorithmic order routing, and default-merchant selection (e.g., Amazon's ``buy box'') reshape two-sided market outcomes, seller competition, and inventory incentives \citep{LeeMusolff2025PlatformGuidedSearch,CaoHu2024,HagiuJullien2011DivertSearch,Raval2023SteeringOneClick}.\smallskip

A complementary literature in operations management examines platforms that bundle marketplace services with fulfillment and inventory management. Amazon's FBA is the canonical example: sellers pre-position inventory in the platform's logistics network, and the platform handles pick-pack-ship and customer service. Researchers study why a platform would offer such services even when third-party sellers compete with its own retail arm \citep{LaiLiuXiaoZhao2022FBA}, how seller-side tying of fulfillment services intensifies price competition \citep{decorniere2025fulfilled}, and how fulfillment design interacts with fee structure, asymmetric information, and format choice \citep{LiShenLiaoCaiChen2024FulfillmentService,LiZhaoZhang2024DealershipMarketplaceFulfillment,sun2020fulfillment,sun2021cross}. A common theme is that fulfillment programs transform the platform's revenue model: inventory held within the platform ecosystem becomes a direct source of monetization through storage fees and fulfillment activity \citep{ShiYuDong2021WarehousingRM,LiShenLiaoCaiChen2024FulfillmentService}. 
Our contribution is complementary: rather than focusing on fee design or contracting, we study how a platform can influence sellers' endogenous inventory decisions through \emph{temporal engineering of demand allocation} under neutrality-style constraints.\smallskip


\subsection{Neutrality and Platform Accountability}\label{subsec:lit_neutrality}

Platforms face growing reputational, legal, and political pressures to treat comparable sellers even-handedly. Recent policy scrutiny of recommendation fairness \citep{Phillips2024HowFair}, regulatory discussions under the Digital Markets Act \citep{amazon_dma_2024}, and investigative work on platform-controlled defaults such as Amazon's ``buy box'' \citep{Markup2021BuyBox,Raval2023SteeringOneClick} all highlight these accountability concerns.\smallskip


In response, a growing operations literature models \emph{explicit} neutrality or fairness constraints that restrict the platform’s feasible
set of strategies. In the context of search neutrality and online retail platforms, recent work studies how imposing such
constraints changes equilibrium outcomes, platform incentives, and welfare \citep{ZouZhou2025SearchNeutrality}. 
Related
constraint-based approaches appear in fairness-aware assortment and recommendation design such as market-share balancing
\citep{housni2025fairness}, pairwise fairness \citep{fair}, and minimum exposure or share constraints \citep{Lu2024}.
Our modeling choice follows this
tradition: we impose \emph{neutrality} on demand allocation by requiring that comparable sellers be treated
symmetrically in expectation---each seller receives the same long-run mean demand share and faces the same one-step-ahead
forecast uncertainty. This notion captures a practical accountability constraint while still allowing the platform to retain
design flexibility through the \emph{temporal} structure of allocation.

\subsection{Forecasting stationary time series and its role in inventory and supply chains}\label{subsec:lit_forecasting}

Standard inventory theory suggests that sellers replenish inventory using base-stock policies, under which the order-up-to level decomposes into a forecast component plus a safety-stock component, thereby creating a direct link between \emph{inventory holdings} and \emph{forecast uncertainty}. A modern
synthesis of the forecasting--inventory interface is provided by \citet{GoltsosSyntetosGlockIoannou2022MindTheGap}, which
highlights both the historical separation between forecasting and inventory research and recent efforts to integrate them.
Our paper contributes from a platform-design perspective: the platform sits upstream of sellers’ forecasting
problems and can shape the time-series properties of the demand data sellers observe.\smallskip

A related literature studies how forecast-error estimation, autocorrelation, and estimation uncertainty affect safety-stock calculations for lead-time demand \citep{PrakTeunterSyntetos2017SafetyStocksForecastedDemand,BabaiDaiLiSyntetosWang2022LeadTimeVariance,BeutelMinner2012SafetyStockCausalForecasting}. These results reinforce a core premise behind our mechanism-design problem: inventory depends on the \emph{structure} of forecast uncertainty at the relevant horizons, which is shaped not only by the mean and variance of demand but also by its autocovariance structure.\smallskip

Finally, a long-standing operations literature studies the value of information and information sharing in supply chains, and time-series perspectives
integrate forecasting and inventory control across stages \citep{ozan_info_design_supply_chain,GavirneniKapuscinskiTayur1999ValueInformation,CachonFisher2000ValueSharedInformation}. More recently,
\citet{CaldenteyGiloniHurvichZhang2025InfoDesignSharing} develop an information-design view of inventory systems and show how
the structure of an ordering process and non-invertibility can determine the value of information. Our paper is closely
aligned with this perspective but shifts the setting to a marketplace with multiple symmetric sellers: the platform’s demand
allocation mechanism governs how aggregate innovations are revealed to sellers over time and therefore how forecastability
propagates into base-stock choices under neutrality-style constraints.\smallskip

\paragraph{\bf Positioning and contributions.}

Relative to the platform and fulfillment literatures, we study a distinct mechanism: dynamic demand allocation that engineers the \emph{forecastability} of seller-level demand, thereby shaping inventory and FBP adoption. Relative to the inventory-forecasting and information-sharing literatures, we \emph{endogenize} each seller's demand process through a platform design choice under neutrality constraints. Unlike other information-design mechanisms, such as the burgeoning literature on Bayesian persuasion \citep{KamenicaGentzkow2011,BergemannMorris2019}, the platform's instrument here is not a signal structure but an allocation rule constrained by the requirement that seller-level demands sum to market demand period by period. As a result, our proposed mechanism operates not through belief manipulation, but through the manipulation of forecastability. \smallskip

We view our contributions as having both a conceptual and a concrete component. At a conceptual level, we identify a novel operational mechanism through which a platform can influence sellers’ inventory-management decisions. By controlling how realized aggregate demand is allocated across sellers over time, the platform effectively shapes the demand histories---and the information they carry---that sellers observe and use for forecasting, thereby influencing both their fulfillment-mode choices (\textup{FBP} versus \textup{FBM}) and the inventory level. More concretely, we characterize the implementable range of seller-level one-step-ahead forecast uncertainty under neutral demand allocation, showing that uniform splitting attains a sharp lower bound, while simple low-order moving-average perturbations---with low memory and based only on recent demand realizations---can generate higher uncertainty; we show that achieving strict improvements beyond the uniform benchmark necessarily requires non-invertibility, thereby identifying a precise sense in which the platform’s informational advantage can matter even under symmetric treatment; and we reduce the platform’s design problem to a tractable one-dimensional optimization problem, which makes transparent the trade-off between fulfillment adoption and inventory level.\smallskip

\section{Model Setup}\label{sec:motivating-example}
\setcounter{footnote}{1}

This section introduces the model of a platform that allocates demand across sellers who manage inventory and choose between \textup{FBM} and \textup{FBP}, as schematically depicted in Figure~\ref{fig:System}.

\begin{figure}[htbp]
\centering
\includegraphics[width=0.7\textwidth]{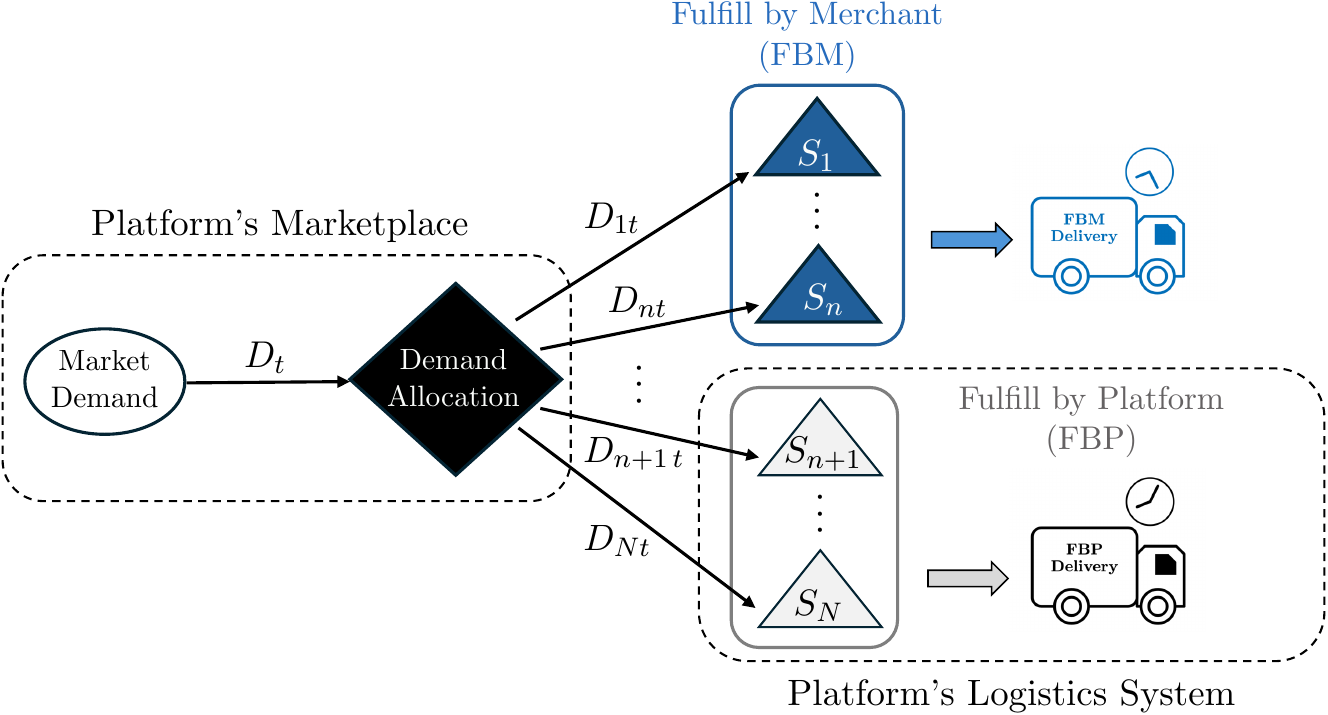}\smallskip 

\caption{Overview of the platform marketplace. Aggregate demand $D_t$ is allocated across sellers, who choose between FBM and FBP and manage their inventory stocks using state-dependent base-stock policies.}
\label{fig:System}
\end{figure}

The platform observes aggregate demand $D_t$ each period and allocates it across $N$ sellers, generating seller-level demand streams $\{D_{nt}\}_{n\in[N]}$ satisfying $\sum_{n} D_{nt}=D_t$. The allocation may depend on past demand. Sellers choose between self-fulfillment (\textup{FBM}) and platform fulfillment (\textup{FBP}). Each seller manages inventory using a base-stock policy with base-stock level $S_n$ based on its allocated demand process, so the platform’s allocation affects both inventory levels and fulfillment-mode choices. The remainder of this section formalizes these elements.\medskip


\noindent
$\diamond$  \textbf{\textsc{Platform intermediation and fulfillment modes}.}
We consider a platform that intermediates transactions between consumers and a pool of
$N$ third-party sellers, indexed by $n\in[N]:=\{1,\ldots,N\}$.
Time is discrete, $t=0,1,2,\ldots$. In each period $t$, aggregate market demand for a representative product
(or product class) is $D_t$. The platform acts as a \emph{market maker}: it determines how realized market
demand is split across sellers by choosing an allocation vector $(D_{nt})_{n\in[N]}$ satisfying
\begin{equation}\label{eq:allocation}
\sum_{n\in[N]} D_{nt} \;=\; D_t , \qquad t=0,1,2,\ldots .
\end{equation}
The allocation may be history-dependent and may reflect the platform's information at time $t$.\medskip

The platform charges a fixed intermediation fee $\rho$ per transaction (per unit of fulfilled demand),
independently of how the order is fulfilled.
Each seller chooses between two fulfillment modes: \emph{fulfill-by-merchant} (FBM) or \emph{fulfill-by-platform} (FBP). Under FBM, the seller manages fulfillment independently; under FBP, the platform provides storage and fulfillment services.\medskip

\noindent
$\diamond$ \textbf{\textsc{Seller heterogeneity and operating costs}.}
Sellers differ in their operating costs. If seller $n$ fulfills orders itself (\textup{FBM}), its per-unit, per-period costs are $(h_n,b_n,f_n)$, where $h_n$ is the holding cost, $b_n$ the stockout/backorder penalty, and $f_n$ the fulfillment cost. Under platform fulfillment (\textup{FBP}), the relevant triplet is $(H,b_n,F)$, where $H$ and $F$ are the effective holding and fulfillment costs under \textup{FBP}. 

Let ${\cal M}_n\in\{\textup{FBM},\textup{FBP}\}$ denote seller $n$’s mode and define
\[
\bar h_n :=
\begin{cases}
h_n, & {\cal M}_n=\textup{FBM},\\
H,   & {\cal M}_n=\textup{FBP},
\end{cases}
\qquad
\bar f_n :=
\begin{cases}
f_n, & {\cal M}_n=\textup{FBM},\\
F,   & {\cal M}_n=\textup{FBP}.
\end{cases}
\]
Thus, seller $n$ always faces backorder cost $b_n$, while holding and fulfillment costs depend on ${\cal M}_n$.

In what follows, we impose the following assumption to focus on instances that most clearly capture the canonical
trade-off between the \textup{FBM} and \textup{FBP} operating modes:

\begin{assum}\label{assm:nonnegative_deltas}
For each seller $n\in[N]$, $F\le f_n$ and $H\ge h_n$.
\end{assum}

That is, platform fulfillment weakly reduces per-order fulfillment cost but weakly increases per-unit inventory carrying cost. This canonical trade-off is empirically grounded: \textup{FBP} pools shipments across sellers and positions inventory in distributed warehouses, lowering per-order costs ($F\leq f_n$), but premium urban warehouse space and storage fees raise carrying costs ($H\geq h_n$) \citep{amazon_fba_fees_guide,houde_newberry_seim_density,LiShenLiaoCaiChen2024FulfillmentService}. The assumption is not essential for the theory or the analysis we develop, but it ensures that neither mode trivially dominates the other along both cost dimensions.

\medskip

\noindent
$\diamond$ \textbf{\textsc{Market demand}.} 
We focus on stationary demand to isolate the effects of demand allocation from other forces such as seasonality, life-cycle trends, or macroeconomic shocks. This assumption still allows for a rich class of demand processes, including ARMA models, while ensuring that the long-run statistical properties of demand remain stable. As a result, sellers face constant one-step-ahead forecast uncertainty and can therefore manage inventory using time-invariant base-stock levels.
\smallskip

\begin{assum}[Stationary Demand]\label{assm:Normal}
The demand process $\{D_t\}_{t\ge 0}$ is a weakly stationary, purely non-deterministic Gaussian process that admits the
$\mathrm{MA}(\infty)$\footnote{See \ref{app:MA_infty} for discussion on $\mathrm{MA}(\infty)$ representation and Gaussian demand.} representation
\begin{equation}\label{def:Dt}
D_t=\mu+\sum_{k=0}^\infty \psi_k\,\eps_{t-k},
\end{equation}
for some sequence $\psi=\{\psi_k\}_{k\ge 0}\in\ell^2$, where $\mu=\mathbb{E}[D_t]$ and $\{\eps_t\}$ is a Gaussian
white-noise sequence with $\Var(\eps_t)=\sigma_\eps^2$. We further assume that $\{D_t\}$ is invertible\footnote{A formal definition and discussion of invertibility is provided in Section \ref{sec:preliminaries}.} with respect to
$\{\eps_t\}$.
\end{assum}\smallskip

In what follows, we normalize $\sigma_\eps=1$. This is without loss of generality because any scale factor in the
innovations can be absorbed into $\psi$ (equivalently, we measure demand in units of the innovation standard
deviation). \smallskip

Letting ${\cal B}$ denote the {\sf Backshift operator}, i.e., ${\cal B}\eps_t=\eps_{t-1}$, the MA($\infty$) representation of the market demand in \eqref{def:Dt} can be written compactly as $D_t=\mu+\psi({\cal B})\,\eps_t$, where $\psi(z)=\sum_{k=0}^\infty \psi_k\, z^k$ is the $z$-transform representation of centered demand process $D_t-\mu$.\smallskip

\noindent
$\diamond$ \textbf{\textsc{Admissible demand allocation policies}.}
In each period $t$, the platform observes $D_t$ and chooses an allocation $(D_{nt})_{n\in[N]}$ satisfying \eqref{eq:allocation}, possibly via a history-dependent rule measurable with respect to the platform's information set $\mathcal{F}_t$, the filtration generated by $\{D_s\}_{s\le t}$. We model allocation policies as linear filters applied to aggregate demand, which preserves stationarity and enables tractable characterization of forecastability.

\begin{dfn}[Admissible Policy]\label{def:Admissible}
An admissible demand allocation policy\footnote{Under invertibility, representations based on $\{D_t\}$ and $\{\eps_t\}$ are equivalent for demand processes whose $z$-transform $\psi(z)$ satisfies mild technical conditions, see \ref{app:MA_infty} for details.} $\pi=\bigl(\mu_n,\C{T}_n(z):n\in[N]\bigr)$,
consists, for each seller $n\in[N]$, of a mean allocation $\mu_n$ and a transfer function $\C{T}_n(z)$, where \linebreak $\C{T}_n(z)=\sum_{k=0}^\infty \tau_{nk}z^k$ for some sequence $\{\tau_{nk}\}_{k\ge 0}\in\ell^2$. Under such a policy, seller $n$'s demand process is given by
\[
D_{nt}=\mu_n+\sum_{k=0}^{\infty}\tau_{nk}\,(D_{t-k}-\mu),\quad \mbox{or equivalently},\quad D_{nt}=\mu_n+\C{T}_n(\C{B})\,(D_t-\mu).
\]
The feasibility condition $D_t=\sum_{n\in[N]}D_{nt}$ then requires both
\[
\sum_{n\in[N]}\mu_n=\mu
\qquad \mbox{and} \qquad
\sum_{n\in[N]}\C{T}_n(z)=1.
\]

We let $\Pi(\psi)$ denote the class of admissible allocation policies for market demand $\psi$.\footnote{
One could also consider allocation rules that do not satisfy $D_t=\sum_{n\in[N]}D_{nt}$ period-by-period, which would require the platform to absorb the mismatch, for example by holding its own inventory. We focus on contemporaneous splits to isolate the role of order allocation across sellers.}
\end{dfn}\smallskip

One particular allocation that serves as a natural benchmark is the {\em uniform allocation}, under which
\begin{equation}\tag{Uniform Allocation}
D_{nt}=\frac{1}{N}\,D_t, \qquad \mbox{that is,}\qquad \mu_n=\frac{\mu}{N},\quad \mbox{and}\quad \C{T}_n(z)=\frac{1}{N}\quad \mbox{for all }n \in [N].
\end{equation}


\noindent
$\diamond$ \textbf{\textsc{Inventory management, seller payoffs, and participation}.}
Each seller $n$ manages inventory to satisfy its allocated demand stream $\{D_{nt}\}$. For expositional simplicity, our baseline model assumes zero replenishment lead time; in \ref{subsec:leadtime_extension} we extend the model to allow for arbitrary lead times that may vary across sellers and may also depend on the seller’s chosen fulfillment mode.\smallskip

Let $\mathcal{F}_{nt}$ denote seller $n$'s information set at time $t$, defined as the filtration generated by $\{D_{ns}\}_{s\le t}$. Just before demand in period $t+1$ is realized, seller $n$ chooses a state-dependent order-up-to level $S_{nt}$ measurable with respect to $\mathcal{F}_{nt}$ and places an order
\[
Q_{nt}=\bigl(S_{nt}-I_{nt}\bigr)^+,
\]
where $I_{nt}$ denotes net inventory after serving demand in period $t$ and $(x)^+=\max\{x,0\}$. 

Following the standard base-stock logic, seller $n$ selects $S_{nt}$ to minimize expected holding and backorder costs for period $t+1$ given $\mathcal{F}_{nt}$:
\[
S_{nt}
=\argmin_{S}\;
\e\!\left[
\,\bar h_n\, (S-D_{n,t+1})^+
+ b_n\, (D_{n,t+1}-S)^+
\;\middle|\; \mathcal{F}_{nt}
\right].
\]
Under \cref{assm:Normal}, $D_{n,t+1}$ is Gaussian conditional on $\mathcal{F}_{nt}$, and the optimal base-stock policy takes the familiar form
\begin{equation}\label{eq:basestock}
S_{nt}=m_{nt}+\zeta_n\,\sigma_n,
\qquad
\zeta_n:=\Phi^{-1}\!\left(\frac{b_n}{\bar h_n+b_n}\right),
\end{equation}
where $m_{nt}:=\e[D_{n,t+1}\mid \mathcal{F}_{nt}]$ is the one-step-ahead mean forecast and the {\em one-step-ahead MSFE} is given by
\begin{equation}\label{eq:MSFE}
\sigma_n^2:=\Var\!\bigl[D_{n,t+1}-m_{nt}\mid \mathcal{F}_{nt}\bigr].
\end{equation}
Thus, safety stock is proportional to the one-step-ahead MSFE, so inventory depends directly on forecastability. Our  model assumes that sellers act as \emph{optimal forecasters}. This assumption allows us to isolate the impact of the platform's  allocation policy from suboptimal forecasting heuristics on the sellers' part. In \cref{sec:suboptimal-forecasts}, we discuss the case in which sellers use suboptimal forecasting methods.
\smallskip


Under an admissible allocation policy $\pi=\bigl(\mu_n,\C{T}_n(z):n\in[N]\bigr)$, seller $n$ receives a stationary demand stream with mean $\mu_n$ and one-step-ahead root MSFE $\sigma_n^\pi$, as defined in \eqref{eq:MSFE}. An explicit representation of $\sigma_n^\pi$ in terms of $\psi(z)$ and $\C{T}_n(z)$ is provided in \cref{sec:preliminaries}. Substituting the optimal base-stock level into the seller’s expected single-period profit yields
\begin{equation}\label{eq:sellerprofit_modes}
U_n^m(\pi)
=
(r-\rho-f_n^m)\,\mu_n
-
K_n^m\,\sigma_n^\pi,
\qquad m\in\{\rm FBM,FBP\},
\end{equation}
where $(f_n^{\ti{FBM}},f_n^{\ti{FBP}})=(f_n,F)$ and $K_n^m$ represents the seller's expected inventory cost coefficient. It aggregates the expected holding and backorder costs associated with the optimal safety-stock level $\zeta_n^m$, and is given by 
\[
K_n^m
=
h_n^m\,\zeta_n^m
+
(h_n^m+b_n)\,\mathcal{L}(\zeta_n^m),
\qquad
\zeta_n^m
=
\Phi^{-1}\!\left(\frac{b_n}{h_n^m+b_n}\right),
\]
with $(h_n^{\ti{FBM}},h_n^{\ti{FBP}})=(h_n,H)$ and $\mathcal{L}(z)=\phi(z)-z(1-\Phi(z))$ denoting the standard Normal loss function.  Stationarity implies that $\sigma_n^\pi$ is independent of $t$, so $U_n^m(\pi)$ is time-independent as well.\footnote{Equation \eqref{eq:sellerprofit_modes} follows from substituting the optimal base-stock level into the expected single-period holding/backorder cost under Gaussian demand.} Therefore, seller $n$ chooses \textup{FBP} over \textup{FBM} whenever
\begin{equation}\label{eq:seller_adoption_condition}
U_n^{\ti{FBP}}(\pi)\ge U_n^{\ti{FBM}}(\pi).
\end{equation}
\vspace{-0.1cm}

\noindent
$\diamond$ \textbf{\textsc{Neutral (non-discriminatory) allocation policies}.}
When a platform steers order flow, its allocation algorithm may be scrutinized for favoritism toward particular sellers or the platform’s own retail arm \citep{Markup2021BuyBox,Phillips2024HowFair,amazon_dma_2024}. To capture a baseline notion of non-discrimination, we impose a neutrality constraint requiring symmetric treatment of sellers in expectation: each seller receives the same mean demand share and faces the same one-step-ahead forecast uncertainty (MSFE). This is the notion of non-discrimination most relevant for inventory decisions. We measure the demand risk induced by policy $\pi$ using the one-step-ahead root MSFE $\sigma_n^\pi$, which governs sellers’ inventory and safety-stock decisions. This leads to the following definition.\smallskip


\begin{dfn}[Neutral admissible policies]\label{def:NeutralAdmissible}
Let $\pi=\bigl(\mu_n, \C{T}_n(z):n\in[N]\bigr)$ be an admissible demand-allocation policy. We say that $\pi$ is \emph{neutral} if there exists a constant $\sigma^\pi>0$ such that, for all $n\in[N]$, $\mu_n=\frac{\mu}{N}$, and $\sigma_n^\pi=\sigma^\pi$, where $\sigma_n^\pi$ denotes seller $n$'s one-step-ahead root MSFE under policy $\pi$ as defined in \eqref{eq:MSFE}. We denote by $\Pi_{\ti N}(\psi)$ the class of neutral admissible policies.
\end{dfn}\smallskip

Since every neutral policy assigns the same mean demand $\mu_n=\frac{\mu}{N}$ to each seller, we will abuse notation slightly and write a neutral admissible policy simply as $\pi=\bigl(\C{T}_n(z):n\in[N]\bigr)$, suppressing the dependence on the common mean demand allocation.\smallskip


Under neutrality, 
seller $n$ prefers \textup{FBP} if and only if
\begin{equation}\label{eq:adoption_condition}
\Delta F_n \;\ge\; \frac{N\,\sigma^{\pi}}{\mu}\,\Delta K_n,
\qquad
\text{where}\qquad
\Delta F_n:=f_n-F,
\qquad
\Delta K_n:=K_n^{\ti{FBP}}-K_n^{\ti{FBM}}.
\end{equation}
Under Assumption~\ref{assm:nonnegative_deltas}, we restrict attention to instances in which $\Delta K_n\ge 0$ and $\Delta F_n\ge 0$ for all $n\in[N]$. Thus, higher forecast uncertainty raises inventory costs and reduces FBP adoption. Here, $\Delta F_n$ is the seller’s per-order fulfillment cost savings from using FBP, while $\Delta K_n$ is the corresponding increase in the inventory holding cost. \smallskip

Under a neutral policy, the boundary
$\Delta F=\frac{N\,\sigma^{\pi}}{\mu}\,\Delta K$
separates FBP adopters from FBM sellers; see Figure~\ref{fig:FBP_partition_pos_orthant}. As $\sigma^{\pi}$ increases, this boundary rotates upward, reducing adoption while increasing safety stock for remaining FBP sellers. Hence the platform faces an extensive--intensive trade-off: higher $\sigma^{\pi}$ lowers participation but raises inventory among adopters. Because both margins affect profits, neither extreme is generally optimal.

\begin{figure}[h]
\centering
\begin{tikzpicture}[scale=0.9]

  \pgfmathsetmacro{\xmin}{0}
  \pgfmathsetmacro{\xmax}{6}
  \pgfmathsetmacro{\ymin}{0}
  \pgfmathsetmacro{\ymax}{5}

  \pgfmathsetmacro{\slope}{0.70}

  \begin{scope}
    \clip (\xmin,\ymin) rectangle (\xmax,\ymax);
    \fill[gray!15]
      (\xmin,{\slope*\xmin}) --
      (\xmin,\ymax) --
      (\xmax,\ymax) --
      (\xmax,{\slope*\xmax}) -- cycle;
  \end{scope}

  \draw[->] (\xmin,0) -- (\xmax,0) node[below right] {$\Delta K$};
  \draw[->] (0,\ymin) -- (0,\ymax) node[above left] {$\Delta F$};
  \node[below left] at (0,0) {\scriptsize $0$};

  \draw[thick] (\xmin,{\slope*\xmin}) -- (\xmax,{\slope*\xmax})
    node[pos=0.86, above, sloped, yshift=0pt]
    {\small $\Delta F=\frac{N\,\sigma^{\pi}}{\mu}\,\Delta K$};

  \node at (1.05,4.55) {\large \textup{FBP}};
  \node[fill=white, inner sep=2pt] at (5.45,0.35) {\large \textup{FBM}};

  \tikzset{seller/.style={circle, fill=black, inner sep=1.25pt}}

  \node[seller] at (0.45,0.35) {};
  \node[seller] at (0.65,1.05) {};
  \node[seller] at (0.85,0.70) {};
  \node[seller] at (1.05,1.55) {};
  \node[seller] at (1.25,0.45) {};
  \node[seller] at (1.65,1.10) {};
  \node[seller] at (1.85,2.05) {};
  \node[seller] at (2.05,0.60) {};
  \node[seller] at (2.25,2.95) {};
  \node[seller] at (2.45,1.35) {};
  \node[seller] at (2.65,2.55) {};
  \node[seller] at (2.85,0.95) {};
  \node[seller] at (3.05,3.40) {};
  \node[seller] at (3.25,1.85) {};
  \node[seller] at (3.45,2.20) {};
  \node[seller] at (3.35,3.05) {};
  \node[seller] at (4.05,1.15) {};
  \node[seller] at (4.25,4.35) {};
  \node[seller] at (4.45,2.55) {};
  \node[seller] at (4.65,1.70) {};
  \node[seller] at (5.05,2.80) {};
  \node[seller] at (5.25,0.85) {};
  \node[seller] at (5.45,3.25) {};
  \node[seller] at (5.65,1.35) {};
  \node[seller] at (5.80,2.15) {};

  \node[seller] at (1.20,0.88) {};  
  \node[seller] at (2.40,1.58) {};  
  \node[seller] at (3.20,2.30) {};  
  \node[seller] at (4.10,2.65) {};  

  \node[seller] at (0.35,3.65) {};
  \node[seller] at (2.55,4.10) {};
  \node[seller] at (1.80,3.35) {};
  \node[seller] at (0.60,2.85) {};
  \node[seller] at (0.95,3.75) {};
  \node[seller] at (0.95,3.75) {};

\end{tikzpicture}
\caption{Partition of sellers (dots) in the $(\Delta K,\Delta F)$ plane under a neutral policy $\pi$.
Sellers above $\Delta F=\frac{N\sigma^{\pi}}{\mu}\Delta K$ adopt \textup{FBP} (shaded region), while sellers below prefer \textup{FBM}.}
\label{fig:FBP_partition_pos_orthant}
\end{figure}
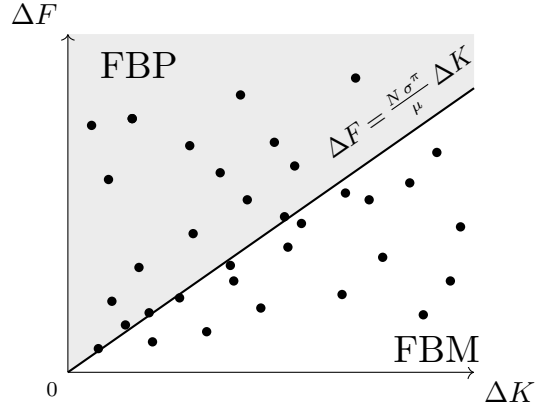


 \medskip

\setcounter{footnote}{1}
\noindent
$\diamond$ \textbf{\textsc{Platform payoff and objective}.}
Under an allocation policy $\pi$, the platform collects a fixed intermediation fee $\rho$ for each fulfilled unit.
Because total transaction volume equals aggregate demand,\footnote{In our baseline model demand is backordered rather than
lost, so every unit of demand eventually becomes a transaction. Allowing for lost sales would make total transactions
(and therefore intermediation revenue) depend on service level.}
the platform’s per-period  expected intermediation revenue equals
$ \mathbb{E}[\rho\,D_t]=\rho\,\mu$, which is independent of $\pi$.\smallskip

In addition, the platform monetizes participation in its \textup{FBP} service. Specifically, we assume that for each unit transacted by a
seller using FBP, the platform collects a fulfillment payoff (net of delivery cost) $\Delta f$, and for each unit of inventory
stored in the platform’s warehouse network it collects a rental fee $\Delta h$ per unit per period.
Thus, to compute the platform’s total payoff under $\pi$, we must first determine which sellers choose FBP vs. FBM. As illustrated in \cref{fig:FBP_partition_pos_orthant}, the set of sellers who adopt the FBP mode is given by 
\[N^{\ti{FBP}}(\pi)
:=\left\{\, n\in[N] : \Delta F_n \geq \frac{N\,\sigma^{\pi}}{\mu}\,\Delta K_n \right\}.\]

Let $\bar S_n^\pi:=\e[S_{nt}]$ denote the average inventory level held by seller $n$ under policy $\pi$. By
\eqref{eq:basestock}, $\bar S_n^\pi=\mu_n+\sigma^{\pi}\,\zeta^{\ti{FBP}}_n$, and the platform’s expected
\emph{gross} per-period payoff can be written as
\begin{eqnarray}\nonumber
V^{\pi}
&:=&
\rho\,\mu + \sum_{n\in N^{\ti{FBP}}(\pi)} \big[\Delta f\, \mu_n + \Delta h\,\bar S_n^\pi\big] =
\rho\,\mu + \sum_{n\in N^{\ti{FBP}}(\pi)} \left[\Delta f\,\frac{\mu}{N} + \Delta h\left(\frac{\mu}{N}+\sigma^{\pi}\,\zeta^{\ti{FBP}}_n\right)\right] \nonumber\\ \label{eq:platform-profit}
&=&
\rho\,\mu + (\Delta f+\Delta h)\,\frac{\mu}{N}\,|N^{\ti{FBP}}(\pi)| + \Delta h\,\Gamma^{\ti{FBP}}(\pi).
\end{eqnarray}

where $|N^{\ti{FBP}}(\pi)|$ denotes the number of sellers adopting FBP and $\Gamma^{\ti{FBP}}(\pi)$ their aggregate safety stock under policy~$\pi$.  Equation \eqref{eq:platform-profit} shows that the platform’s payoff
depends on $\pi$ through (i) which sellers adopt FBP and
(ii) the inventory levels induced by the demand allocation rule.\medskip

Putting all the pieces together, the platform’s optimization problem is given by
\begin{align}\label{eq:platform_prob}
\mbox{Platform's Problem:}\qquad
&\sup_{\pi \in \Pi_{\ti{N}}(\psi)} V^{\pi}
= \rho\,\mu
+ (\Delta f+\Delta h)\,{\,\mu \over N}\, |N^{\ti{FBP}}(\pi)|
+ \Delta h\, \Gamma^{\ti{FBP}}(\pi) \\ \nonumber
\mbox{where}\quad
& N^{\ti{FBP}}(\pi)
=\left\{\, n\in[N] : \Delta F_n \geq \frac{N\,\sigma^{\pi}}{\mu}\,\Delta K_n \right\}
\quad \mbox{and} \quad
\Gamma^{\ti{FBP}}(\pi)=\sigma^{\pi}\,\sum_{n\in N^{\ti{FBP}}(\pi)} \zeta^{\ti{FBP}}_n.
\end{align}
This problem lacks a simple closed-form solution because varying $\sigma^{\pi}$ rotates the adoption boundary, causing discrete changes in the adopter set. Thus, both  $|N^{\ti{FBP}}(\pi)|$ and  $\Gamma^{\ti{FBP}}(\pi)$ are non-smooth in $\sigma^{\pi}$, see \cref{fig:FBP_sigma_payoff}.

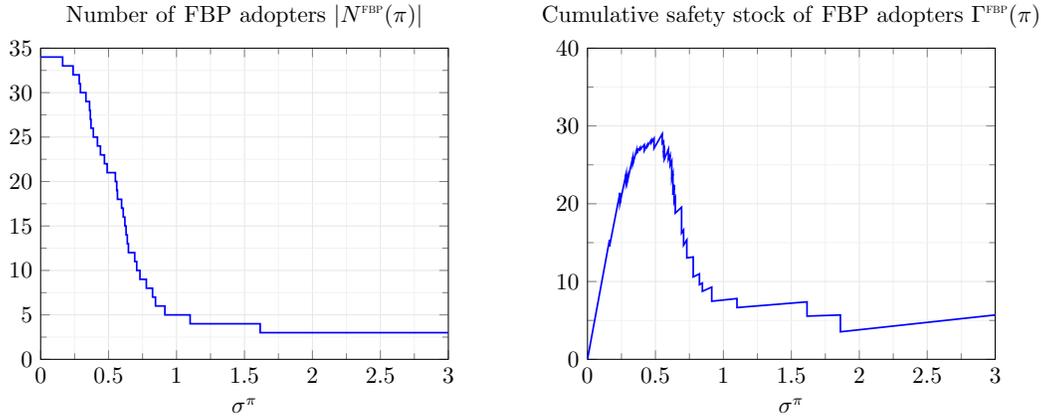
\begin{figure}[h]
\centering
\begin{tikzpicture}[scale=0.8]

\begin{axis}[
  name=plotA,
  width=0.47\textwidth,
  height=0.38\textwidth,
  xmin=0, xmax=3,
  ymin=0, ymax=35,
  xlabel={$ \sigma^{\pi}$},
  title={Number of \textup{FBP} adopters $|N^{\ti{FBP}}(\pi)|$},
  title style={align=center},
  ytick={0,5,10,15,20,25,30,35},
  grid=both,
  major grid style={draw=gray!20},
  minor grid style={draw=gray!10},
  minor tick num=1,
]
\addplot+[const plot mark right, thick, mark=none]
coordinates {
  (0.0000,34)
  (0.1619,34)
  (0.1619,33)
  (0.2389,33)
  (0.2389,32)
  (0.2840,32)
  (0.2840,31)
  (0.2927,31)
  (0.2927,30)
  (0.3333,30)
  (0.3333,29)
  (0.3600,29)
  (0.3600,28)
  (0.3656,28)
  (0.3656,27)
  (0.3707,27)
  (0.3707,26)
  (0.3878,26)
  (0.3878,25)
  (0.4179,25)
  (0.4179,24)
  (0.4400,24)
  (0.4400,23)
  (0.4694,23)
  (0.4694,22)
  (0.4894,22)
  (0.4894,21)
  (0.5500,21)
  (0.5500,20)
  (0.5600,20)
  (0.5600,19)
  (0.5652,19)
  (0.5652,18)
  (0.5957,18)
  (0.5957,17)
  (0.6078,17)
  (0.6078,16)
  (0.6200,16)
  (0.6200,15)
  (0.6286,15)
  (0.6286,14)
  (0.6377,14)
  (0.6377,13)
  (0.6460,13)
  (0.6460,12)
  (0.6923,12)
  (0.6923,11)
  (0.7071,11)
  (0.7071,10)
  (0.7308,10)
  (0.7308,9)
  (0.7778,9)
  (0.7778,8)
  (0.8235,8)
  (0.8235,7)
  (0.8452,7)
  (0.8452,6)
  (0.9143,6)
  (0.9143,5)
  (1.1000,5)
  (1.1000,4)
  (1.6154,4)
  (1.6154,3)
  (1.8611,3)
  (1.8611,3)
  (3.0000,3)
};
\end{axis}

\begin{axis}[
  at={(plotA.outer east)},
  anchor=outer west,
  xshift=1.2cm,
  width=0.47\textwidth,
  height=0.38\textwidth,
  xmin=0, xmax=3,
  ymin=0, ymax=40,
  xlabel={$ \sigma^{\pi}$},
  title={Cumulative safety stock of \textup{FBP} adopters $\Gamma^{\ti{FBP}}(\pi)$},
  title style={align=center},
  grid=both,
  major grid style={draw=gray!20},
  minor grid style={draw=gray!10},
  minor tick num=1,
]
\addplot+[thick, mark=none]
coordinates {
  (0.0000,0.0000)
  (0.1619,15.3567)
  (0.1619,14.5067)
  (0.2389,21.4088)
  (0.2389,20.0588)
  (0.2840,23.8377)
  (0.2840,22.6877)
  (0.2927,23.3854)
  (0.2927,22.7854)
  (0.3333,25.9500)
  (0.3333,25.1000)
  (0.3600,26.8200)
  (0.3600,26.3700)
  (0.3656,26.8881)
  (0.3656,26.5681)
  (0.3707,27.1259)
  (0.3707,26.8209)
  (0.3878,27.2432)
  (0.3878,26.9532)
  (0.4179,27.5816)
  (0.4179,26.8280)
  (0.4400,27.5880)
  (0.4400,27.3460)
  (0.4694,28.1933)
  (0.4694,27.8347)
  (0.4894,28.3961)
  (0.4894,27.1136)
  (0.5500,28.9740)
  (0.5500,27.7675)
  (0.5600,28.0000)
  (0.5600,26.5720)
  (0.5652,26.8962)
  (0.5652,25.7045)
  (0.5957,27.0077)
  (0.5957,25.6977)
  (0.6078,26.2422)
  (0.6078,24.9444)
  (0.6200,25.4820)
  (0.6200,23.1940)
  (0.6286,23.8186)
  (0.6286,21.1243)
  (0.6377,22.4337)
  (0.6377,20.2648)
  (0.6460,20.7076)
  (0.6460,18.7976)
  (0.6923,19.5523)
  (0.6923,16.2331)
  (0.7071,16.6279)
  (0.7071,14.6800)
  (0.7308,15.3454)
  (0.7308,13.0285)
  (0.7778,13.1600)
  (0.7778,10.5889)
  (0.8235,10.9853)
  (0.8235,9.6000)
  (0.8452,9.7970)
  (0.8452,8.7417)
  (0.9143,9.2857)
  (0.9143,7.4686)
  (1.1000,7.8100)
  (1.1000,6.6550)
  (1.6154,7.3846)
  (1.6154,5.5538)
  (1.8611,5.7000)
  (1.8611,3.5370)
  (3.0000,5.7000)
};
\end{axis}

\end{tikzpicture}
\caption{Left: number of \textup{FBP} adopters $|N^{\ti{FBP}}(\pi)|$. Right: cumulative safety stock of \textup{FBP} adopters, $\Gamma^{\ti{FBP}}(\pi)$.}
\label{fig:FBP_sigma_payoff}
\end{figure}
Nevertheless, the problem can be solved efficiently by exploiting the step-function nature of  $|N^{\ti{FBP}}(\pi)|$. Under Assumption~\ref{assm:nonnegative_deltas}, seller $n$ adopts \textup{FBP} under a neutral policy if and only if $\sigma^{\pi}\le \mu\,\Delta F_n/(N\,\Delta K_n)$\footnote{Note that under this convention, the right-hand side is $+\infty$ when$\Delta K_n=0$}. Hence, the number of adopters $|N^{\ti{FBP}}(\pi)|$ is a piecewise-constant, non-increasing
function of $\sigma^{\pi}$, with possible jumps only at the finitely many threshold values
\[
\left\{\; \frac{\mu\,\Delta F_n}{N\,\Delta K_n}\;:\; n\in[N],\ \Delta K_n>0 \right\}.
\]
Between two consecutive thresholds, the adopter set is fixed, so
$\Gamma^{\ti{FBP}}(\pi)$ is linear in $\sigma^{\pi}$.
As a result, the platform’s payoff is piecewise linear in $\sigma^{\pi}$, and an optimal solution can be obtained by
evaluating $V^{\pi}$ at the finitely many critical thresholds. This procedure, however, requires
knowing the range of feasible values of $\sigma^{\pi}$ that the platform can implement under a neutral
allocation policy. 
Although an arbitrarily large $\sigma^\pi$ is theoretically feasible, we impose a participation constraint requiring nonnegative seller payoffs, which yields an upper bound $\sigma_{\UB}$.\medskip

Based on \eqref{eq:sellerprofit_modes}, define
$U_n(\sigma):=\max\{U_n^{\ti{FBP}}(\sigma),\,U_n^{\ti{FBM}}(\sigma)\}$, and let
\begin{equation}\label{eq:sigmaUB}
\sigma_{\UB}:=\sup\bigl\{\sigma \geq 0 \colon U_n(\sigma) \geq 0 \ \text{for all } n \in [N]\bigr\}.
\end{equation}

On the other hand, in \cref{sec:bounds} we show that there exists a lower bound $\sigma_{\LB}$ on the set
of values of $\sigma^{\pi}$ that are attainable under neutral demand-allocation policies. This bound is
derived after introducing the mathematical framework that links an allocation policy $\pi$ to the sellers’
root MSFE $\sigma^{\pi}$ in \cref{sec:preliminaries}.\medskip


\setcounter{footnote}{1}
$\diamond$ \textbf{\textsc{Discussion of Neutrality}.} We close this section by clarifying two central and interconnected modeling choices: our use of \emph{neutral} demand-allocation policies (see \cref{def:NeutralAdmissible}) and our use of forecastability (MSFE), rather than an alternative measure such as unconditional variance, as the relevant notion of demand risk. \smallskip

At a high level, the neutrality assumption reflects growing expectations of non-discrimination in platforms that control order flow \citep{Markup2021BuyBox,Phillips2024HowFair,amazon_dma_2024}. In our formulation, neutrality requires symmetry \emph{in expectation} along two dimensions: mean exposure (equal long-run demand shares) and demand risk (equal root MSFE). Root MSFE is the natural measure here because sellers' safety-stock decisions depend on short-horizon predictability, and it can in principle be audited from realized order histories. This notion implicitly assumes sellers form optimal one-step-ahead forecasts; in \cref{sec:suboptimal-forecasts} we examine how the analysis changes under suboptimal forecasting methods.
A natural alternative is to measure demand risk, and hence neutrality, in terms of unconditional variance rather than MSFE. The next example shows that these two measures are not necessarily aligned.

\begin{exm}[Forecastability and Unconditional Variance]\label{exm:msfe_vs_variance} {\sf  Consider a system with two sellers. 
\;  \\
\textbf{Case 1: Same MSFE, Different Variances.}
Let demand follow the MA(1) process $D_t=\mu+\eps_t+0.8\eps_{t-1}$. Suppose the platform allocates demand so that
\[
D_{1t}=\frac{\mu}{2}+0.5\eps_t-0.2\eps_{t-1}-0.48\eps_{t-2},\qquad
D_{2t}=\frac{\mu}{2}+0.5\eps_t+1.0\eps_{t-1}+0.48\eps_{t-2}.
\]
Then $\Var(D_{1t})=0.5204$ and $\Var(D_{2t})=1.4804$, but the two allocations have the same one-step-ahead forecast uncertainty, with $\mathrm{MSFE}_1=\mathrm{MSFE}_2=0.36$. Thus, if sellers 1 and 2 forecast optimally, they should be indifferent between these two allocations, whereas na\"ive sellers who forecast demand as if it were i.i.d.\ would rank allocation 1 as less risky, and hence preferable, than allocation 2.

\textbf{Case 2: Same Variance, Different MSFEs.}
Now consider aggregate demand $D_t=\mu+\eps_t+0.5\eps_{t-1}$ and
\[
D_{1t}=\frac{\mu}{2}+0.2\eps_t+0.85\eps_{t-1},\qquad
D_{2t}=\frac{\mu}{2}+0.8\eps_t-0.35\eps_{t-1}.
\]
In this case, the two allocations have the same variance, $\Var(D_{1t})=\Var(D_{2t})=0.7625,
$ but different forecast errors:
$
\mathrm{MSFE}_1=0.7225,\,
\mathrm{MSFE}_2=0.64.
$
In this case, optimal forecasters prefer allocation 2 over allocation 1, whereas na\"ive forecasters would be indifferent. Consequently, a neutrality criterion based solely on unconditional variance may fail to detect economically relevant disparities in the actual inventory risk borne by sellers.} \scalebox{0.7}{$\blacksquare$}\smallskip

\begin{figure}[htpb]
    \centering
    \begin{minipage}{0.48\textwidth}
        \centering
        \resizebox{\linewidth}{!}{
\begin{tikzpicture}[scale=0.8]

\begin{axis}[
    width=12cm,
    title style={text width=11cm, align=center},
    height=5.5cm,
    title={\textbf{Seller 1 Demand:} $ \mathrm{Var} = 0.5204 \mid \mathrm{MSFE} = 0.36$},
    ymin=-3.50, ymax=3.50,
    xmin=0, xmax=49,
    grid=major,
    grid style={dashed, gray!30},
    ylabel={$D_{1t}$},
    xticklabels={\empty},
    legend pos=south east
]
\addplot[color=blue, mark=*, mark size=1pt, thick] coordinates {
    (0,0.248) (1,-0.168) (2,0.113) (3,0.698) (4,-0.733) (5,-0.801) (6,0.949) (7,0.180) (8,-1.146) (9,-0.003) (10,-0.115) (11,-0.401) (12,0.437) (13,-0.781) (14,-0.596) (15,0.982) (16,0.434) (17,0.630) (18,-0.031) (19,-0.675) (20,1.451) (21,0.272) (22,-0.625) (23,-0.618) (24,-0.020) (25,0.848) (26,-0.336) (27,0.365) (28,0.177) (29,-0.206) (30,0.046) (31,1.186) (32,-0.088) (33,-1.415) (34,0.629) (35,-0.267) (36,-0.046) (37,-0.436) (38,-0.372) (39,1.305) (40,0.967) (41,-0.157) (42,-0.447) (43,-0.210) (44,-0.624) (45,0.080) (46,0.623) (47,0.966) (48,0.181) (49,-1.458)
};
\addplot[color=black, dashed, thick] coordinates {(0,0) (49,0)};
\end{axis}

\begin{axis}[
    width=12cm,
    title style={text width=11cm, align=center},
    height=5.5cm,
    yshift=-4.8cm,
    title={\textbf{Seller 2 Demand:} $ \mathrm{Var} = 1.4804 \mid \mathrm{MSFE} = 0.36$},
    ymin=-3.50, ymax=3.50,
    xmin=0, xmax=49,
    grid=major,
    grid style={dashed, gray!30},
    xlabel={Time ($t$)},
    ylabel={$D_{2t}$}
]
\addplot[color=red, mark=x, mark size=2pt, thick] coordinates {
    (0,0.248) (1,0.428) (2,0.424) (3,1.343) (4,1.717) (5,0.380) (6,0.443) (7,1.851) (8,1.291) (9,0.170) (10,0.086) (11,-0.436) (12,-0.567) (13,-0.938) (14,-2.660) (15,-2.924) (16,-1.897) (17,-1.126) (18,-0.626) (19,-1.463) (20,-1.115) (21,0.675) (22,0.511) (23,-0.753) (24,-1.665) (25,-1.173) (26,-0.726) (27,-0.910) (28,-0.477) (29,-0.566) (30,-0.881) (31,0.184) (32,1.557) (33,0.347) (34,-0.653) (35,-0.296) (36,-0.722) (37,-1.357) (38,-2.524) (39,-2.170) (40,-0.071) (41,0.919) (42,0.468) (43,-0.184) (44,-1.096) (45,-1.983) (46,-1.660) (47,-0.278) (48,1.008) (49,-0.030)
};
\addplot[color=black, dashed, thick] coordinates {(0,0) (49,0)};
\end{axis}

\end{tikzpicture}


        }
    \end{minipage}\hfill
    \begin{minipage}{0.48\textwidth}
        \centering
        \resizebox{\linewidth}{!}{
\begin{tikzpicture}

\begin{axis}[
    width=12cm,
    title style={text width=11cm, align=center}, 
    height=5.5cm,
    title={\textbf{Seller 1 Demand:} $\mathrm{Var} = 0.7625 \mid \mathrm{MSFE} = 0.7225$},
    ymin=-3.5, ymax=3.5, 
    xmin=0, xmax=49,
    grid=major, 
    grid style={dashed, gray!30},
    ylabel={$D_{1t}$},
    xticklabels={\empty}, 
    legend pos=south east
]
\addplot[color=blue, mark=*, mark size=1pt, thick] coordinates {
    (0,0.395) (1,0.012) (2,0.855) (3,1.248) (4,-0.246) (5,0.117) (6,1.496) (7,0.558) (8,-0.291) (9,0.368) (10,-0.487) (11,-0.347) (12,-0.177) (13,-1.971) (14,-1.579) (15,-0.681) (16,-0.798) (17,0.086) (18,-1.054) (19,-0.907) (20,1.201) (21,-0.178) (22,-0.228) (23,-1.320) (24,-0.441) (25,-0.136) (26,-0.903) (27,0.199) (28,-0.569) (29,-0.368) (30,-0.141) (31,1.572) (32,-0.223) (33,-0.735) (34,0.455) (35,-0.996) (36,-0.214) (37,-1.931) (38,-1.090) (39,0.315) (40,0.662) (41,0.123) (42,-0.159) (43,-0.552) (44,-1.401) (45,-0.704) (46,-0.180) (47,0.967) (48,-0.061) (49,-1.434)
};
\addplot[color=black, dashed, thick] coordinates {(0,0) (49,0)};
\end{axis}

\begin{axis}[
    width=12cm,
    title style={text width=11cm, align=center}, 
    height=5.5cm,
    yshift=-4.8cm, 
    title={\textbf{Seller 2 Demand:} $\mathrm{Var} = 0.7625 \mid \mathrm{MSFE} = 0.64$},
    ymin=-3.5, ymax=3.5, 
    xmin=0, xmax=49,
    grid=major, 
    grid style={dashed, gray!30},
    xlabel={Time ($t$)}, 
    ylabel={$D_{2t}$}
]
\addplot[color=red, mark=x, mark size=2pt, thick] coordinates {
    (0,-0.284) (1,0.567) (2,0.992) (3,-0.720) (4,-0.105) (5,1.345) (6,0.061) (7,-0.644) (8,0.598) (9,-0.561) (10,-0.210) (11,0.357) (12,-1.615) (13,-0.710) (14,0.154) (15,-0.613) (16,0.606) (17,-0.836) (18,-0.812) (19,1.667) (20,-0.694) (21,0.133) (22,-1.163) (23,0.063) (24,0.279) (25,-0.960) (26,0.703) (27,-0.612) (28,-0.023) (29,-0.379) (30,1.692) (31,-0.659) (32,-0.841) (33,1.028) (34,-1.265) (35,0.594) (36,-1.641) (37,-0.377) (38,0.622) (39,0.522) (40,-0.121) (41,-0.152) (42,-0.200) (43,-1.077) (44,-0.058) (45,-0.117) (46,1.007) (47,-0.095) (48,-1.531) (49,0.876)
};
\addplot[color=black, dashed, thick] coordinates {(0,0) (49,0)};
\end{axis}

\end{tikzpicture}


        }
    \end{minipage}
   \caption{Demand sample paths: \textit{Left:} equal MSFE, different variances; \textit{Right:} equal variances, different MSFE.}
    \label{fig:msfe_vs_variance}
\end{figure}
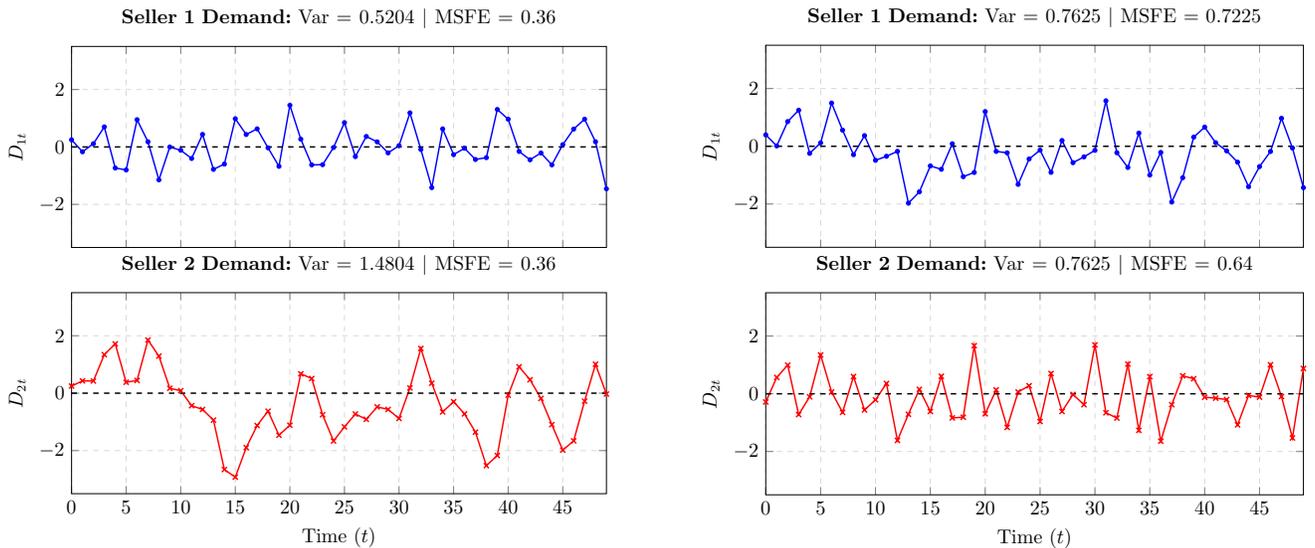
\end{exm}


\section{Preliminaries}\label{sec:preliminaries}
The platform’s allocation rule reshapes the demand signal observed by each seller and, therefore, its forecastability. To analyze this mechanism, we work in the $z$-domain rather than the time domain and use inner--outer factorization to separate two key features of demand: forecastability, governed by the outer factor, and phase structure, captured by the inner factor. This representation allows us to characterize the range of forecast uncertainty that can be implemented under neutrality.


\subsection{Foundations: Inner--outer factorization and invertibility}
\label{subsec:hardy_tools}
We briefly review the essential mathematical tools that underlie our analysis. These tools are mainly drawn from the theory of the $\mathbb{H}^2$ Hardy space, with particular emphasis on inner--outer factorization and its time-series interpretation. The interested reader can consult \ref{app:hardy_tools} for a more detailed presentation, including additional technical background and references.

\smallskip
Our analysis is facilitated by representing stationary demand processes in the $z$-domain. In particular, under an admissible demand allocation policy $\pi=\bigl(\mu_n,\C{T}_n(z):n\in[N]\bigr)$, seller $n$'s demand process satisfies
\[
D_{nt}=\mu_n+\C{T}_n(\C{B})\,(D_t-\mu).
\]
Moreover, by \cref{assm:Normal}, the centered market demand can be represented by $D_t-\mu=\psi(\C{B})\,\eps_t$, where $\psi(z)=\sum_{k=0}^\infty \psi_k z^k$ is the $z$-transform of $\{D_t-\mu\}$. Thus, seller $n$'s demand can equivalently be written as
\[
D_{nt}=\mu_n+\psi_n(\C{B})\,\eps_t,
\qquad \mbox{where}\qquad
\psi_n(z)=\C{T}_n(z)\psi(z).
\]

Thus, under policy $\pi$, the forecastability of seller $n$'s demand---and in particular its root MSFE $\sigma_n^\pi$---is determined by the transfer function $\psi_n(z)$. Note also that the feasibility condition $D_t=\sum_{n\in[N]}D_{nt}$ implies $\sum_{n\in[N]}\psi_n(z)=\psi(z)$.\smallskip

A key tool is the inner–outer factorization of $\psi_n(z)$.
A function $\mathcal{I}\in\mathbb{H}^2$ is \textbf{inner} if it is unimodular on the unit circle, i.e., $|\mathcal{I}(z)|=1$ a.e.\ on $|z|=1$. A function $\mathcal{O}\in\mathbb{H}^2$ is \textbf{outer} if it has no zeros in the open unit disk, i.e., $\mathcal{O}(z)\neq 0$ for $|z|<1$. We write $\mathbb{I}$ and $\mathbb{O}$ for the classes of inner and outer functions in $\mathbb{H}^2$.\smallskip

\begin{lem}[Inner--outer factorization]\label{lem:inner-outer}
Every $f(z)\in\mathbb{H}^2$ admits a factorization
\[
f(z)=\mathcal{O}(z)\,\mathcal{I}(z),
\]
where $\mathcal{I}\in\mathbb{I}$ is \emph{inner} and $\mathcal{O}\in\mathbb{O}$ is \emph{outer}.
\end{lem}\smallskip

Inner and outer functions admit a natural forecasting interpretation. Inner factors act as \emph{all-pass} filters: multiplying a transfer function by an inner factor preserves its spectral density and, hence, its entire autocovariance structure. Thus, these inner factors are ``invisible'' to an optimal forecaster. By contrast, outer factors are causal, minimum-phase, and invertible, and therefore determine the fundamental one-step-ahead forecastability of the process. Thus, by \cref{lem:inner-outer}, any transfer function can be decomposed into an inner component, which captures phase distortions, and an outer component, which governs forecast uncertainty. The following lemma summarizes the implication most relevant for our analysis.\smallskip

\begin{lem}[Root MSFE]\label{lem:msfe_outer_main}
Let $D_{nt}=\mu_n+\C{T}_n(\C{B})\,(D_t-\mu)$ denote seller $n$'s allocated demand process under an admissible allocation policy $\pi=\bigl(\mu_n,\C{T}_n(z):n\in[N]\bigr)$, and define $\psi_n(z):=\C{T}_n(z)\,\psi(z)$. Then $\psi_n\in\mathbb{H}^2$ and, by \cref{lem:inner-outer}, admits an inner--outer factorization of the form
\[
\psi_n(z)=\C{O}_n(z)\,\C{I}_n(z),
\]
for some $\C{O}_n\in\mathbb{O}$ and $\C{I}_n\in\mathbb{I}$. Moreover, seller $n$'s one-step-ahead root mean squared forecast error under policy $\pi$ is given by
\[
\sigma_n^{\pi}=|\C{O}_n(0)|.
\]
\end{lem}

This lemma identifies the channel through which the platform can influence seller $n$'s root MSFE, and hence its inventory decisions, by shaping the outer component of the allocated demand process. By contrast, inner factors do not affect forecastability, but they provide the platform with additional flexibility to satisfy the feasibility condition $\sum_{n\in[N]}\psi_n(z)=\psi(z)$.\smallskip

Another important concept in our analysis is \emph{invertibility}. Given seller $n$'s allocated demand process $D_{nt}=\mu_n+\psi_n(\C{B})\,\eps_t$, we ask whether seller $n$ can recover the underlying demand shocks $\{\eps_t\}$ using only its own observed demand stream $\{D_{nt}\}$. If so, we say that $\{D_{nt}\}$ is invertible with respect to $\{\eps_t\}$. The next result gives a necessary and sufficient condition.\smallskip

\begin{lem}[Invertibility]\label{lem:invertible_outer_main}
Seller $n$'s demand process $\{D_{nt}\}$ is invertible with respect to the demand shocks $\{\eps_t\}$ if and only if $\psi_n(z)$ is outer.
\end{lem}\smallskip

As we show later, the platform's optimal allocation policy generally induces a nontrivial inner factor in $\psi_n(z)$, so seller-level demand is not invertible. In that case, seller $n$ cannot fully recover the underlying demand shocks and therefore has incomplete information about aggregate market demand.\smallskip

We conclude with a simple example that illustrates the main ideas in this section.\smallskip

\begin{exm}[MA$(q)$ demand allocation]
{\sf Suppose the $z$-transform $\psi_n(z)$ of seller $n$'s demand admits an MA$(q)$ representation
$$\psi_n(z)=\sum_{k=0}^q \varphi_k\,z^k.$$
By fundamental theorem of algebra, we can factor $\psi_n(z)$ in terms of its $q$ (possibly complex) roots $\{a_j\}_{j=1}^q$ as
\[
\psi_n(z)=c\,\prod_{j=1}^q (z-a_j).
\]
To obtain the inner--outer factorization of $\psi_n(z)$, split the roots according to whether they lie inside the unit disk or not. Let $\C{A}_{\mathbb{I}}=\{a_j:\abs{a_j}<1\}$ and $\C{A}_{\mathbb{O}}=\{a_j:\abs{a_j}\ge 1\}$. It follows that
\begin{align*}
\psi_n(z)
&= c\,\prod_{j \in \C{A}_{\mathbb{O}}} (z-a_j)\,\prod_{j \in \C{A}_{\mathbb{I}}} (z-a_j) = \underbrace{c\,\prod_{j \in \C{A}_{\mathbb{O}}} (z-a_j)\,\prod_{j \in \C{A}_{\mathbb{I}}} (1-\bar{a}_j z)}_{\C{O}_n(z)}
\underbrace{\prod_{j \in \C{A}_{\mathbb{I}}} \frac{z-a_j}{1-\bar{a}_j z}}_{\C{I}_n(z)}.
\end{align*}

The roots of $\C{O}_n(z)$ are given by the set $\C{A}_{\mathbb{O}}\cup\{\bar{a}_j^{-1}:a_j\in\C{A}_{\mathbb{I}}\}$, all of which lie outside the unit disk; hence $\C{O}_n(z)$ is outer. On the other hand, by construction, one can check that $\abs{\C{I}_n(z)}=1$ for all $\abs{z}=1$, so $\C{I}_n(z)$ is inner; alternatively, it is called a Blaschke product. Therefore, $\psi_n(z)$ admits the inner--outer factorization $\psi_n(z)=\C{O}_n(z)\,\C{I}_n(z)$ given above.\smallskip

It then follows from \cref{lem:msfe_outer_main} that the root MSFE of $\psi_n(z)$ is
$$\sigma_n=|\C{O}_n(0)|=|c|\,\prod_{j \in \C{A}_{\mathbb{O}}} \abs{a_j}.$$
Finally, by \cref{lem:invertible_outer_main}, seller $n$'s demand is invertible with respect to the demand shocks if and only if $\C{A}_{\mathbb{I}}=\emptyset$, that is, if and only if $\abs{a_j}\ge 1$ for all $j=1,\dots,q$.} \scalebox{0.7}{$\blacksquare$}
\end{exm}

\subsection{Demand allocation in the $z$-domain}\label{subsec:z_domain_reformulation}


We now apply the inner-outer representation to the platform's demand-allocation problem in  \eqref{eq:platform_prob}. Let $\psi(z)=\sum_{k=0}^\infty \psi_k z^{k}$ denote the $z$-transform of the market demand process. Because \eqref{def:Dt} is the Wold representation of market demand $D_t$ (see \ref{app:MA_infty} for details), the corresponding transfer function $\psi$ is outer, i.e., $\psi\in\mathbb{O}$.
In the $z$-transform domain, an admissible allocation policy $\pi$ is specified by a collection $\{\psi_n(z)=\C{T}_n(z)\, \psi(z)\}_{n\in[N]}$ of transfer functions
representing the seller-level allocated demand streams, subject to the admissibility constraint
\[
\sum_{n\in[N]}\psi_n(z)=\psi(z).
\]
Equivalently, writing the inner--outer factorization $\psi_n(z)=\mathcal{O}_n(z)\,\mathcal{I}_n(z)$ with $\mathcal{I}_n\in\mathbb{I}$ and $\mathcal{O}_n\in\mathbb{O}$,
we may represent an admissible policy as $\pi=\{(\mathcal{I}_n,\mathcal{O}_n)\}_{n\in[N]}$ satisfying $\sum_{n\in[N]}\mathcal{O}_n(z)\,\mathcal{I}_n(z)=\psi(z)$.
\smallskip

\cref{lem:msfe_outer_main} makes explicit the link between outer factors and sellers' forecastability. In particular, seller $n$'s one-step-ahead root MSFE is given by $\sigma_n=|\mathcal{O}_n(0)|$. We use this relation to express the notion of a \emph{neutral} allocation policy as the algebraic condition $|\mathcal{O}_n(0)|=\sigma$ for all $n\in[N]$, where $\sigma$ is a platform-chosen constant independent of $n$. This allows us to rewrite the platform's optimization problem \eqref{eq:platform_prob} directly in terms of inner and outer functions:
\begin{align}\label{eq:platform-z}
& V^*=\sup_{\pi \in \Pi_{\ti{N}}(\psi)} \rho\,\mu
+ (\Delta f+\Delta h)\,\frac{\mu}{N}\, |N^{\ti{FBP}}(\sigma)|
+ \Delta h\, \Gamma^{\ti{FBP}}(\sigma) \\ \nonumber
\text{subject to} \quad
& \sum_{n \in [N]} \psi_n(z)=\psi(z), \\ \nonumber
& \psi_n(z)=\mathcal{O}_n(z)\,\mathcal{I}_n(z),
\quad \mathcal{O}_n \in \mathbb{O},\ \mathcal{I}_n \in \mathbb{I}, \quad \forall n \in [N], \\ \nonumber
& |\mathcal{O}_n(0)|=\sigma, \quad \forall n \in [N], \\ \nonumber
\mbox{where} \quad & N^{\ti{FBP}}(\sigma)
=\left\{\, n\in[N] : \Delta F_n \geq \frac{N\,\sigma}{\mu}\,\Delta K_n \right\}
\quad \mbox{and} \quad
\Gamma^{\ti{FBP}}(\sigma)=\sigma\,\sum_{n\in N^{\ti{FBP}}(\sigma)} \zeta^{\ti{FBP}}_n.
\end{align}

As previously noted, we can solve \eqref{eq:platform-z} by conducting an exhaustive search over the finitely many values of $\sigma$ at which the piecewise-constant
function $|N^{\ti{FBP}}(\sigma)|$ jumps. This search, however, must be carried out over the range of values of $\sigma$ that are feasible under a neutral allocation policy
$\pi\in\Pi_{\ti{N}}(\psi)$. We investigate this implementability question next by deriving a lower bound on the feasible set $\{\sigma^{\pi}:\pi\in\Pi_{\ti{N}}(\psi)\}$.

\subsection{Uniform allocation as a lower bound on root MSFE}\label{sec:bounds}

In this section we derive a lower bound on the minimum possible value of $\sigma^{\pi}$ under any neutral demand allocation policy. As we will see, the bound is attained by the simple uniform allocation in which each seller receives a fixed fraction of market demand each period. Intuitively, uniform splitting minimizes volatility by removing unnecessary dispersion across sellers: when total demand $D_t$ is divided equally, each seller receives the same deterministic share $1/N$, so only aggregate uncertainty remains. Any deviation from this equal split reintroduces dispersion and increases the mean squared forecast error (MSFE).
\medskip

We first state a lower-bound result for neutral allocation policies.

\begin{prop}[Lower bound]\label{prop:lowerbound}
Let market demand $D_t$ be given as in \eqref{def:Dt}, and consider an admissible neutral allocation policy $\pi=\{(\mathcal{I}_n,\mathcal{O}_n)\}_{n\in[N]}$.
Let $\sigma^{\pi}_n$ denote the root MSFE of seller $n$ under $\pi$. Then
\[
\sum_{n \in [N]} \sigma^{\pi}_n \;\ge\;  \abs{\psi(0)}.
\]
It follows that for any neutral demand allocation policy $\pi$, $\sigma^{\pi}_n \geq \frac{\abs{\psi(0)}}{N}$.
\end{prop}\smallskip

For future reference, let $\sigma_{\LB}:=\frac{\abs{\psi_0}}{N}$ denote the lower bound on a seller's root MSFE that the platform can implement using an admissible policy.
The following corollary shows that this lower bound is attained by uniform splitting.\smallskip

\begin{cor}[Uniform allocation]\label{cor:uniform}
Let $\pi$ be the uniform demand allocation policy under which each seller $n$ faces a demand process with $z$-transform $\psi_{n}(z)=\psi(z)/N$.
Then $\sigma_n^{\pi}=\sigma_{\LB}$ for all $n\in[N]$.
\end{cor}\smallskip

The lower bound in \cref{prop:lowerbound} has an immediate implication for the platform’s ability to induce FBP adoption. Recall from \cref{eq:adoption_condition} that under
a neutral allocation policy $\pi$, seller $n$ adopts FBP if and only if
\[
\Delta F_n \ge \frac{N\,\sigma^{\pi}}{\mu}\,\Delta K_n.
\]
By \cref{prop:lowerbound}, any neutral policy must satisfy $\sigma^{\pi}\ge \sigma_{\LB}$.
As a result, the platform cannot necessarily induce \emph{every} seller to join the FBP system. In particular, sellers with
\[
\Delta F_n < \frac{N\,\sigma_{\LB}}{\mu}\,\Delta K_n
= \frac{\abs{\psi(0)}}{\mu}\,\Delta K_n
\]
strictly prefer FBM under \emph{every} implementable neutral policy $\pi\in\Pi_{\ti{N}}(\psi)$.
In words, the existence of the lower bound $\sigma_{\LB}$ implies that some sellers may never find it profitable to adopt the platform’s fulfillment service, regardless of how
the platform allocates demand while maintaining neutrality.\smallskip

The lower-bound result above reveals that uniform allocation provides the \emph{best} forecasting environment the platform can create under neutrality, since it minimizes each seller’s root MSFE and therefore serves as a natural benchmark. A complementary question is whether the platform can move in the opposite direction and deliberately make seller-level demand \emph{less} forecastable. As noted above, doing so may be attractive because higher forecast uncertainty raises safety stocks and, through that channel, affects inventory held inside the platform's logistics system. The next result shows, however, that increasing seller-level root MSFE relative to the uniform benchmark is not costless. In particular, any deviation from uniform allocation that preserves neutrality must also increase the unconditional variance of the allocated demand, thereby imposing a \emph{volatility tax} on sellers.\smallskip

\begin{prop}\label{prop:cost}
Suppose market demand $D_t$ is given as in \eqref{def:Dt}. Let $\Var^{\ti{unif}} = \frac{1}{N^2}\norm{\psi}^2$ denote the variance of each seller's demand under the uniform demand allocation policy.
If an admissible policy $\pi \in \Pi(\psi)$ satisfies $\Var(D_{nt}) = \Var^{\ti{unif}}$ for all $n \in [N]$, then $\pi$ must be the uniform demand allocation policy, meaning $\psi_{n} = \frac{1}{N}\psi$ for all $n \in [N]$.
\end{prop}
This proposition shows that the platform cannot raise sellers' forecast uncertainty above the uniform-allocation benchmark while keeping their unconditional demand variance unchanged. Any non-uniform allocation therefore comes with strictly higher overall demand volatility for the sellers.\medskip

\section{Optimal Neutral Demand Allocation Policies}\label{sec:optpolicy}

Building on the inner--outer factorization of  demand allocation introduced in the previous section, we now solve \eqref{eq:platform-z} over the class $\Pi_{\ti N}(\psi)$ of neutral demand-allocation policies, as defined in \cref{def:NeutralAdmissible}, and characterize the structure and implementability of optimal solutions.\smallskip

Our analysis shows that, absent participation constraints, the platform can induce any target root MSFE $\sigma^{\pi}\ge \sigma_{\LB}$ using remarkably simple low-order moving-average perturbations of the uniform split: when the number of sellers $N$ is even, an alternating $\mathrm{MA}(1)$ filter suffices, whereas when $N$ is odd, a minimal $\mathrm{MA}(2)$ modification applied to two sellers restores admissibility while preserving neutrality. We derive these constructions through an inner--outer factorization of seller transfer functions, which clarifies how introducing a nontrivial inner component reallocates phase across sellers and raises each seller's forecast error by changing the outer factor at the origin. A key implication is that achieving $\sigma^{\pi}>\sigma_{\LB}$ under neutrality requires non-invertibility for at least one seller, so that no seller can recover the innovation sequence $\{\eps_t\}$, and hence aggregate demand, from its own demand stream alone. Finally, we translate the resulting transfer-function policies into explicit time-domain allocation rules and complement the ex post benchmark with an order-level routing mechanism that implements the allocation dynamically as orders arrive, tracking the benchmark up to a small rounding error.

\smallskip

Our construction of optimal solutions proceeds in three steps. First, for expositional convenience, we reparametrize the platform's demand allocation to the class of transfer functions  that express each seller’s allocation relative to the uniform baseline introduced in \cref{cor:uniform}. \smallskip

\begin{dfn}[Transfer Function Representation]
\label{def:allocation}
Using the uniform allocation $\psi(z)/N$ as a baseline, we represent the $z$-transform for seller $n$ as
\begin{equation}\label{eq:psi_n}
\psi_n(z)=\frac{\psi(z)}{N}\,\mathcal{T}_n(z),\qquad \text{with }\;\sum_{n=1}^N \psi_n(z)=\psi(z).
\end{equation}
\end{dfn}\smallskip

Because ${\psi(z)}/{N}$ represents a uniform allocation, $\mathcal{T}_n(z)$ captures deviations from this baseline while preserving admissibility. Second, we decompose $\mathcal{T}_n$ via its inner--outer factorization $\mathcal{T}_n(z)=\mathcal{O}_n(z)\,\mathcal{I}_n(z)$, where $\mathcal{O}_n\in\mathbb{O}$ and $\mathcal{I}_n$ is a finite Blaschke product whose order is $1$ when $N$ is even, and at most $2$ when $N$ is odd. This decomposition isolates variance scaling (the outer factor) from phase reallocation (the inner factor). Third, for a given target value $\sigma\geq \sigma_{\LB}$ of a seller’s root MSFE, we choose the zeros of the Blaschke factors so that (i) neutrality holds and (ii) each seller’s target root MSFE binds with equality. This pins down the constants $\mathcal{O}_n(0)$ and, using $\sigma_n=\abs{\psi(0)}\,\abs{\mathcal{O}_n(0)}/N$, attains the desired target $\sigma$. Our derivation also implies that the optimization problem in \eqref{eq:platform-z} admits multiple solutions, allowing the platform to impose additional criteria beyond neutrality to break ties. We discuss this non-uniqueness in \cref{sec:non-unique}. \smallskip

Before turning to the formal details of this program, we first illustrate the approach with an example.

\begin{exm}[IID Demand]\label{exm:iid}{\sf 
Suppose there are two sellers, $N=2$, and market demand is i.i.d., given by $D_t=\mu+\psi_0\,\eps_t$, for some scalar $\psi_0>0$. In this case, $\psi(z)=\psi_0$ and $\sigma_{\LB}=\psi_0/2$. Consider the following class of admissible allocation policies based on $\mathrm{MA}(1)$ processes:
\begin{equation}\label{eq:exiid}
\psi_{n}(z)=\frac{\psi_0}{2}\,\C{T}_n(z), \quad \mbox{where }\C{T}_n(z)=1-\alpha_n\,z, \qquad n=1,2,
\end{equation}
for scalars $\alpha_1$ and $\alpha_2$ such that $\alpha_1=-\alpha_2$. Here, $\psi_0(z)/2$ is the uniform allocation. This policy is admissible since $\psi_{1}(z)+\psi_{2}(z)=\psi(z)$. 
\smallskip

Seller $n$'s mean squared forecast error depends on whether its demand process is invertible (see \citealp{BrockwellDavis2006} for details). Invertibility is determined by the location of the root $z_n$ of the transfer function $\psi_n(z)$, which here equals $z_n:=1/\alpha_n$. We distinguish two cases:
\begin{enumerate}
    \item If $\abs{z_n}>1$, i.e., $\abs{\alpha_n}<1$, then $\{D_{nt}\}$ is invertible with $\psi_n(z)\in\mathbb{O}$ and $\sigma_n=\abs{\psi_n(0)}=\psi_0/2=\sigma_{\LB}$.
    \item If $\abs{z_n}\le 1$, i.e., $\abs{\alpha_n}\ge 1$, then $\{D_{nt}\}$ is noninvertible. In this case, $\psi_n(z)$ admits the inner--outer factorization
    \[
    \psi_n(z)=\frac{\psi_0}{2}\,(1-\alpha_n z)=\C{O}_n(z)\,\C{I}_n(z), \quad \mbox{with}\quad 
    \C{O}_n(z)=\frac{\psi_0}{2}\,(\alpha_n-z) \quad \mbox{and}\quad 
    \C{I}_n(z)=\frac{1-\alpha_n z}{\alpha_n-z},
    \]
    where $\C{O}_n(z)\in\mathbb{O}$ and $\C{I}_n(z)\in\mathbb{I}$ is a Blaschke factor. Hence, \cref{lem:msfe_outer_main} implies
    \[
    \sigma_n=\abs{\C{O}_n(0)}=\frac{\psi_0}{2}\, \abs{\alpha_n}=\sigma_{\LB}\, \abs{\alpha_n}.
    \]
\end{enumerate}

Therefore, combining cases 1 and 2, under the $\mathrm{MA}(1)$ allocation in \eqref{eq:exiid}, each seller’s root MSFE equals $\sigma_n=\sigma_{\LB}\,\max\{1,\abs{\alpha_n}\}$. The platform can thus implement any target root MSFE $\sigma\ge \sigma_{\LB}$ by choosing $\alpha_1$ and $\alpha_2$ so that $\abs{\alpha_n}\,\sigma_{\LB}=\sigma$.
\smallskip

Finally, note that when $\sigma>\sigma_{\LB}$ under the $\mathrm{MA}(1)$ allocation in \eqref{eq:exiid}, each seller’s demand process $D_{nt}$ is non-invertible in the demand shocks $\{\eps_t\}$. In the i.i.d.\ demand case, this implies that no single seller can infer aggregate market demand $D_t$ solely from observing its own demand process $D_{nt}$. By contrast, under the uniform allocation in \cref{cor:uniform}, each seller can recover $D_t$ exactly from $D_{nt}$, provided the number of sellers $N$ operating on the platform is common knowledge.
}  \scalebox{0.7}{$\blacksquare$}
\end{exm}
\smallskip

As we show in the following result, the method from the previous example extends to a general demand process when the number of sellers is even. The case of an odd number of sellers requires a slightly different approach, which we discuss next.
\smallskip

\begin{prop}\label{prop:weaklyeven}
Consider a target root MSFE $\sigma > \sigma_{\LB}$ and suppose the number of sellers $N$ is even. Consider a policy $\pi\in \Pi_{\ti{N}}(\psi)$ such that seller $n$'s demand allocation admits the $z$-transform representation \\ $\psi_n(z)=\frac{1}{N}\,\psi(z)\,\C{T}_n(z)$ with transfer function
\[
\C{T}_n(z)=1+(-1)^n\,\ba\, z,\qquad \mbox{where}\quad \ba=\frac{N\,\sigma}{\abs{\psi(0)}}={\sigma \over \sigma_{\LB}}.
\]
Then $\pi$ is neutral and satisfies $\sigma^\pi=\sigma$.
\end{prop}
\smallskip

It is worth noting that the allocation for seller $n$ in \cref{prop:weaklyeven} is obtained by taking the uniform split $\psi(z)/N$ and multiplying it by the deliberately non-invertible transfer $\C{T}_n(z)=1-\alpha_n z$. This introduces an inner factor—specifically, a Blaschke product—into the representation of seller $n$’s demand, which in turn modifies the outer component and increases the root MSFE relative to the uniform split. The value of $\alpha_n$ is carefully selected so as to match the target root MSFE $\sigma$.
\smallskip

When $N$ is odd, the previous construction, that is, multiplying the uniform allocation by the monomial factor $\bigl(1-\alpha_n z\bigr)$, cannot satisfy admissibility ($\sum_{n}\alpha_n=0$) and neutrality (the quantities $\abs{\alpha_n}$ are equal to the same positive constant for every $n$). Consequently, for odd $N$ the optimal design requires a different (e.g., higher-order or mixed) inner–outer modification rather than the simple monomial factor. \smallskip


\begin{prop}\label{prop:weaklyodd}
Consider a target root MSFE $\sigma > \sigma_{\LB}$ and suppose the number of sellers $N \geq 3$ is odd.  Let $\pi\in\Pi_{\ti N}(\psi)$ be a demand allocation policy such that the $z$-transform for seller $n$'s demand  admits the representation $\psi_n(z)=\frac{1}{N}\,\psi(z)\,\C{T}_n(z)$, with transfer function 
\begin{align*}
\C{T}_1(z)&=1+\ba z+\ba\,z^2,  \\
\C{T}_2(z)&=1-\ba z^2,  \\
\C{T}_n(z) &=1+(-1)^n\,\ba\, z, \quad \mbox{for } n \geq 3,
\end{align*}
where $\ba =\frac{N\,\sigma}{\abs{\psi(0)}}$. Then $\pi$ is neutral and implements $\sigma^\pi=\sigma$.
\end{prop}
\smallskip

A distinctive feature of the optimal allocation policies identified in Propositions~\ref{prop:weaklyeven}
and~\ref{prop:weaklyodd} is that the transfer functions $\C{T}_n(z)$ used to perturb the uniform allocation
possess a nontrivial inner component. As a result, the induced demand process $\{D_{nt}\}$ faced by any seller
$n$ is non-invertible with respect to the demand-shock sequence $\{\epsilon_t\}$. Indeed, as the following
proposition shows, non-invertibility is typically necessary for a neutral allocation policy to achieve
$\sigma^{\pi}>\sigma_{\LB}$.
\smallskip

\begin{prop}\label{prop:noninv}
Suppose the platform uses a  neutral policy $\pi \in \Pi_{\ti N}$ 
whose associated root MSFE $\sigma^{\pi}$ is strictly larger than the lower bound $\sigma_{\LB}$. Then at least one of the demand processes $\{D_{nt}\}$ must be non-invertible with respect to $\{\epsilon_t\}$.
\end{prop}\smallskip

\cref{prop:noninv} establishes that, under the neutrality requirement, the platform can increase the system’s safety stock beyond the normal level achieved under uniform allocation only by adopting allocation policies that render the demand process non-invertible for some sellers.

\subsection{Implementable order-level routing}
\label{sec:implementable-routing}

The transfer-function policies characterized in Propositions~\ref{prop:weaklyeven} or \ref{prop:weaklyodd} describe seller-level demand processes in the $z$-transform domain. To connect these prescriptions to marketplace operations, this subsection (i) translates the optimal policies into explicit time-domain allocation rules, and (ii) provides a simple order-by-order routing mechanism that implements the resulting benchmark allocations as orders arrive. The benchmark should be interpreted as an idealized target that is computed \emph{ex post} from aggregate demand, while the routing mechanism is \emph{online} and uses only information available at the time each order arrives.\smallskip

\begin{cor}\label{cor:timedomain}
Consider the demand allocation with $z$-transform representation
$\psi_n(z)=\frac{1}{N}\,\psi(z)\,\mathcal{T}_n(z)$, where the transfer function $\mathcal{T}_n(z)$
is given in Propositions~\ref{prop:weaklyeven} or \ref{prop:weaklyodd} depending on whether $N$
is even or odd. Then, the time-domain representation of seller $n$'s demand $D_{nt}$ is given by
\begin{enumerate}[\rm(a)]
\item For $N$ even:
\[
D_{nt}=\frac{D_t}{N}+(-1)^n\,\frac{\sigma}{N\,\sigma_{\LB}}\,(D_{t-1}-\mu),\qquad n\in[N].
\]

\item For $N$ odd:
\begin{align*}
D_{1t}&=\frac{D_t}{N}+\frac{\sigma}{N\,\sigma_{\LB}}\Big[(D_{t-1}-\mu)+(D_{t-2}-\mu)\Big],\\
D_{2t}&=\frac{D_t}{N}-\frac{\sigma}{N\,\sigma_{\LB}}\,(D_{t-2}-\mu),\\
D_{nt}&=\frac{D_t}{N}+(-1)^n\,\frac{\sigma}{N\,\sigma_{\LB}}\,(D_{t-1}-\mu),\qquad n\ge 3.
\end{align*}
\end{enumerate}
\end{cor}\smallskip

The time-domain representations in Corollary~\ref{cor:timedomain} can be interpreted as an \emph{ex post} allocation rule:
at the end of period $t$ the platform observes the realized aggregate demand $D_t$ and then assigns seller-level
demands $\{D_{nt}\}_{n\in[N]}$ according to a simple linear function of $(D_t,D_{t-1},D_{t-2})$ (and $\mu$). Under
this interpretation, implementation is immediate: the platform maintains a running record of past aggregate demand
realizations, computes the adjustments prescribed by the corollary, and posts the corresponding seller-level
quantities for period $t$. In particular, the rule is lightweight---it requires only $\mu$, $\sigma/\sigma_{\LB}$,
and one or two lags of $D_t$---and it is transparent, since each seller's allocation is a uniform share $D_t/N$ plus
a signed correction that depends on recent deviations of market demand from its mean.\smallskip

In most operational settings, however, demand is not allocated in a single end-of-period ``batch''. Rather, orders
arrive sequentially and must be routed in real time to a specific seller (or fulfillment mode) subject to service
levels, capacity, and latency constraints. When demand allocations are executed unit-by-unit, the platform typically
cannot wait until the end of the period to learn $D_t$ before deciding $\{D_{nt}\}$, and even if some form of batching is possible, deferring assignment would delay fulfillment and degrade customer experience. Moreover, practical marketplaces often impose additional frictions---partial cancellations, heterogeneous shipping promises, and seller availability---that make a pure after-the-fact redistribution of realized demand infeasible.\smallskip

For these reasons, Corollary~\ref{cor:timedomain} is best viewed as a \emph{theoretical benchmark} describing the
target seller-level demand processes implied by the transfer-function design. We therefore complement it with a
dynamic allocation mechanism that routes each arriving order immediately, using only information available up to
that moment, while ensuring that the induced seller-level demand streams track (as closely as possible) the
benchmark allocations prescribed by the corollary.\smallskip

Corollary~\ref{cor:timedomain} implies that the benchmark allocation can be written as
\begin{equation}\label{eq:benchmark_decomp}
D_{nt}=\frac{D_t}{N}+b_{n,t},
\end{equation}
where the offset $b_{n,t}$ depends only on lagged aggregate demand (and parameters), hence is known at the beginning
of period $t$ even though $D_t$ is not yet realized. Concretely,
\begin{equation}\label{eq:offsets}
b_{n,t}= \dfrac{\sigma}{N\,\sigma_{\LB}}\,
\begin{cases}
(-1)^n\,(D_{t-1}-\mu), & \text{if $N$ is even or $N$ is odd and $n\ge 3$},\\[1.2ex]
(D_{t-1}-\mu)+(D_{t-2}-\mu), & \text{if $N$ is odd and $n=1$},\\[1.2ex]
-(D_{t-2}-\mu), & \text{if $N$ is odd and $n=2$}.
\end{cases}
\end{equation}\smallskip

Within period $t$, index arriving orders by $k=1,2,\dots$. Define $A_{n,t}(k)$ as the number of orders assigned to
seller $n$ among the first $k-1$ arrivals in period $t$, with $A_{n,t}(1)=0$. The platform routes orders using the
offset-tracking policy in Algorithm~\ref{alg:offsettracking}.\smallskip

\begin{algorithm}[h]
\caption{Offset-Tracking}
\label{alg:offsettracking}
\begin{algorithmic}[1]
\Require Number of sellers $N$, parameters $(\mu,\sigma,\sigma_{\LB})$, lagged aggregate demand $(D_{t-1},D_{t-2})$.
\State Compute offsets $\{b_{n,t}\}_{n\in[N]}$ from \eqref{eq:offsets}.
\State Initialize counters $A_{n}\gets 0$ for all $n\in[N]$.
\For{$k=1,2,\ldots$ (as orders arrive during period $t$)}
    \State Choose
    \[
    n_k \in \arg\min_{n\in[N]} \bigl\{A_{n}-b_{n,t}\bigr\}
    \]
    (break ties uniformly at random).
    \State Assign order $k$ to seller $n_k$ and update $A_{n_k}\gets A_{n_k}+1$.
\EndFor
\State \Return Implemented seller demands $\widehat D_{nt}\gets A_{n}$ for all $n\in[N]$.
\end{algorithmic}
\end{algorithm}

Because the rule always assigns the next order to the seller with the smallest ``offset-adjusted'' count
$A_n-b_{n,t}$, it equalizes these adjusted counts as tightly as possible as orders arrive. Consequently, at the end
of period $t$ the implemented allocation matches the benchmark up to indivisibility: that is,
\linebreak $\bigl|\widehat D_{nt}-D_{nt}\bigr|\le 1$ for all $n\in[N]$.\smallskip

\begin{prop}[Offset-tracking matches the benchmark up to indivisibility]\label{lem:offsettracking_discrepancy}
Fix a period $t$ and suppose that $D_t\in\mathbb{Z}_+$ orders arrive during the period.
Let $\{b_{n,t}\}_{n\in[N]}$ be the offsets in \eqref{eq:offsets} and define the benchmark target shares
\[
x_{n,t}\;:=\;\frac{D_t}{N}+b_{n,t},\qquad n\in[N].
\]
Assume that $x_{n,t}\ge 0$ for all $n$.\footnote{This feasibility condition is automatically satisfied whenever the
benchmark allocation in \eqref{eq:benchmark_decomp} is interpreted as a (nonnegative) order allocation for the
realized $D_t$.} Let $\widehat D_{nt}$ denote the implemented seller demands returned by
Algorithm~\ref{alg:offsettracking}. Then for every $n\in[N]$,
\[
\bigl|\widehat D_{nt}-x_{n,t}\bigr|
=
\left|\widehat D_{nt}-\left(\frac{D_t}{N}+b_{n,t}\right)\right|
\le 1.
\]
\end{prop}\smallskip

\begin{rem}[Additional operational constraints]{\sf
In practice, order assignments may be subject to additional operational constraints (e.g., inventory availability,
delivery promises, or capacity limits). The same policy can be modified by restricting the minimization in
Algorithm~\ref{alg:offsettracking} to a feasible set $\mathcal{S}_t(k)\subseteq[N]$ of sellers available to serve
the $k^{\mbox{\tiny th}}$ order arriving during period $t$.}
\end{rem}

\subsection{Non-Uniqueness}\label{sec:non-unique}
Propositions~\ref{prop:weaklyeven} and \ref{prop:weaklyodd} are \emph{existence} results: for any admissible target $\sigma$, typically infinitely many distinct allocation rules generate the same seller-level root MSFE. For instance, when $N$ is even, the demand allocation policy $\psi_n(z)=\frac{\psi(z)}{N}\,\C{T}_n(z)$ with
\begin{equation}\label{eq:non_unique}
\C{T}_n(z)=1+(-1)^n\,\bar\alpha\, z^k,
\qquad
\bar\alpha=\frac{\sigma}{\sigma_{\LB}},
\end{equation}
implements $\sigma_n^\pi=\sigma$ for every $n\in[N]$ and every integer $k\in\mathbb{N}$, so varying the lag length $k$ yields a distinct policy with identical forecasting implications.\footnote{Indeed, for each $n$, the transfer function $\C{T}_n(z)$ admits an inner--outer factorization of the form
$\C{T}_n(z)=\C{O}_n(z)\,\C{I}_n(z)$,
with
$\C{O}_n(z)=(-1)^n\bar\alpha+z^k$
and
$\C{I}_n(z)=({1+(-1)^n\bar\alpha\, z^k})/({(-1)^n\bar\alpha+z^k})$.
Since $|\C{O}_n(0)|=\bar\alpha=\sigma/\sigma_{\LB}$, it follows from \cref{lem:msfe_outer_main} that the induced root MSFE is exactly $\sigma$.} Our preference for the constructions in Propositions~\ref{prop:weaklyeven} and \ref{prop:weaklyodd} is driven by \emph{parsimony}: they deliver minimal-memory policies with transparent time-domain interpretations. By \cref{cor:timedomain}, the even-$N$ policy depends only on current and one-period-lagged demand, whereas \eqref{eq:non_unique} with $k>1$ requires demand $k$ periods in the past. Thus, although both policies induce exactly the same seller-level root MSFE, the former is operationally simpler, requires less historical information, and is arguably a more natural benchmark for implementation.

\section{Illustrative Example}\label{sec:IllustrativeExample}

Consider ten sellers serving demand concentrated in a single urban market, but located at different distances from that market, as depicted in the left panel of Figure~\ref{fig:HoldingFulfillment_Combined}. Geography creates a fundamental cost tension. Sellers closer to the city can ship cheaply because they are nearer to customers, but they pay more for urban warehouse space. Sellers farther away face the opposite problem: storage is cheap, but every shipment travels far and costs more (see \citealp{fried2023ecommerce} for a detailed discussion).\footnote{Last-mile facilities command substantially higher rents than more remote warehouses; see, for example, \citet{dablanc2014logistics}, \citet{wealthmgmt2024lastmile}. Likewise, parcel shipping costs rise with delivery distance and shipping zone; see \citet{ups2026rates}.} The right panel of Figure~\ref{fig:HoldingFulfillment_Combined} reports the cost parameters $(h_n,b_n,f_n)$, which reflect this geographical trade-off, along with the implied safety-stock coefficients $K_n^{\ti{FBM}}$ and $K_n^{\ti{FBP}}$ for each seller. This cost tension shapes each seller's choice between \textup{FBP} and \textup{FBM}. Distant sellers, facing high fulfillment costs,  gain the most from outsourcing logistics to the platform. Urban sellers, who already ship cheaply on their own, have less benefit. The question is then how the platform's demand-allocation policy interacts with this geographic heterogeneity.


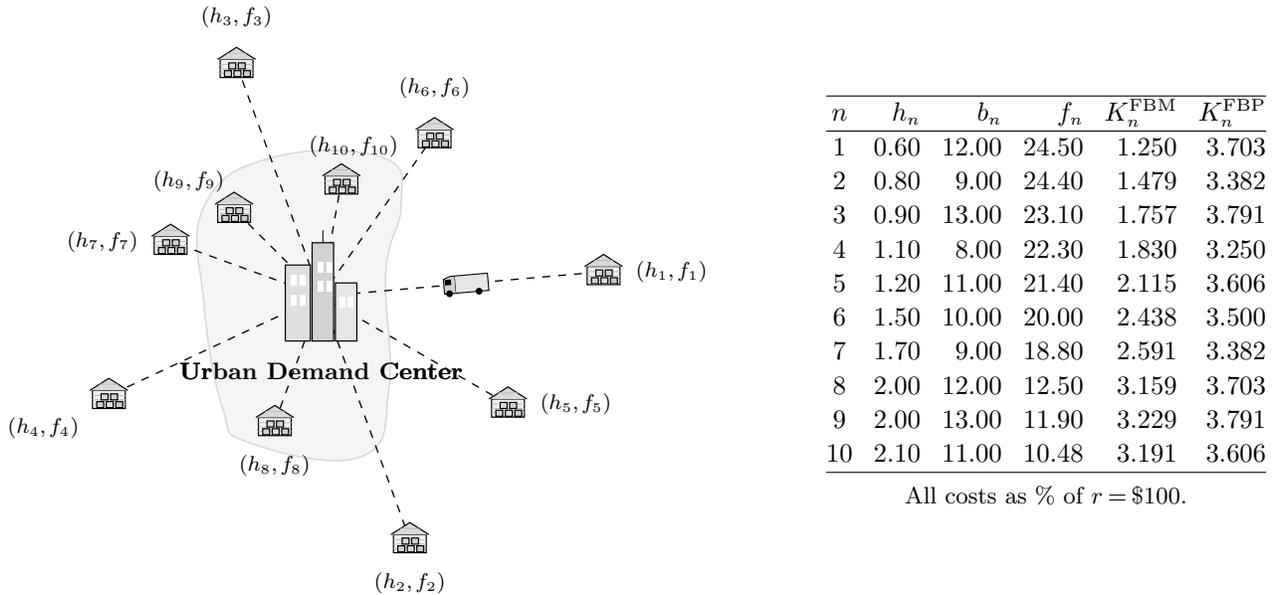
\begin{figure}[h!]
\centering

\begin{minipage}[c]{0.57\linewidth}
\centering
\begin{tikzpicture}[font=\small,
    costlabel/.style={font=\scriptsize, align=center, text=black},
    scale=0.82]


  \newcommand{\cityicon}[3]{%
    \begin{scope}[shift={(#1,#2)}, scale=#3, line width=0.45pt]
      \fill[gray!25] (-0.40,-0.50) rectangle (-0.12,0.35);
      \draw[black]   (-0.40,-0.50) rectangle (-0.12,0.35);
      \fill[white] (-0.34,-0.10) rectangle (-0.28,0.02);
      \fill[white] (-0.24,-0.10) rectangle (-0.18,0.02);
      \fill[white] (-0.34, 0.12) rectangle (-0.28,0.24);
      \fill[white] (-0.24, 0.12) rectangle (-0.18,0.24);
      \fill[gray!35] (-0.10,-0.50) rectangle (0.14,0.60);
      \draw[black]    (-0.10,-0.50) rectangle (0.14,0.60);
      \fill[white] (-0.04,0.05) rectangle (0.02,0.17);
      \fill[white] ( 0.06,0.05) rectangle (0.12,0.17);
      \fill[white] (-0.04,0.27) rectangle (0.02,0.39);
      \fill[white] ( 0.06,0.27) rectangle (0.12,0.39);
      \draw[black, line width=0.3pt] (0.02,0.60) -- (0.02,0.74);
      \fill[gray!20] (0.16,-0.50) rectangle (0.40,0.15);
      \draw[black]   (0.16,-0.50) rectangle (0.40,0.15);
      \fill[white] (0.21,-0.12) rectangle (0.27,0.00);
      \fill[white] (0.30,-0.12) rectangle (0.36,0.00);
    \end{scope}%
  }

  \newcommand{\warehouseicon}[3]{%
    \begin{scope}[shift={(#1,#2)}, scale=#3, line width=0.45pt]
      \fill[gray!12] (-0.38,-0.30) rectangle (0.38,0.12);
      \draw[black]   (-0.38,-0.30) rectangle (0.38,0.12);
      \foreach \yy in {-0.20,-0.10,0.00}{%
        \draw[gray!50, line width=0.25pt] (-0.38,\yy) -- (0.38,\yy);}
      \fill[gray!30] (-0.44,0.12) -- (0,0.40) -- (0.44,0.12) -- cycle;
      \draw[black]   (-0.44,0.12) -- (0,0.40) -- (0.44,0.12);
      \draw[gray!50, line width=0.25pt] (-0.22,0.26) -- (0.22,0.26);
      \fill[gray!35] (-0.28,-0.28) rectangle (-0.12,-0.14);
      \draw[black]   (-0.28,-0.28) rectangle (-0.12,-0.14);
      \fill[gray!25] (-0.06,-0.28) rectangle ( 0.10,-0.14);
      \draw[black]   (-0.06,-0.28) rectangle ( 0.10,-0.14);
      \fill[gray!40] ( 0.16,-0.28) rectangle ( 0.32,-0.14);
      \draw[black]   ( 0.16,-0.28) rectangle ( 0.32,-0.14);
      \fill[gray!30] (-0.18,-0.10) rectangle (-0.02, 0.04);
      \draw[black]   (-0.18,-0.10) rectangle (-0.02, 0.04);
      \fill[gray!35] ( 0.04,-0.10) rectangle ( 0.20, 0.04);
      \draw[black]   ( 0.04,-0.10) rectangle ( 0.20, 0.04);
    \end{scope}%
  }

  \newcommand{\truckicon}[4]{%
    \begin{scope}[shift={(#1,#2)}, rotate=#3, scale=#4, yscale=-1, line width=0.40pt]
      \fill[gray!20] (-0.42,-0.14) rectangle (0.12,0.14);
      \draw[black]   (-0.42,-0.14) rectangle (0.12,0.14);
      \fill[gray!35] (0.12,-0.14) -- (0.30,-0.14) -- (0.30,0.10)
                     -- (0.19,0.14) -- (0.12,0.14) -- cycle;
      \draw[black]   (0.12,-0.14) -- (0.30,-0.14) -- (0.30,0.10)
                     -- (0.19,0.14) -- (0.12,0.14);
      \fill[black] (0.18,0.00) rectangle (0.28,0.09);
      \fill[black] (-0.26,-0.14) circle (0.055);
      \fill[black] ( 0.22,-0.14) circle (0.055);
    \end{scope}%
  }

  \coordinate (C) at (0,0);

  \cityicon{0}{0}{1.45}
  \node[font=\footnotesize\bfseries, below=24pt, fill=white, inner sep=1pt] at (C)
    {Urban Demand Center};

  \fill[gray!8, draw=gray!30, line width=0.6pt]
    plot[smooth cycle, tension=0.75] coordinates {
      ( 1.00,  2.25)   
      ( 1.30,  1.50)   
      ( 1.05,  0.40)   
      ( 1.10, -1.50)   
      ( 0.65, -2.60)   
      (-1.15, -2.38)   
      (-1.60, -1.80)   
      (-1.90, -0.10)   
      (-1.90,  1.45)   
      (-0.80,  2.25)   
    };

  \draw[black, dashed, line width=0.5pt] (C) -- ( 0.33,  1.87);  
  \draw[black, dashed, line width=0.5pt] (C) -- (-1.41,  1.41);  
  \draw[black, dashed, line width=0.5pt] (C) -- ( 1.84,  2.62);  
  \draw[black, dashed, line width=0.5pt] (C) -- (-2.44,  0.89);  
  \draw[black, dashed, line width=0.5pt] (C) -- (-3.44, -1.61);  
  \draw[black, dashed, line width=0.5pt] (C) -- (-0.75, -2.05);  
  \draw[black, dashed, line width=0.5pt] (C) -- ( 1.44, -3.95);  
  \draw[black, dashed, line width=0.5pt] (C) -- (-1.37,  3.76);  
  \draw[black, dashed, line width=0.5pt] (C) -- ( 3.03, -1.75);  
  \draw[black, dashed, line width=0.5pt] (C) -- ( 4.58,  0.40);  


  \coordinate (W10) at (0.33, 1.87);
  \warehouseicon{0.33}{1.87}{0.7}
  \node[costlabel, anchor=south west] at (-0.36, 2.10) {$(h_{10},f_{10})$};

  \coordinate (W9) at (-1.41, 1.41);
  \warehouseicon{-1.41}{1.41}{0.7}
  \node[costlabel, anchor=south east] at (-1.4, 1.54) {$(h_9,f_9)$};

  \coordinate (W6) at (1.84, 2.62);
  \warehouseicon{1.84}{2.62}{0.7}
  \node[costlabel, anchor=south] at (1.84, 3.05) {$(h_6,f_6)$};

  \coordinate (W7) at (-2.44, 0.89);
  \warehouseicon{-2.44}{0.89}{0.7}
  \node[costlabel, anchor=east] at (-2.80, 0.89) {$(h_7,f_7)$};

  \coordinate (W4) at (-3.44, -1.61);
  \warehouseicon{-3.44}{-1.61}{0.7}
  \node[costlabel, anchor=north east] at (-3.76, -1.80) {$(h_4,f_4)$};

  \coordinate (W8) at (-0.75, -2.05);
  \warehouseicon{-0.75}{-2.05}{0.7}
  \node[costlabel, anchor=north] at (-0.75, -2.41) {$(h_8,f_8)$};

  \coordinate (W2) at (1.44, -3.95);
  \warehouseicon{1.44}{-3.95}{0.7}
  \node[costlabel, anchor=north] at (1.44, -4.31) {$(h_2,f_2)$};

  \coordinate (W3) at (-1.37, 3.76);
  \warehouseicon{-1.37}{3.76}{0.7}
  \node[costlabel, anchor=south] at (-1.37, 4.19) {$(h_3,f_3)$};

  \coordinate (W5) at (3.03, -1.75);
  \warehouseicon{3.03}{-1.75}{0.7}
  \node[costlabel, anchor=west] at (3.39, -1.75) {$(h_5,f_5)$};

  \coordinate (W1) at (4.58, 0.40);
  \warehouseicon{4.58}{0.40}{0.7}
  \node[costlabel, anchor=west] at (4.94, 0.40) {$(h_1,f_1)$};

  \truckicon{2.29}{0.20}{185}{1.0}

  \cityicon{0}{0}{1.45}
  \node[font=\footnotesize\bfseries, below=24pt, fill=none, inner sep=1pt] at (C)
    {Urban Demand Center};

\end{tikzpicture}
\end{minipage}%
\hfill
\begin{minipage}[c]{0.40\linewidth}
\centering
\small
\renewcommand{\arraystretch}{1.12}
\begin{tabular}{@{}crrrrr@{}}
\hline
$n$ & $h_n$ & $b_n$ & $f_n$ & $K^{\mathrm{FBM}}_n$ & $K^{\mathrm{FBP}}_n$ \\
\hline
 1 & 0.60 & 12.00 & 24.50 & 1.250 & 3.703 \\
 2 & 0.80 &  9.00 & 24.40 & 1.479 & 3.382 \\
 3 & 0.90 & 13.00 & 23.10 & 1.757 & 3.791 \\
 4 & 1.10 &  8.00 & 22.30 & 1.830 & 3.250 \\
 5 & 1.20 & 11.00 & 21.40 & 2.115 & 3.606 \\
 6 & 1.50 & 10.00 & 20.00 & 2.438 & 3.500 \\
 7 & 1.70 &  9.00 & 18.80 & 2.591 & 3.382 \\
 8 & 2.00 & 12.00 & 12.50 & 3.159 & 3.703 \\
 9 & 2.00 & 13.00 & 11.90 & 3.229 & 3.791 \\
10 & 2.10 & 11.00 & 10.48 & 3.191 & 3.606 \\
\hline
\end{tabular}

\smallskip
{\footnotesize All costs as \% of $r=\$100$.}
\end{minipage}

\caption{Geographic trade-off between holding and fulfillment costs (left)
and seller cost parameters with safety-stock coefficients (right).
Each seller~$n$ is located at holding cost $h_n$ and fulfillment cost $f_n$;
sellers nearer the city face higher $h_n$ but lower $f_n$. The shaded region highlights the urban area containing sellers 8, 9, and 10, who switch from FBP to FBM in the optimal policy.}
\label{fig:HoldingFulfillment_Combined}
\end{figure}


To make this concrete, we normalize the gross per-unit margin to $r=\$100$ and express all operating costs as percentages of $r$ (and hence, numerically, in dollars). The platform charges an intermediation fee $\rho=15$ per unit, and its effective \textup{FBP} seller-side costs are
\[
(F,H)=(10,2.5).
\]
On the platform side, the net \textup{FBP} fulfillment payoff and storage payoff are
\[
\Delta f=2,\qquad \Delta h=2.
\]
These values imply that $F\leq f_n$ and $H\geq h_n$ for all sellers: the platform fulfills more cheaply than any seller can on its own, but its storage costs exceed every seller's holding cost.\footnote{For example, Amazon's \textup{FBA} storage and fulfillment fees provide a useful benchmark for the values of $H$ and $F$ used here; see \citet{amazon2025fba,sellerapp2026fbastorage}.} 

For simplicity, we assume i.i.d.\ market demand,
\[
D_t=\mu+5\,\eps_t,\qquad \mu=15,\qquad \Var(\eps_t)=1,
\]
so that the market-demand $z$-transform is $\psi(z)=5$. Hence, by \cref{prop:lowerbound},
\[
\sigma_{\LB}=\frac{|\psi(0)|}{N}=\frac{5}{N}.
\]
With $N=10$, this yields
\[
\sigma_{\LB}=0.5.
\]
Because $N$ is even, \cref{prop:weaklyeven} implies that the platform can implement any $\sigma\geq \sigma_{\LB}$ using an MA(1) modification of the uniform split. The platform's design space is therefore the participation-truncated interval $[\sigma_{\LB},\sigma_{\UB}]$, where $\sigma_{\UB}$ is defined in \eqref{eq:sigmaUB}. In this example,
\[
\sigma_{\UB}=\min_{n\in[N]}\max\!\left\{
\frac{(r-\rho-f_n)(\mu/N)}{K_n^{\ti{FBM}}},
\frac{(r-\rho-F)(\mu/N)}{K_n^{\ti{FBP}}}
\right\}\approx 33.
\]

The platform thus faces a wide range of implementable forecast-error levels. To see how this lever interacts with the geographic trade-off, Figure~\ref{fig:Ex_DFvsDK} maps each seller into the $(\Delta K,\Delta F)$ plane. The shaded cone corresponds to feasible values of $\sigma\in[\sigma_{\LB},\sigma_{\UB}]$, bounded by the two lines
\[
\Delta F=\frac{N\sigma_{\LB}}{\mu}\Delta K
\qquad\text{and}\qquad
\Delta F=\frac{N\sigma_{\UB}}{\mu}\Delta K.
\]
Sellers inside this region may choose either \textup{FBP} or \textup{FBM}, depending on the root MSFE induced by the platform's allocation rule. 

\begin{figure}[h]
\centering

\pgfmathsetmacro{\figscale}{0.7}
\begin{tikzpicture}
\begin{axis}[
    width=\figscale*0.78\linewidth,
    height=\figscale*0.55\linewidth,
    xmin=0, xmax=2.697962,
    ymin=0, ymax=15.000000,
    grid=both,
    major grid style={line width=.2pt, draw=gray!40},
    minor grid style={line width=.1pt, draw=gray!20},
    tick align=outside,
    tick pos=left,
    xlabel={$ \Delta K $},
    ylabel={$ \Delta F $},
    title={Sellers in $(\Delta K,\Delta F)$-space},
]

\addplot[draw=none, fill=gray!10, forget plot]
coordinates { (0.000000,0.000000) (0.662607,15.000000) (2.697962,15.000000) (2.697962,0.899321) (0.000000,0.000000) } \closedcycle;

\addplot[only marks, mark=*, mark size=2.6pt] coordinates {
(2.452693,14.500000)
(1.902455,14.400000)
(2.034033,13.100000)
(1.420230,12.300000)
(1.490433,11.400000)
(1.061128,10.000000)
(0.790304,8.800000)
(0.543475,2.500000)
(0.562064,1.900000)
(0.414429,0.480000)
};

\node[anchor=west] at (axis cs:2.506652,14.500000) {1};
\node[anchor=west] at (axis cs:1.956414,14.400000) {2};
\node[anchor=west] at (axis cs:2.087992,13.100000) {3};
\node[anchor=west] at (axis cs:1.474190,12.300000) {4};
\node[anchor=west] at (axis cs:1.544393,11.400000) {5};
\node[anchor=west] at (axis cs:1.115088,10.000000) {6};
\node[anchor=west] at (axis cs:0.844263,8.800000) {7};
\node[anchor=west] at (axis cs:0.597434,2.500000) {8};
\node[anchor=west] at (axis cs:0.616023,1.900000) {9};
\node[anchor=west] at (axis cs:0.468388,0.480000) {10};

\addplot[black, line width=1.0pt] coordinates {(0,0) (2.697962,0.899321)} node[pos=0.53, sloped, above] {$\Delta F = \frac{N\sigma_{\LB}}{\mu}\,\Delta K$};
\addplot[black, line width=1.0pt] coordinates {(0,0) (0.662607,15.000000)} node[pos=0.62, sloped, above] {$\Delta F = \frac{N\sigma_{\UB}}{\mu}\,\Delta K$};
\addplot[blue, dashed, line width=1.0pt] coordinates {(0,0) (2.537269,15.000000)} node[pos=0.4, sloped, above, text=blue] {$\Delta F = \frac{N\sigma^{\star}}{\mu}\,\Delta K$};

\end{axis}
\end{tikzpicture}


\caption{Sellers in the $(\Delta K,\Delta F)$ plane. The shaded cone corresponds to feasible values $\sigma\in[\sigma_{\LB},\sigma_{\UB}]$, defined by the boundaries $\Delta F=(N\sigma_{\LB}/\mu)\Delta K$ and $\Delta F=(N\sigma_{\UB}/\mu)\Delta K$.}
\label{fig:Ex_DFvsDK}
\end{figure}
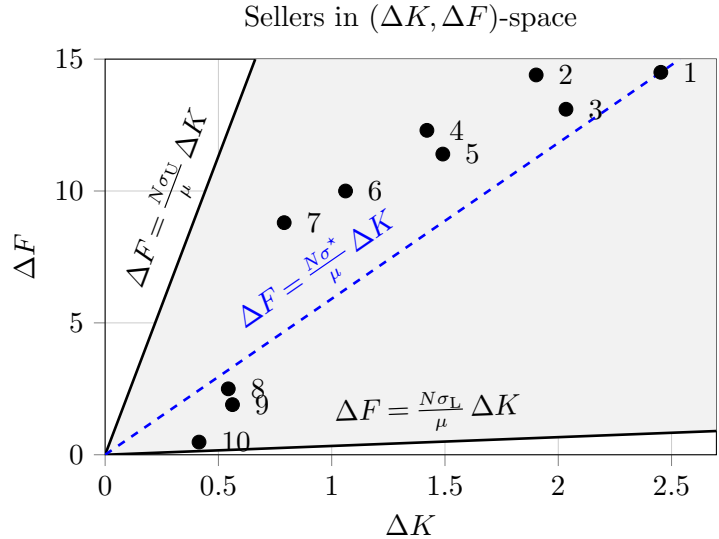

For each $\sigma\in[\sigma_{\LB},\sigma_{\UB}]$, the platform payoff under neutrality is
\[
V(\sigma)
=\rho\mu
+(\Delta f+\Delta h)\frac{\mu}{N}\,\big|N^{\ti{FBP}}(\sigma)\big|
+\Delta h\,\sigma\sum_{n\in N^{\ti{FBP}}(\sigma)} \zeta^{\ti{FBP}}_n,\; \text{where}
\]
\[
N^{\ti{FBP}}(\sigma)=\left\{n\in[N]: \Delta F_n \ge (N\sigma/\mu)\Delta K_n\right\}.
\]
As $\sigma$ rises, each seller $n$ reaches a breakpoint
\[
\sigma_n=\frac{\mu\,\Delta F_n}{N\,\Delta K_n}
\]
at which it switches from \textup{FBP} to \textup{FBM}, causing a jump in $V(\sigma)$. The sellers who switch first---at low values of $\sigma_n$---are the urban sellers, for whom the fulfillment advantage $\Delta F_n$ is small relative to the safety-stock disadvantage $\Delta K_n$. Remote sellers, with large $\Delta F_n$, remain in \textup{FBP} even at high MSFE.\medskip

Figure~\ref{fig:IllustrativeExample_Panels} plots the number of \textup{FBP} adopters and the platform payoff as functions of $\sigma$. The shaded band marks the feasible interval $\sigma\in[\sigma_{\LB},\sigma_{\UB}]$. By the participation constraint, we set $V(\sigma)=0$ for $\sigma>\sigma_{\UB}$.

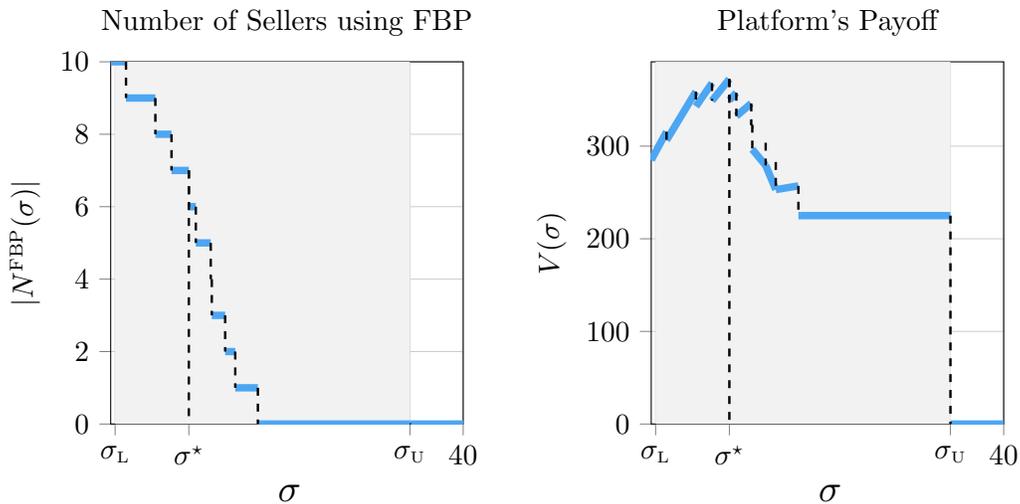
\begin{figure}[h]
\centering

\pgfmathsetmacro{\figscale}{0.75}
\begin{tikzpicture}
\definecolor{lightblue}{rgb}{0.30,0.65,0.95}

\pgfmathsetmacro{\sigL}{0.500000}
\pgfmathsetmacro{\sigU}{33.956782}
\pgfmathsetmacro{\sigstar}{8.867804}

\begin{groupplot}[
    group style={group size=2 by 1, horizontal sep=2.5cm},
    width=\figscale*0.47\linewidth,
    height=\figscale*0.48\linewidth,
    xmin=0.000000, xmax=40.000000,
    grid=both,
    major grid style={line width=.2pt, draw=gray!40},
    minor grid style={line width=.1pt, draw=gray!20},
    tick align=outside,
    tick pos=left,
    xtick={\sigL,\sigstar,\sigU,40},
    xticklabels={$\sigma_{\LB}$,$\sigma^\star$,$\sigma_{\UB}$,$40$},
    xlabel={\large $\,\sigma$},
]

\nextgroupplot[
    ymin=0, ymax=10,
    ylabel={\normalem $|N^{\mathrm{FBP}}(\sigma)|$},
    title={Number of Sellers using FBP},
    title style={font=\normalem},
]
\addplot[draw=none, fill=gray!10, forget plot]
coordinates {(\sigL,0) (\sigU,0) (\sigU,10) (\sigL,10)} \closedcycle;

\addplot[lightblue, line width=2.8pt] coordinates {(0.000000,10) (1.737330,10)};
\addplot[lightblue, line width=2.8pt] coordinates {(1.737330,9) (5.070600,9)};
\addplot[lightblue, line width=2.8pt] coordinates {(5.070600,8) (6.900048,8)};
\addplot[lightblue, line width=2.8pt] coordinates {(6.900048,7) (8.867804,7)};
\addplot[lightblue, line width=2.8pt] coordinates {(8.867804,6) (9.660610,6)};
\addplot[lightblue, line width=2.8pt] coordinates {(9.660610,5) (11.353750,5)};
\addplot[lightblue, line width=2.8pt] coordinates {(11.353750,4) (11.473172,4)};
\addplot[lightblue, line width=2.8pt] coordinates {(11.473172,3) (12.990850,3)};
\addplot[lightblue, line width=2.8pt] coordinates {(12.990850,2) (14.135895,2)};
\addplot[lightblue, line width=2.8pt] coordinates {(14.135895,1) (16.702444,1)};
\addplot[lightblue, line width=2.8pt] coordinates {(16.702444,0) (33.956782,0)};
\addplot[lightblue, line width=2.8pt] coordinates {(33.956782,0) (40.000000,0)};
\addplot[black, dashed, line width=0.9pt] coordinates {(1.737330,10) (1.737330,9)};
\addplot[black, dashed, line width=0.9pt] coordinates {(5.070600,9) (5.070600,8)};
\addplot[black, dashed, line width=0.9pt] coordinates {(6.900048,8) (6.900048,7)};
\addplot[black, dashed, line width=0.9pt] coordinates {(8.867804,7) (8.867804,6)};
\addplot[black, dashed, line width=0.9pt] coordinates {(9.660610,6) (9.660610,5)};
\addplot[black, dashed, line width=0.9pt] coordinates {(11.353750,5) (11.353750,4)};
\addplot[black, dashed, line width=0.9pt] coordinates {(11.473172,4) (11.473172,3)};
\addplot[black, dashed, line width=0.9pt] coordinates {(12.990850,3) (12.990850,2)};
\addplot[black, dashed, line width=0.9pt] coordinates {(14.135895,2) (14.135895,1)};
\addplot[black, dashed, line width=0.9pt] coordinates {(16.702444,1) (16.702444,0)};
\addplot[black, dashed, line width=0.9pt] coordinates {(33.956782,0) (33.956782,0)};
\addplot[black, dashed, line width=0.9pt] coordinates {(8.867804,7) (8.867804,0)};

\nextgroupplot[
    ymin=0, ymax=391.074061,
    ylabel={\normalem $V(\sigma)$},
    title={Platform's Payoff},
    title style={font=\normalem},
]
\addplot[draw=none, fill=gray!10, forget plot]
coordinates {(\sigL,0) (\sigU,0) (\sigU,391.074061) (\sigL,391.074061)} \closedcycle;

\addplot[lightblue, line width=2.8pt] coordinates {(0.000000,285.000000) (1.737330,315.491411)};
\addplot[lightblue, line width=2.8pt] coordinates {(1.737330,306.378882) (5.070600,358.908484)};
\addplot[lightblue, line width=2.8pt] coordinates {(5.070600,342.877126) (6.900048,368.088446)};
\addplot[lightblue, line width=2.8pt] coordinates {(6.900048,349.051915) (8.867804,372.451487)};
\addplot[lightblue, line width=2.8pt] coordinates {(8.867804,349.697198) (9.660610,357.626971)};
\addplot[lightblue, line width=2.8pt] coordinates {(9.660610,332.515027) (11.353750,346.100482)};
\addplot[lightblue, line width=2.8pt] coordinates {(11.353750,322.365157) (11.473172,323.136839)};
\addplot[lightblue, line width=2.8pt] coordinates {(11.473172,296.581966) (12.990850,279.159337)};
\addplot[lightblue, line width=2.8pt] coordinates {(12.990850,279.159337) (14.135895,253.081224)};
\addplot[lightblue, line width=2.8pt] coordinates {(14.135895,253.081224) (16.702444,257.090346)};
\addplot[lightblue, line width=2.8pt] coordinates {(16.702444,225.000000) (33.956782,225.000000)};
\addplot[lightblue, line width=2.8pt] coordinates {(33.956782,0.000000) (40.000000,0.000000)};
\addplot[black, dashed, line width=0.9pt] coordinates {(1.737330,315.491411) (1.737330,306.378882)};
\addplot[black, dashed, line width=0.9pt] coordinates {(5.070600,358.908483) (5.070600,342.877126)};
\addplot[black, dashed, line width=0.9pt] coordinates {(6.900048,368.088445) (6.900048,349.051915)};
\addplot[black, dashed, line width=0.9pt] coordinates {(8.867804,372.451487) (8.867804,349.697198)};
\addplot[black, dashed, line width=0.9pt] coordinates {(9.660610,357.626971) (9.660610,332.515027)};
\addplot[black, dashed, line width=0.9pt] coordinates {(11.353750,346.100482) (11.353750,322.365157)};
\addplot[black, dashed, line width=0.9pt] coordinates {(11.473172,323.136838) (11.473172,296.581966)};
\addplot[black, dashed, line width=0.9pt] coordinates {(12.990850,303.669819) (12.990850,279.159337)};
\addplot[black, dashed, line width=0.9pt] coordinates {(14.135895,282.875363) (14.135895,253.081224)};
\addplot[black, dashed, line width=0.9pt] coordinates {(16.702444,257.090346) (16.702444,225.000000)};
\addplot[black, dashed, line width=0.9pt] coordinates {(33.956782,225.000000) (33.956782,0)};
\addplot[black, dashed, line width=0.9pt] coordinates {(8.867804,372.451487) (8.867804,0)};

\end{groupplot}
\end{tikzpicture}


\caption{Number of \textup{FBP} adopters (left) and platform payoff $V(\sigma)$ (right) as functions of $\sigma$. The shaded area corresponds to the feasible region $\sigma\in[\sigma_{\LB},\sigma_{\UB}]$, with $\sigma_{\LB}=0.5$ and $\sigma_{\UB}=33$.}
\label{fig:IllustrativeExample_Panels}
\end{figure}

The platform-optimal choice is $\sigma^\star\approx 8.87$, far above the minimum $\sigma_{\LB}=0.5$ under a uniform allocation. At this level of forecast error, the most urban sellers, $8$, $9$, and $10$, switch to \textup{FBM}, while the remaining seven, more remote sellers stay in \textup{FBP}. This partition is captured by the shaded region in the left panel of \cref{fig:HoldingFulfillment_Combined}.\smallskip 

Table~\ref{tab:IllustrativeExample_CompareSigma} contrasts the platform's optimal strategy with the uniform benchmark.
\begin{table}[h!]
\centering
\renewcommand{\arraystretch}{1.2}
\rowcolors{3}{gray!10}{white}
\begin{tabular}{l c c}
\specialrule{1.2pt}{0pt}{0pt}
{\bf Demand Allocation} & \textbf{Uniform} & \textbf{Optimal}\\
\midrule
Root MSFE: $\sigma$ & $\sigma_{\LB}=0.50$ & $\sigma^\star=8.87$\\
FBP adopters: $N^{\ti{FBP}}(\sigma)$
& $\{1,2,3,4,5,6,7,8,9,10\}$
& $\{1,2,3,4,5,6,7\}$\\
Number of FBP adopters: $\big|N^{\ti{FBP}}(\sigma)\big|$
& $10$ & $7$\\
Platform payoff: $V(\sigma)$
& $293.78$ & $372.45$\\
Cumulative FBP safety stock: $\Gamma^{\ti{FBP}}(\sigma)$
& $4.39$ & $52.73$\\
Cumulative FBM safety stock: $\Gamma^{\ti{FBM}}(\sigma)$
& $0$ & $28.12$\\
Cumulative seller utility: $\sum_{n} U_n(\sigma)$
& $1107.14$ & $814.47$\\
\specialrule{1.2pt}{0pt}{0pt}
\end{tabular}\smallskip
\caption{Outcomes under the uniform allocation $\sigma=\sigma_{\LB}$ and the platform-optimal choice $\sigma=\sigma^\star$.}
\label{tab:IllustrativeExample_CompareSigma}
\end{table}
Under uniform allocation, all ten sellers adopt \textup{FBP} and carry minimal safety stock but generates a platform payoff of only $V(\sigma_{\LB})=293.78$. Under the optimal policy, the platform raises $\sigma$ nearly eighteen-fold. The urban sellers leave \textup{FBP} because their fulfillment savings were modest---their proximity to the demand center already kept shipping costs low, so the extra inventory burden is too costly. The remote sellers stay, because the platform's fulfillment advantage is large enough to offset the higher safety stock cost. The platform's payoff rises by $26.78\%$ to $V(\sigma^\star)=372.45$. This improvement is accompanied by a sharp increase in aggregate safety stock:
\[
\Gamma^{\ti{FBP}}(\sigma_{\LB})+\Gamma^{\ti{FBM}}(\sigma_{\LB})=4.39
\qquad\text{versus}\qquad
\Gamma^{\ti{FBP}}(\sigma^\star)+\Gamma^{\ti{FBM}}(\sigma^\star)=80.85.
\]
Total buffer inventory rises by a factor of approximately $18$, while cumulative seller utility falls from $1107.14$ to $814.47$, a decline of $26.44\%$. The geographical heterogeneity is what makes this work: the platform can afford to lose the urban sellers because the remote sellers, who benefit most from \textup{FBP}, are also the ones most willing to carry more inventory. In short, the platform exploits the holding--fulfillment trade-off created by geography to simultaneously raise service reliability and extract more surplus from the sellers who rely the most on the platform's logistics network.

\section{Discussion}\label{sec:extension}

This section discusses one extension of the baseline model that helps clarify its scope and robustness. In particular, we examine how our analysis changes when sellers rely on simple forecasting heuristics rather than optimal one-step-ahead predictors. A second extension, which considers heterogeneous replenishment lead times across sellers and fulfillment modes, is relegated to \ref{subsec:leadtime_extension}. Together, these extensions illustrate how the platform’s allocation policy continues to shape \textup{FBP} adoption and inventory levels once we move beyond the benchmark assumptions.

\subsection{Suboptimal forecasts}\label{sec:suboptimal-forecasts}

Our baseline assumes sellers form the \emph{optimal} one-step-ahead forecast, so the uncertainty in the base-stock rule \eqref{eq:basestock} is the \emph{minimal} MSFE, pinned down by the outer factor of the seller-level demand filter (\cref{lem:msfe_outer_main}). In practice, forecasting is often done with low-touch heuristics—moving averages, exponential smoothing, or Holt--Winters—rather than the optimal predictor \citep{Fildes2009EffectiveForecasting}.\smallskip

Suboptimal forecasts typically raise one-step-ahead MSFE and hence require additional safety stock, which could benefit the platform by increasing inventory intensity. But higher perceived uncertainty can also deter \textup{FBP} adoption by worsening sellers’ adoption incentives. Thus, suboptimal forecasting can increase safety stock \emph{conditional on adoption} while reducing the set of adopters, making the net \emph{ex ante} effect on the platform's payoffs ambiguous.\smallskip

Suppose seller $n$ uses some ad-hoc point-forecast method that estimates the demand $D_{n,t+1}$ in period $t+1$ by means of a stable linear filter of the seller’s own demand history,
$$\widehat m_{nt}=\sum_{k=0}^\infty\,\widehat{\psi}_{nk}\,D_{n,t-k}\equiv \widehat{\psi}_n(\mathcal{B})\,D_{nt},$$
where $\mathcal{B}$ denotes the backshift operator and $\widehat{\psi}_n(z)$ satisfies $\widehat{\psi}_n(1)=1$ so that $\e[\widehat m_{nt}]=\e [D_{nt}]=\mu_n$.\smallskip 

For example, if seller $n$ forecasts using {\em simple exponential smoothing}
(SES) with a smoothing parameter $\lambda\in(0,1]$, as in  \cite{ChenRyanSimchiLevi2000ESBullwhip} (see also
\citealp{Gardner1990ForecastPerfInventoryControl,SnyderKoehlerHyndmanOrd2004MeansVariancesLTD,EppenMartin1988SafetyStockStochasticLeadTime}), 
$$\widehat{\psi}^{\ti{SES}}_n(z)=\lambda\,\sum_{k=0}^\infty (1-\lambda)^k\, z^k={\lambda \over 1-(1-\lambda)\,z}.$$

In the steady state, seller $n$'s forecast error can be written as
\[
e_{n,t+1}:=D_{n,t+1}-\widehat m_{nt}
=
D_{n,t+1}-\widehat{\psi}_n(\mathcal{B})\,D_{nt}
=
\big(1-\widehat{\psi}_n(\mathcal{B})\,\mathcal{B}\big)\,D_{n,t+1}.
\]
We measure perceived uncertainty by the stationary root
MSFE,
\[
\tilde\sigma_n^\pi(\lambda):=\Big(\e[e_{n,t+1}^2]\Big)^{1/2},
\]
which replaces $\sigma_n^\pi$ in the seller’s base-stock calculation when forecasting is performed via the linear-filter $\widehat{\psi}_n(z)$.\medskip

As in \eqref{eq:psi_n}, write $\psi_n(z)=\frac{1}{N}\psi(z)\C{T}_n(z)$ for a transfer function $\C{T}_n(z)$
with an inner--outer factorization $\C{T}_n(z)=\C{O}_n(z)\C{I}_n(z)$. It follows that
\begin{equation}\label{eq:SES-MSFE}
(\tilde\sigma_n^\pi(\lambda))^2
=
\frac{1}{2\pi N}\int_{-\pi}^{\pi}
\left|1-\widehat{\psi}_n(e^{-i\theta})\,e^{-i\theta}\right|^2
\left|\psi(e^{-i\theta})\right|^2\,\left|\C{O}_n(e^{-i\theta})\right|^2\,\D\theta
\;\ge\;
\frac{1}{N}\abs{\psi(0)}^2\abs{\C{O}_n(0)}^2
=
(\sigma_n^\pi)^2,
\end{equation}
where the inequality follows from Parseval/Cauchy--Schwarz in $\mathbb{H}^2$\footnote{For
$f(z)=\sum_{k\ge0}f_k z^k\in\mathbb{H}^2$,
\[
\frac{1}{2\pi}\int_{-\pi}^{\pi}\abs{f(e^{-i\theta})}^2\,\D\theta
=
\sum_{k=0}^{\infty}\abs{f_k}^2
\;\ge\;
\abs{f_0}^2
=
\abs{f(0)}^2.
\]
}.
Equation \eqref{eq:SES-MSFE} formalizes that (for a fixed allocation policy) suboptimal forecasts weakly
inflate a seller’s one-step-ahead MSFE, and therefore inflate the safety stock required to maintain a given
service level.\medskip


When sellers forecast optimally, their demand uncertainty is driven solely by the baseline, unpredictable ``noise'' in their demand stream ($\abs{\psi(0)}\,\abs{\C{O}_n(0)}$). They can perfectly filter out all other predictable patterns. In contrast, when sellers rely on simpler, ad-hoc forecasting rules, they cannot perfectly filter the signal. As a result, their perceived uncertainty $\tilde\sigma_n^\pi(\lambda)$ is inflated by the \emph{entire pattern} of the demand stream they receive. This means that if sellers use simple heuristics, the platform cannot simply tune a single, overall level of volatility; it must also carefully consider the specific shape of the fluctuations it routes to sellers (for example, whether the allocated demand bounces up and down rapidly day-by-day, or drifts in slow, extended waves).\medskip

To illustrate this point, let us return to the i.i.d.\ demand example in \cref{sec:IllustrativeExample} and assume seller $n$ uses SES to forecast demand with smoothing parameter $\lambda_n \in [0,1)$. Suppose the platform chooses
its allocation as if sellers were using optimal forecasts, implementing the MA(1) perturbation
\[
\psi_n(z)=\frac{\abs{\psi(0)}}{N}\Big(1+(-1)^n\alpha^\star z\Big),
\qquad
\alpha^\star=\frac{\sigma^\star}{\sigma_{\LB}}.
\]
Then the induced SES MSFE admits the closed form (\ref{Append_Der} for a derivation), 
\begin{equation}\label{eq:MSFE_SES}
(\tilde\sigma_n^\pi(\lambda_n))^2 = \frac{\abs{\psi(0)}^2}{N^2}\left[ 1+\abs{(-1)^n\alpha^\star-\lambda_n}^2+\frac{\lambda_n}{2-\lambda_n}\abs{1-\lambda_n+(-1)^n\alpha^\star}^2 \right]. 
\end{equation}
Furthermore, minimizing over $\lambda_n\in[0,1]$ yields
$\lambda_n^\star=0$. Consequently,
\begin{equation}\label{eq:opt_MSFE-SES}
(\tilde\sigma_n^\pi(\lambda_n^\star))^2=\frac{\abs{\psi(0)}^2}{N^2}+(\sigma^\star)^2
=(\sigma_{\LB})^2+(\sigma^\star)^2,
\end{equation}
so restricting sellers to SES inflates the MSFE relative to the optimal-forecast benchmark $(\sigma^\star)^2$
by the additive term $(\sigma_{\LB})^2$. Interestingly, the corner solution $\lambda_n^\star=0$ corresponds to an \emph{infinitely sluggish} SES rule: the update
$\widehat m_{nt}=\lambda D_{nt}+(1-\lambda)\widehat m_{n,t-1}$ stops reacting to new observations, so the forecast
$\widehat m_{nt}$ becomes essentially constant. In effect, the seller behaves as if demand were
unpredictable from its own history (as if i.i.d.\ around a fixed mean), so the one-step-ahead error is
realized demand relative to that baseline.\smallskip

{Let us apply these results to the numerical instance in \cref{sec:IllustrativeExample} with $N=10$ and $\sigma_{\LB}=0.5$. If the platform keeps the benchmark design
$\sigma^\star\approx 8.87$, then, according to \eqref{eq:opt_MSFE-SES}, sellers perceive $\tilde{\sigma}\approx 8.88$.
As a result, the set of
\textup{FBP} adopters shrinks from $\{1, 2, 3, 4, 5, 6, 7\}$ to $\{2, 3, 4, 5, 6, 7\}$ (seller~$1$ drops out).
Total \textup{FBP}
safety stock falls from $\Gamma^{\ti{FBP}}(\sigma^\star)\approx 52.73$ to
$\Gamma^{\ti{FBP}}_{\ti{SES}}(\sigma^\star)\approx 44.35$ (a $15.9\%$ reduction), and the platform payoff drops
from $V(\sigma^\star)\approx 372.45$ to $V_{\ti{SES}}(\sigma^\star)\approx 349.70$ (a $6.1\%$ reduction).}\medskip

To summarize, when sellers rely on low-touch forecasting heuristics, allocation policy affects not only \emph{realized}
demand risk but also \emph{perceived} risk, and the latter feeds directly into both safety-stock levels and
platform adoption. The key implication is that ``turning up volatility'' to stimulate inventory intensity can
backfire by pushing marginal sellers out of platform fulfillment. Platforms should therefore (i) anticipate
the forecasting rules that sellers actually use when optimizing allocation, and (ii) design policies that are
robust to forecast suboptimality---either by moderating induced volatility near adoption breakpoints or by
reducing the perception gap via forecasting support, defaults, or information that makes demand more
predictable from the sellers’ perspective.

\section{Conclusion}\label{sec:conclusions}

Digital marketplaces increasingly act as \emph{market makers}: beyond matching, they algorithmically route order flow among competing sellers and, in many categories, provide managed fulfillment and inventory-carrying services. This paper studies how a platform can use \emph{intertemporal} demand allocation to influence sellers’ inventory decisions when sellers replenish via base-stock policies and choose between fulfill-by-merchant (FBM) and fulfill-by-platform (FBP). Motivated by accountability and fairness concerns, we focus on \emph{neutral} allocation mechanisms that treat sellers symmetrically in expectation by equalizing both long-run mean demand shares and one-step-ahead forecast uncertainty (root MSFE).\smallskip

Our main results show that neutrality still leaves substantial scope for operational design. Uniform splitting attains a sharp lower bound on sellers’ root MSFE, and any higher level can be implemented via simple low-order moving-average perturbations of the uniform split---MA(1) when $N$ is even, MA(2) when $N$ is odd---without changing mean demand shares. 
This characterization yields a tractable platform problem: under neutrality, the platform effectively chooses a scalar uncertainty level that determines (i) which sellers adopt FBP (the extensive margin) and (ii) the safety stock carried by adopters (the intensive margin). A distinctive structural implication is that any strict improvement beyond the uniform benchmark requires \emph{non-invertibility} from the sellers’ perspective, highlighting the platform’s informational advantage.\smallskip

Managerially, the model clarifies a fundamental extensive--intensive trade-off for platforms that monetize both fulfillment activity and on-platform storage. Increasing forecast uncertainty raises safety stocks (and therefore on-platform inventory level) among remaining FBP adopters, but it can also rotate marginal sellers toward FBM by making FBP more costly. Because adoption changes only at seller-specific thresholds while induced safety stock scales linearly between thresholds, the platform’s payoff is piecewise linear in the induced root MSFE, making optimal neutral policies transparent and implementable. We also provide a demand-allocation interpretation of our transfer-function allocations, and we show that
when sellers forecast using common heuristics (e.g., simple exponential smoothing), the platform’s design lever shifts from controlling a scalar MSFE to shaping the \emph{frequency content} of allocated demand.\smallskip


Several directions for future research would sharpen these conclusions and address natural limitations of our baseline model. \emph{First}, neutrality can be defined in richer ways. Beyond equal mean shares and equal one-step-ahead MSFE, one could require neutrality over longer horizons (multi-step forecast errors), preserve the entire second-order structure (e.g., identical autocovariances or spectral densities), or impose exposure/share floors motivated by fairness regulation; characterizing the implementable set and the platform’s optimal design under such alternative constraints is an important next step. \emph{Second}, we restrict attention to stationary demand and stationary allocation rules. Extending the analysis to nonstationary demand (seasonality, trends, regime shifts) and time-varying allocations would connect the mechanism to realistic planning environments. 
\emph{Third}, incorporating lost sales and fulfillment-capacity congestion would link allocation design more directly to service levels and operational bottlenecks and would endogenize additional participation and feasibility constraints. \emph{Finally}, integrating strategic behavior, including pricing, learning, and platform competition, would broaden the framework from inventory shaping to a more comprehensive theory of marketplace operations under governance constraints.\smallskip

Together, these extensions suggest that intertemporal demand allocation is a general and operationally meaningful platform lever: even under neutrality constraints, dynamic order routing can act as an information and operations design instrument that shapes forecastability, inventory, and fulfillment adoption in two-sided marketplaces.

{\footnotesize
\bibliographystyle{ormsv080}
\bibliography{References}
}

\newpage
\begin{appendices}
 \renewcommand{\thesection}{Appendix \Alph{section}}

\renewcommand{\thefootnote}{\fnsymbol{footnote}}
\setcounter{footnote}{1}
\setcounter{equation}{0}
\renewcommand{\theequation}{A\arabic{equation}}
\vspace{0.3cm}

\section{Proofs}\label{Append:Proofs}

{\sc Proof of \cref{prop:lowerbound}:} {
Let $\psi(z)$ and $\psi_n(z)$ denote the $z$-transforms of the market demand process $\{D_t\}$ and seller $n$'s demand process $\{D_{nt}\}$, respectively. Because the policy $\pi$ is admissible, we have
\[
\psi(z)=\sum_{n \in [N]} \psi_n(z), \qquad z \in \mathbb{D}.
\]

For each $n \in [N]$, let $\psi_n(z)=\C{O}_n(z)\,\C{I}_n(z)$ denote the Nevanlinna inner--outer factorization of $\psi_n(z)$. It follows that $\sigma^{\pi}_n=\abs{\C{O}_n(0)}$, and hence
\[
\sum_{n \in [N]} \sigma^{\pi}_n \;=\; \sum_{n \in [N]} \abs{\C{O}_n(0)}.
\]

Also, from the admissibility of $\pi$,
\[
\psi(0)=\sum_{n \in [N]} \psi_n(0)=\sum_{n \in [N]} \C{O}_n(0)\,\C{I}_n(0).
\]
As a result,
\[
\abs{\psi(0)}
=
\Big|\sum_{n \in [N]} \C{O}_n(0)\,\C{I}_n(0)\Big|
\leq
\sum_{n \in [N]} \abs{\C{O}_n(0)}\,\abs{\C{I}_n(0)}
\leq
\sum_{n \in [N]} \abs{\C{O}_n(0)}
=
\sum_{n \in [N]} \sigma^{\pi}_n,
\]
where the second inequality follows from the fact that each $\C{I}_n(z)$ is inner. Indeed, since $\C{I}_n$ is analytic on the unit disc $\mathbb{D}=\{z \in \mathbb{C}\colon |z|<1\}$, the maximum modulus principle implies
\[
\abs{\C{I}_n(0)}
\leq
\max_{z \in \mathbb{D}} \abs{\C{I}_n(z)}
=
\max_{z \in \mathbb{T}} \abs{\C{I}_n(z)}
=1,
\]
and the last equality holds because $\C{I}_n$ is unimodular on the unit circle $\mathbb{T}=\{z \in \mathbb{C}\colon |z|=1\}$. We therefore have that $
\sum_{n \in [N]} \sigma^{\pi}_n \;\ge\; \abs{\psi(0)}$. 
Finally, if $\pi$ is a neutral demand allocation policy then $\sigma^{\pi}_n=\sigma$, independent of $n$, and the previous inequality implies  $\sigma^{\pi}_n \geq \frac{\abs{\psi(0)}}{N}$. 
\qed
}\medskip

{\sc Proof of \cref{cor:uniform}:} 
Since the market demand in \eqref{def:Dt} is invertible by construction, it follows that $\psi(z)$ is outer. As a result, under the uniform allocation policy $\psi_n(z)=\psi(z)/N$, seller $n$'s allocated demand filter is also outer. Hence, by \cref{lem:msfe_outer_main},
\[
\sigma_n^{\pi}=\abs{\psi_n(0)}=\frac{\abs{\psi(0)}}{N}=\sigma_{\LB}.
\]
\hfill \qed
\medskip

{\sc Proof of \cref{prop:cost}:}
In this context, with a slight abuse of notation, we let $\psi \in \ell^2$ and $\psi_n \in \ell^2$ denote the coefficient sequences of the $\mathrm{MA}(\infty)$ representations of the aggregate market demand and seller $n$'s allocated demand, respectively. By definition, the unconditional variance of a weakly stationary $\mathrm{MA}(\infty)$ process driven by standard white noise ($\sigma_\eps=1$) is the squared $\ell^2$-norm of this coefficient sequence. We are given the variance constraints for all $n \in [N]$:
\begin{equation*}
    \norm{\psi_{n}}^2 = \frac{1}{N^2}\norm{\psi}^2.
\end{equation*}
Admissibility (\cref{def:Admissible}) requires that the allocation perfectly clears the market:
\begin{equation*}
    \sum_{n=1}^N \psi_{n} = \psi.
\end{equation*}
We expand the squared norm of the aggregate demand using the inner product space properties of $\ell^2$:
\begin{align*}
    \norm{\psi}^2 &= \norm{\sum_{n=1}^N \psi_{n}}^2 \nonumber \\
    &= \sum_{n=1}^N \norm{\psi_{n}}^2 + \sum_{i \neq j} \langle \psi_{i}, \psi_{j} \rangle.
\end{align*}
Substituting the given variance constraints into the expansion yields:
\begin{equation*}
    \norm{\psi}^2 = N \left( \frac{1}{N^2}\norm{\psi}^2 \right) + \sum_{i \neq j} \langle \psi_{i}, \psi_{j} \rangle = \frac{1}{N}\norm{\psi}^2 + \sum_{i \neq j} \langle \psi_{i}, \psi_{j} \rangle.
\end{equation*}
Rearranging this equation to isolate the sum of the inner products gives:
\begin{equation}
    \sum_{i \neq j} \langle \psi_{i}, \psi_{j} \rangle = \frac{N-1}{N}\norm{\psi}^2. \label{eq:inner_product_sum}
\end{equation}
We then apply the Cauchy-Schwarz inequality, which states that $\langle x, y \rangle \le \norm{x}\norm{y}$. Applying this to each pair of our allocation sequences:
\begin{equation*}
    \langle \psi_{i}, \psi_{j} \rangle \le \norm{\psi_{i}} \norm{\psi_{j}} = \left( \frac{1}{N}\norm{\psi} \right) \left( \frac{1}{N}\norm{\psi} \right) = \frac{1}{N^2}\norm{\psi}^2.
\end{equation*}
Because there are exactly $N(N-1)$ terms in the summation over $i \neq j$, the maximum for the sum of the inner products is:
\begin{equation*}
    \sum_{i \neq j} \langle \psi_{i}, \psi_{j} \rangle \le N(N-1) \left( \frac{1}{N^2}\norm{\psi}^2 \right) = \frac{N-1}{N}\norm{\psi}^2.
\end{equation*}
Comparing this upper bound to our derived exact value in \eqref{eq:inner_product_sum}, we see that the Cauchy-Schwarz inequality must hold with strict equality for every pair $(i, j)$. 

Therefore, we have $\psi_{1} = \psi_{2} = \dots = \psi_{N}$. Substituting this back into the admissibility constraint, we obtain $\psi_{n} = \frac{1}{N}\psi$ for all $n \in [N]$.
\hfill  \qed
\medskip

{\sc Proof of \cref{prop:weaklyeven}.}  Define $\alpha_n=(-1)^n\, \ba$.
First, since $N$ is even, we have $\sum_n\alpha_n=0$, which implies $\sum_n\psi_n(z)=\psi(z)$; hence $\pi$ is an admissible demand allocation. Moreover, since $\sigma \geq \sigma_{\LB}$, we have $\abs{\alpha_n}\ge 1$ for all $n$. Thus, as in \cref{exm:iid}, transfer function $\C{T}_n(z)=1+\alpha_n z$ admits the inner–outer factorization
\[
\C{T}_n(z)=\C{O}_n(z)\,\C{I}_n(z),\quad \text{with~~  }\C{O}_n(z)=\alpha_n + z\quad \text{and}\quad \C{I}_n(z)=\frac{1+\alpha_n z}{\alpha_n + z},
\]

where $\C{O}_n(z)\in\mathbb{O}$ and $\C{I}_n(z)\in\mathbb{I}$ is a Blaschke factor. It follows from \cref{lem:msfe_outer_main} that
\[
\sigma^{\pi}_n=\frac{1}{N}\,\abs{\psi(0)}\,\abs{\C{O}_n(0)}=\frac{\abs{\psi(0)}\,\abs{\alpha_n}}{N}=\sigma.
\]
Thus, $\pi$ is neutral and implements the target root MSFE $\sigma$ as required. \hfill \qed

\medskip

{\sc Proof of \cref{prop:weaklyodd}.} To show that $\pi$ is admissible, note that since $N$ is odd we have
\begin{align*}
\sum_{n=1}^3 \C{T}_n(z)
&= 1+\ba z+\ba\,z^2
+1-\ba z^2
+1-\ba\,z
= 3,\\
\sum_{n=4}^N \C{T}_n(z)
&= \sum_{n=4}^N \bigl(1+(-1)^n \ba\, z\bigr)
= N-3.
\end{align*}
It follows that $\sum_{n=1}^N \psi_n(z)=\psi(z)$.

\medskip
From the proof of \cref{prop:weaklyeven}, we know that the root MSFE for each seller $n\ge 3$ is $\sigma^\pi_n=\sigma$. For seller $1$, the transfer $\C{T}_1(z)=1+\ba z+\ba\,z^2$ admits the inner–outer factorization
\[
\C{T}_1(z)=\C{O}_1(z)\,\C{I}_1(z), \quad \mbox{where}\quad 
\C{O}_1(z)=\ba+\ba\,z+z^2, \qquad \C{I}_1(z)={1+\ba z+\ba\,z^2 \over \ba+\ba\,z+z^2}. \]

The fact that $\ba > 1$ implies that the two roots of the polynomial $1+\ba z+\ba\,z^2$ lie in $\mathbb{D}$, and therefore $\mathcal{I}_1(z)$ is a Blaschke product and hence an inner factor (see \citealp{GarciaMashreghiRoss2018}). It follows that
\[
\sigma^\pi_1=\frac{1}{N}\,\abs{\psi(0)}\,\abs{\C{O}_1(0)}=\frac{\ba\,\abs{\psi(0)}}{N}=\sigma.
\]
Similarly, for seller 2, we have that 
\[
\C{T}_2(z)=1-\ba\,z^2
=\C{O}_2(z)\,\C{I}_2(z),\quad \mbox{where}\quad
 \C{O}_2(z)=\ba-z^2, \quad \C{I}_2(z)={1-\ba\,z^2 \over \ba-z^2} \]
where again the fact that $\ba > 1$ implies that $\C{I}_2(z)$ is a Blaschke product and 
\[
\sigma^\pi_2=\frac{1}{N}\,\abs{\psi(0)}\,\abs{\C{O}_2(0)}=\frac{\ba\,\abs{\psi(0)}}{N}=\sigma.
\]
In conclusion, for all sellers $n \in [N]$, we have $\sigma^\pi_n = \sigma$ under policy~$\pi$, thereby proving that $\pi$ is neutral and implements $\sigma^{\pi}=\sigma$. \qed
\medskip

{\sc Proof of \cref{prop:noninv}:} Let $\psi(z)$ and $\psi_n(z)$ denote the $z$-transforms of the market demand $D_t$ and seller $n$'s demand $D_{nt}$, respectively.
Let
\[
\psi_n(z)=\frac{1}{N}\,\psi(z)\,\mathcal{T}_n(z),\qquad \text{with }\; \mathcal{T}_n(0)=1,\; \mbox{for all }n \in [N]\quad\mbox{and} \quad\sum_{n=1}^N \mathcal{T}_n(z)=N\ \text{for } z\in\mathbb{D},
\]

Suppose, by contradiction, that all the demand processes $D_{nt}$ are invertible. Thus, by \cref{lem:msfe_outer_main} we have $\sigma^{\pi}_n = \abs{\psi_{n}(0)}=\abs{\psi(0)\,\mathcal{T}_n(0)}/N=\abs{\psi(0)}/N=\sigma_{\LB}$, which violates the assumption $\sigma^{\pi}>\sigma_{\LB}$. \qed

\medskip

{\sc Proof of \cref{cor:timedomain}:} Write the (mean-corrected) market demand in its Wold form as
\[
D_t-\mu=\psi(\mathcal{B})\,\varepsilon_t,
\]
where $\{\varepsilon_t\}$ is white noise. Under an allocation with seller-$n$ transfer function
\[
\psi_n(z)=\frac{1}{N}\,\psi(z)\,\mathcal{T}_n(z),
\]
the corresponding time-domain representation is
\[
D_{nt}-\frac{\mu}{N}=\psi_n(\mathcal{B})\,\varepsilon_t
=\frac{1}{N}\,\psi(\mathcal{B})\,\mathcal{T}_n(\mathcal{B})\,\varepsilon_t.
\]
Since $\psi(\mathcal{B})$ and $\mathcal{T}_n(\mathcal{B})$ are both functions of the backshift operator, they
commute, hence
\[
D_{nt}-\frac{\mu}{N}
=\frac{1}{N}\,\mathcal{T}_n(\mathcal{B})\,\psi(\mathcal{B})\,\varepsilon_t
=\frac{1}{N}\,\mathcal{T}_n(\mathcal{B})\,(D_t-\mu).
\]
Therefore,
\[
D_{nt}
=\frac{\mu}{N}+\frac{1}{N}\,\mathcal{T}_n(\mathcal{B})\,(D_t-\mu).
\]
Now plug in the explicit polynomial forms of $\mathcal{T}_n$ from
Propositions~\ref{prop:weaklyeven}--\ref{prop:weaklyodd} and expand.

\smallskip
\noindent\textbf{(a) $N$ even.}
From Proposition~\ref{prop:weaklyeven},
\[
\mathcal{T}_n(z)=1+(-1)^n\,\frac{\sigma}{\sigma_{\LB}}\,z.
\]
Hence
\[
\mathcal{T}_n(\mathcal{B})(D_t-\mu)
=(D_t-\mu)+(-1)^n\,\frac{\sigma}{\sigma_{\LB}}\,(D_{t-1}-\mu),
\]
and substituting into $D_{nt}=\frac{\mu}{N}+\frac{1}{N}\mathcal{T}_n(\mathcal{B})(D_t-\mu)$ gives
\[
D_{nt}=\frac{D_t}{N}+(-1)^n\,\frac{\sigma}{N\,\sigma_{\LB}}\,(D_{t-1}-\mu).
\]

\smallskip
\noindent\textbf{(b) $N$ odd.}
From Proposition~\ref{prop:weaklyodd}, the transfer functions are (with $s:=\sigma/\sigma_{\LB}$)
\[
\mathcal{T}_1(z)=1+s(z+z^2),\qquad
\mathcal{T}_2(z)=1-s z^2,\qquad
\mathcal{T}_n(z)=1+(-1)^n s z\ \ (n\ge 3).
\]
Applying these polynomials to $(D_t-\mu)$ yields
\[
\mathcal{T}_1(\mathcal{B})(D_t-\mu)=(D_t-\mu)+s\big[(D_{t-1}-\mu)+(D_{t-2}-\mu)\big],
\]
\[
\mathcal{T}_2(\mathcal{B})(D_t-\mu)=(D_t-\mu)-s(D_{t-2}-\mu),
\]
\[
\mathcal{T}_n(\mathcal{B})(D_t-\mu)=(D_t-\mu)+(-1)^n s (D_{t-1}-\mu),\qquad n\ge 3.
\]
Substituting into $D_{nt}=\frac{\mu}{N}+\frac{1}{N}\mathcal{T}_n(\mathcal{B})(D_t-\mu)$ and using
$\frac{\mu}{N}+\frac{1}{N}(D_t-\mu)=\frac{D_t}{N}$ gives exactly the stated time-domain formulas:
\[
D_{1t}=\frac{D_t}{N}+\frac{\sigma}{N\sigma_{\LB}}\big[(D_{t-1}-\mu)+(D_{t-2}-\mu)\big],\qquad
D_{2t}=\frac{D_t}{N}-\frac{\sigma}{N\sigma_{\LB}}(D_{t-2}-\mu),
\]
\[
D_{nt}=\frac{D_t}{N}+(-1)^n\,\frac{\sigma}{N\sigma_{\LB}}(D_{t-1}-\mu),\qquad n\ge 3.
\]
\hfill  \qed
\medskip

{\sc Proof of \cref{lem:offsettracking_discrepancy}:}  Summing \eqref{eq:benchmark_decomp} over $n$ yields $D_t=\sum_{n=1}^N D_{nt}=D_t+\sum_{n=1}^N b_{n,t}$, hence
\begin{equation}\label{eq:sum_b_zero}
\sum_{n=1}^N b_{n,t}=0.
\end{equation}
Therefore $\sum_{n=1}^N x_{n,t}=D_t$, so $\{x_{n,t}\}$ is a fractional allocation of the $D_t$ orders.

Let $A_n(k)$ denote the counter for seller $n$ after the first $k$ orders of period $t$ have been assigned
(by definition $A_n(0)=0$ and $\widehat D_{nt}=A_n(D_t)$).  Define the discrepancy process
\[
\delta_n(k)\;:=\;A_n(k)-x_{n,t},\qquad k=0,1,\dots,D_t.
\]
Because $x_{n,t}\ge 0$, we have $\delta_n(0)=-x_{n,t}\le 0$ for all $n$.

\smallskip
At any arrival $k$, Algorithm~\ref{alg:offsettracking} chooses $n_k\in\arg\min_{n}\{A_n(k-1)-b_{n,t}\}$.
Since $x_{n,t}=D_t/N+b_{n,t}$, we have
\[
A_n(k-1)-b_{n,t}
=
\bigl(A_n(k-1)-x_{n,t}\bigr)+\frac{D_t}{N}
=
\delta_n(k-1)+\frac{D_t}{N},
\]
so the platform equivalently chooses $n_k\in\arg\min_n \delta_n(k-1)$. After assigning order $k$ to $n_k$,
\[
\delta_{n_k}(k)=\delta_{n_k}(k-1)+1,
\qquad
\delta_n(k)=\delta_n(k-1)\ \text{ for } n\neq n_k.
\]
In particular, each $\delta_n(k)$ is \emph{nondecreasing} in $k$.

\smallskip
Let $\delta_n:=\delta_n(D_t)$ be the terminal discrepancies and pick an index $i\in\arg\max_n \delta_n$.
If $\delta_i\le 0$, then $\delta_n\le 0$ for all $n$. Since
\[
\sum_{n=1}^N \delta_n
=
\sum_{n=1}^N A_n(D_t)-\sum_{n=1}^N x_{n,t}
=
D_t-D_t
=
0,
\]
it follows that $\delta_n=0$ for all $n$ and the claim holds trivially.

Suppose instead that $\delta_i>0$. Since $\delta_i(0)\le 0$ and $\delta_i(D_t)>0$, seller $i$ must have been
chosen at least once. Let $\tau$ be the (arrival) index of the \emph{last} order assigned to seller $i$.
Immediately before that assignment, seller $i$ is among the minimizers of $\delta(\tau-1)$, hence
\[
\delta_i(\tau-1)\le \delta_j(\tau-1)\qquad\text{for all } j\in[N].
\]
Because $\tau$ is the last time $i$ is chosen, we have $\delta_i(\tau-1)=\delta_i-1$. Moreover, since each
$\delta_j(k)$ is nondecreasing in $k$, $\delta_j(\tau-1)\le \delta_j$ for all $j$. Therefore,
\[
\delta_i-1\le \delta_j \qquad\text{for all } j\in[N],
\]
or equivalently $\delta_i\le \delta_j+1$ for all $j$. Taking the minimum over $j$ yields
\[
\max_{n}\delta_n=\delta_i \;\le\; \min_{n}\delta_n +1.
\]
Hence $\max_n\delta_n-\min_n\delta_n\le 1$.

\smallskip
In conclusion, since $\sum_n \delta_n=0$, we have $\min_n\delta_n\le 0\le \max_n\delta_n$. Combined with
$\max_n\delta_n\le \min_n\delta_n+1$, this implies $\min_n\delta_n\ge -1$ and $\max_n\delta_n\le 1$.
Thus $|\delta_n|\le 1$ for every $n$, i.e.,
\[
\bigl|A_n(D_t)-x_{n,t}\bigr|\le 1
\qquad\text{for all } n\in[N].
\]
Recalling that $\widehat D_{nt}=A_n(D_t)$ and $x_{n,t}=D_t/N+b_{n,t}$ completes the proof.
\hfill  \qed
\medskip

\section{Hardy-space foundations: $\mathbb{H}^2$}
\label{app:hardy_tools}

This appendix formalizes the $z$-domain tools used in \cref{subsec:hardy_tools}. Our goal is to make precise how demand-allocation filters affect seller-level forecastability, why inner factors preserve second moments while altering invertibility, and why the one-step-ahead root MSFE is determined by the outer component of the transfer function. Detailed background on Hardy spaces can be found in \cite{Rudin1987RealComplex,MartinezAvendanoRosenthal2007,Nikolski2019HardySpaces}.\smallskip

Transforming a stationary demand process into the $z$-domain is useful because time-domain convolutions become multiplications. In our setting, this means that the platform's allocation rule acts as an algebraic filter on market demand. The resulting representation separates the magnitude and phase components of demand, which in turn allows us to distinguish the part of seller-level demand that governs forecastability from the part that merely distorts phase while leaving second moments unchanged.\smallskip

Throughout, let $\mathbb{D}=\{z\in\mathbb{C}:|z|<1\}$ denote the open unit disk, let $\mathbb{T}=\{e^{-i\theta}:\theta\in[-\pi,\pi)\}$ denote its boundary, and let $\mathbb{H}^2$ denote the Hardy--Hilbert space of analytic functions on $\mathbb{D}$ with square-summable Taylor coefficients (equivalently, square-integrable boundary values on $\mathbb{T}$).\smallskip

A convenient starting point is the Wold decomposition theorem. It implies that every weakly stationary purely nondeterministic process admits a unique one-sided moving-average representation in terms of its innovation sequence. In particular, if $\{X_t\}$ is such a process, then
\[
X_t=\bar X+\sum_{k=0}^\infty x_k\,\eps_{t-k},
\qquad \{x_k\}\in\ell^2,
\]
where $\{\eps_t\}$ is a white-noise innovation sequence orthogonal to the closed linear span of the past $\{X_{t-1},X_{t-2},\dots\}$. The sequence $\{\eps_t\}$ coincides with the one-step-ahead forecast errors of the process, and the coefficients $\{x_k\}$ are uniquely determined. This representation provides the time-series foundation for the $z$-transform used below.\smallskip

\begin{dfn}[$z$-transform representation]
Let $\{X_t\}$ be a weakly stationary, purely nondeterministic Gaussian process with Wold MA$(\infty)$ representation
\[
X_t=\bar X+\sum_{k=0}^\infty x_k\,\eps_{t-k},
\qquad \{x_k\}\in\ell^2,
\]
where $\{\eps_t\}$ is a white-noise innovation sequence, normalized to have unit variance unless stated otherwise. The $z$-transform of $\{X_t\}$ is the function $\mathcal{X}\in\mathbb{H}^2$ defined by
\[
\mathcal{X}(z)=\sum_{k=0}^\infty x_k z^k,
\qquad z\in\mathbb{D}.
\]
\end{dfn}\smallskip

Under an admissible allocation policy $\pi=\bigl(\mu_n,\C{T}_n(z):n\in[N]\bigr)$, seller $n$'s demand process satisfies
\[
D_{nt}=\mu_n+\C{T}_n(\C{B})(D_t-\mu).
\]
Since market demand admits the representation $D_t-\mu=\psi(\C{B})\eps_t$, seller $n$'s allocated demand can equivalently be written as
\[
D_{nt}=\mu_n+\psi_n(\C{B})\eps_t,
\qquad
\psi_n(z):=\C{T}_n(z)\psi(z).
\]
Hence seller $n$'s forecastability is governed by the transfer function $\psi_n(z)$. The feasibility condition $D_t=\sum_{n\in[N]}D_{nt}$ implies $\sum_{n\in[N]}\psi_n(z)=\psi(z)$.\smallskip

\paragraph{Invertibility.}
A stationary process is said to be \emph{invertible} if its driving innovation sequence can be recovered from the history of the observed process. In the notation above, the representation
\[
X_t=\bar X+\mathcal{X}(\C{B})\eps_t
\]
is invertible if there exists a square-summable sequence $\{\varpi_k\}\in\ell^2$ such that
\[
\eps_t=\sum_{k=0}^\infty \varpi_k\,(X_{t-k}-\bar X).
\]
Equivalently, the shocks can be recovered by applying the inverse filter $\mathcal{X}(\C{B})^{-1}$ to the centered process. Thus, invertibility means that past observations contain enough information to reconstruct the underlying shocks. As we explain below, this is exactly the property captured by the outer component of the transfer function.\smallskip

\begin{dfn}[Inner and outer functions]
A function $\mathcal{I}\in\mathbb{H}^2$ is \textbf{inner} if its boundary values are unimodular almost everywhere on $\mathbb{T}$, that is,
\[
|\mathcal{I}(e^{-i\theta})|=1
\qquad\text{for a.e.\ }\theta\in[-\pi,\pi).
\]

A function $\mathcal{O}\in\mathbb{H}^2$ is \textbf{outer} if it it has no zeros in the open unit disk, i.e., $\mathcal{O}(z)\neq 0$ for $|z|<1$\footnote{
Equivalently, $\mathcal{O}$ is determined by its boundary modulus through the Poisson integral:
\[
\mathcal{O}(z)=c\exp\!\left\{\frac{1}{2\pi}\int_{-\pi}^{\pi}
\frac{e^{-i\theta}+z}{e^{-i\theta}-z}\,
\log|\mathcal{O}(e^{-i\theta})|\,d\theta\right\},
\qquad z\in\mathbb{D},
\]
for some unimodular constant $c\in\mathbb{T}$, with $\log|\mathcal{O}(e^{-i\theta})|\in L^1([-\pi,\pi))$. In particular, outer functions are zero-free on $\mathbb{D}$.}. We write $\mathbb{I}$ and $\mathbb{O}$ for the classes of inner and outer functions in $\mathbb{H}^2$, respectively.
\end{dfn}\smallskip

The forecasting interpretation of these two classes is central to our analysis. Inner factors are \emph{all-pass} filters: they preserve the boundary magnitude of a transfer function and therefore leave its spectral density unchanged. Outer factors, by contrast, are the minimum-phase, invertible components of a transfer function. Accordingly, they encode the part of the process that is relevant for one-step-ahead forecasting.\smallskip

\begin{lem}[All-pass property of inner factors]\label{lem:allpass}
If $\mathcal{I}\in\mathbb{I}$ and $f\in\mathbb{H}^2$, then
\[
|(\mathcal{I}f)(e^{-i\theta})|^2=|f(e^{-i\theta})|^2
\qquad\text{for a.e.\ }\theta\in[-\pi,\pi).
\]
Consequently, multiplying a transfer function by an inner factor preserves the spectral density and hence the entire autocovariance structure of the associated stationary process.
\end{lem}\smallskip

The next result is the basic structural theorem from Hardy-space theory that we use throughout.

\begin{lem}[Nevanlinna inner--outer factorization]\label{lem:inner-outer-app}
Every $\mathcal{X}\in\mathbb{H}^2$ admits a factorization
\[
\mathcal{X}(z)=\mathcal{O}(z)\,\mathcal{I}(z),
\]
where $\mathcal{O}\in\mathbb{O}$ is outer and $\mathcal{I}\in\mathbb{I}$ is inner. This factorization is unique up to multiplication by a unimodular constant.
\end{lem}\smallskip

Because inner factors preserve second moments while outer factors govern invertibility, \cref{lem:inner-outer-app} gives a canonical decomposition of any stationary demand stream into a forecast-relevant component and a pure phase distortion.\smallskip

\begin{lem}[Root MSFE and the outer factor]\label{lem:msfe_outer_app}
Let
\[
X_t=\bar X+\mathcal{X}(\C{B})\eps_t
\]
be a weakly stationary, purely nondeterministic Gaussian process with $\mathcal{X}\in\mathbb{H}^2$, and let $\mathcal{X}=\mathcal{O}\mathcal{I}$ be its inner--outer factorization. Then the one-step-ahead root mean squared forecast error of $\{X_t\}$ is
\[
\sigma_X=|\mathcal{O}(0)|.
\]
Equivalently, only the outer component of the transfer function affects one-step-ahead forecastability.
\end{lem}\smallskip

{\sc Proof of \cref{lem:msfe_outer_app}:} 
In $\mathbb{H}^2$, the one-step predictor is the orthogonal projection of $\mathcal{X}$ onto $z\mathbb{H}^2$, and the prediction error is the orthogonal complement. Multiplication by an inner function is an isometry on $\mathbb{H}^2$, hence
\[
\mathrm{dist}\bigl(\mathcal{X},\,z\mathbb{H}^2\bigr)=\mathrm{dist}\bigl(\mathcal{O},\,z\mathbb{H}^2\bigr).
\]
For any $g(z)=\sum_{k\ge0}g_k z^k\in\mathbb{H}^2$, the orthogonal complement of $z\mathbb{H}^2$ is spanned by the constant function, so $\mathrm{dist}(g, z\mathbb{H}^2)=|g_0|=|g(0)|$. Applying this to $g=\mathcal{O}$ yields $\sigma_X=|\mathcal{O}(0)|$.\qed\medskip

The intuition is simple. Since $\mathcal{I}$ is inner, the filtered sequence $\eta_t:=\mathcal{I}(\C{B})\eps_t$ is again white noise with unit variance by \cref{lem:allpass}. Hence
\[
X_t=\bar X+\mathcal{O}(\C{B})\eta_t,
\]
so the process admits an equivalent outer representation whose innovation variance is unchanged. The one-step-ahead forecast error is therefore the contemporaneous coefficient in the outer filter, namely $|\mathcal{O}(0)|$.\smallskip

Applying this result to seller-level demand gives the main text formula.

\begin{cor}[Seller-level root MSFE]\label{cor:seller_msfe_app}
Under an admissible allocation policy $\pi=\bigl(\mu_n,\C{T}_n(z):n\in[N]\bigr)$, let seller $n$'s transfer function be $\psi_n(z)=\C{T}_n(z)\psi(z)$, and let $\psi_n=\C{O}_n\C{I}_n$ be its inner--outer factorization. Then seller $n$'s one-step-ahead root MSFE under policy $\pi$ is
\[
\sigma_n^\pi=|\C{O}_n(0)|.
\]
\end{cor}\smallskip

The next result formalizes the connection between outer functions and invertibility.

\begin{lem}[Invertibility and outer functions]\label{lem:invertible_outer_app}
Let
\[
X_t=\bar X+\mathcal{X}(\C{B})\eps_t
\]
with $\mathcal{X}\in\mathbb{H}^2$. Then the process $\{X_t\}$ is invertible with respect to $\{\eps_t\}$ if and only if $\mathcal{X}$ is outer.
\end{lem}\smallskip

Thus, nontrivial inner factors represent precisely the part of the filter that is invisible from second moments yet prevents the underlying shocks from being recovered from the observed demand history.\smallskip

\begin{cor}[Seller-level invertibility]\label{cor:seller_invertibility_app}
Under an admissible allocation policy $\pi=\bigl(\mu_n,\C{T}_n(z):n\in[N]\bigr)$, seller $n$'s demand process $\{D_{nt}\}$ is invertible with respect to the aggregate demand shocks $\{\eps_t\}$ if and only if $\psi_n(z)$ is outer.
\end{cor}\smallskip

The inner part can itself be decomposed more finely. Although this is not needed for the main argument, it is useful in examples and in finite-order constructions.

\begin{rem}[Structure of inner functions]\label{rem:inner_structure_app}
Any $\mathcal{I}\in\mathbb{I}$ can be written as the product of a unimodular constant, a Blaschke product, and a singular inner factor:
\[
\mathcal{I}(z)=c\,\mathcal{B}(z)\,\mathcal{S}(z),
\qquad c\in\mathbb{T}.
\]
The Blaschke part collects zeros $\{a_k\}\subset\mathbb{D}$ through
\[
\mathcal{B}(z)=\prod_k \frac{z-a_k}{1-\overline{a_k}z},
\]
while the singular inner factor is generated by a finite positive singular measure $\nu$ on $\mathbb{T}$:
\[
\mathcal{S}(z)=\exp\!\left\{-\frac{1}{2\pi}\int_{-\pi}^{\pi}
\frac{e^{-i\theta}+z}{e^{-i\theta}-z}\,d\nu(\theta)\right\}.
\]
For more details, see \citet[Ch.~1]{Nikolski2019HardySpaces}.
\end{rem}\smallskip

The appendix results above provide the formal underpinning for the main-text discussion: the platform affects seller-level forecastability by shaping the outer components of the transfer functions $\psi_n$, while inner factors supply additional degrees of freedom that preserve second moments and help satisfy the feasibility constraint $\sum_{n\in[N]}\psi_n(z)=\psi(z)$.

\section{$\mathrm{MA}(\infty)$ Representations of Gaussian Demand and the Equivalence of Demand Representations}\label{app:MA_infty}
This appendix provides additional details for the demand assumption. In particular, we discuss about the $\mathrm{MA}(\infty)$ representation of demand, the assumption of Gaussian shocks, and clarify the equivalence between formulations based on observed demand $\{D_t\}$ and underlying demand shocks $\{\eps_t\}$ under standard invertibility conditions.
\medskip

The $\mathrm{MA}(\infty)$ representation in \eqref{def:Dt} is sufficiently general to capture any weakly stationary ARMA process, though it excludes nonstationary models like ARIMA. 
By Wold’s representation theorem \citep{BrockwellDavis2006}, there exists a valid square-summable sequence $\{\psi_k\}$ for this formulation. We also assume invertibility with respect to the demand shocks ${\eps_t}$; in our setting, this entails no real loss of generality, since demand is exogenous to the seller and only ${D_t}$ is observed.
\medskip

Because the demand shocks $\{\eps_t\}$ are Gaussian, the model can in principle generate negative realized demand and, under base-stock control, negative orders.\footnote{This is a standard issue in Gaussian inventory models; see, for example, \citet{JohnsonThompson}, \citet{LST2000}, \citet{Aviv2003}, \citet{Chen_Lee_2010}, and \citet{CGH_OR_Letter}.} We therefore interpret the Gaussian-shock formulation as a convenient approximation of  non-negative demand. This approximation works well for products with moderate demand volatility, where the probability of negative demand is negligible (see Table~4 in \citet{CGH_OR_Letter}). Additional discussion and probability estimates are provided in \ref{Append_Neg}.\medskip

{Finally, we comment on the equivalence between representations based on $\{D_t\}$ and $\{\eps_t\}$. Our analysis focuses on allocation policies defined in terms of observed demand histories $\{D_t,D_{t-1},\dots\}$. This formulation is equivalent to the one based on the shock history $\{\eps_t,\eps_{t-1},\dots\}$ when the demand is invertiable and its $z$-transform $\psi(z)$ satisfies mild technical conditions, e.g., when $\psi(z)$ belongs to the class of causal and invertible ARMA processes. However, the demand-based representation is both more natural and more transparent, as the shocks must otherwise be recovered indirectly (e.g., via spectral factorization).
}
\medskip

\medskip
\section{On the Likelihood of Negative Demand}\label{Append_Neg}

The classes of Gaussian demand models considered in \cref{assm:Normal} allow, in principle, for negative realized
market demand and allocated demand. In what follows, we quantify the probability that market
demand is negative and compare it to the probability that an individual seller observes negative allocated demand.

Recall that the market-demand process $\{D_t\}$ admits the one-sided $\mathrm{MA}(\infty)$ representation
\[
D_t=\mu+\sum_{k=0}^\infty \psi_k\,\eps_{t-k},
\]
where $\mu>0$ is the mean demand and $\{\eps_t\}$ is a Gaussian white-noise sequence with $\E[\eps_t]=0$ and
$\Var(\eps_t)=1$. Hence $D_t\sim{\cal N}(\mu,\sigma_D^2)$ with $\sigma_D^2=\sum_{k=0}^\infty \psi_k^2$, and therefore
\[
\Pro(D_t\le 0)=\Phi\!\left(-\frac{\mu}{\sigma_D}\right)
=\Phi\!\left(-\frac{1}{\mathrm{CV}}\right),
\qquad
\mathrm{CV}:=\frac{\sigma_D}{\mu}
=\frac{1}{\mu}\sqrt{\sum_{k=0}^\infty \psi_k^2},
\]
where $\Phi(\cdot)$ is the standard normal cdf and $\mathrm{CV}$ is the coefficient of variation. In particular,
$\Pro(D_t\le 0)\to 0$ as $\mathrm{CV}\to 0$, so the Gaussian approximation is most appropriate for product categories
with mean demand large relative to its standard deviation.\medskip

To illustrate how the allocation rule affects this probability at the seller level, suppose for simplicity that $N$
is even and the platform uses the allocation in \cref{prop:weaklyeven}, i.e.,
$\psi_n(z)=\frac{1}{N}\psi(z)\mathcal{T}_n(z)$ with transfer function
\[
\mathcal{T}_n(z)=1+(-1)^n\bar\alpha z,
\qquad
\bar\alpha:=\frac{\sigma}{\sigma_{\LB}},
\]
for some $\sigma\ge \sigma_{\LB}$. By \cref{cor:timedomain},
\begin{align*}
D_{nt}
&=\frac{D_t}{N}+(-1)^n\frac{\sigma}{N\sigma_{\LB}}(D_{t-1}-\mu)\\
&=\frac{\mu}{N}+\frac{1}{N}\sum_{k=0}^\infty\Big[\psi_k\,\eps_{t-k}+(-1)^n\bar\alpha\,\psi_k\,\eps_{t-1-k}\Big]\\
&=\frac{\mu}{N}+\frac{1}{N}\psi_0\,\eps_t+\frac{1}{N}\sum_{k=1}^\infty\Big[\psi_k+(-1)^n\bar\alpha\,\psi_{k-1}\Big]\eps_{t-k}.
\end{align*}
Thus $D_{nt}$ is Gaussian with mean $\mu/N$ and variance
\[
\sigma_{n}^2=\frac{1}{N^2}\left(\psi_0^2+\sum_{k=1}^\infty\big[\psi_k+(-1)^n\bar\alpha\,\psi_{k-1}\big]^2\right),
\]
so
\[
\Pro(D_{nt}\le 0)
=\Phi\!\left(-\frac{\mu}{\sqrt{\psi_0^2+\sum_{k=1}^\infty[\psi_k+(-1)^n\bar\alpha\,\psi_{k-1}]^2}}\right)
=\Phi\!\left(-\frac{1}{\mathrm{CV}_n}\right),
\]
where
\[
\mathrm{CV}_n
:=\frac{1}{\mu}\sqrt{\psi_0^2+\sum_{k=1}^\infty\big[\psi_k+(-1)^n\bar\alpha\,\psi_{k-1}\big]^2}.
\]
By the inequality $(x+y)^2\le 2(x^2+y^2)$,
\[
\sum_{k=1}^\infty\big[\psi_k+(-1)^n\bar\alpha\,\psi_{k-1}\big]^2
\le
2\sum_{k=1}^\infty\big(\psi_k^2+\bar\alpha^2\psi_{k-1}^2\big),
\]
which yields the crude bound
\[
\mathrm{CV}_n\le \mathrm{CV}\,\sqrt{2(1+\bar\alpha^2)}.
\]
Although conservative, this bound provides a simple screening criterion for when negative seller-level demand is
negligible. Using U.S.\ Census Bureau data, Table~4 in \cite{CGH_OR_Letter} reports coefficients of variation of
monthly sales across retail sectors. For example, for \emph{Furniture and Home Furnishings} (FHF), $\mathrm{CV}$
ranges from $0.081$ to $0.108$ over the last 30 years, while for \emph{Electronics and Appliance} (EA) it ranges from
$0.166$ to $0.244$. Using midpoints of these ranges, $\Pro(D_t\le 0)$ is on the order of $10^{-26}$ for FHF and
$10^{-7}$ for EA. In turn, the bound above implies that for $\bar\alpha$ below roughly $5$ (FHF) and $2$ (EA),
$\Pro(D_{nt}\le 0)$ remains below $5\%$. If market demand is i.i.d., these thresholds increase to about $7$ and $3$,
respectively.

\medskip
\section{Derivation of \cref{eq:MSFE_SES}}\label{Append_Der}

In this appendix, we derive the value of the MSFE when sellers use simple exponential smoothing to forecast demand. In particular, we consider the case in which market demand is i.i.d. (i.e., $\psi(z)$ is a constant) and 
seller $n$'s demand allocation admits the $z$-transform representation $\psi_n(z)=\frac{1}{N}\,\psi(z)\,\C{T}_n(z)$. For simplicity we consider the case with an even number of sellers. According to  
\cref{prop:weaklyeven}, in this case the platform sets
\[ \psi_n(z)=\frac{\abs{\psi(0)}}{N}\,\Big(1+(-1)^n\,\alpha z\Big), \qquad \text{where}\qquad \alpha= \frac{N \,\sigma}{\abs{\psi(0)}}. \] 

It follows that 
\begin{align*} (\tilde\sigma_n^\pi(\lambda))^2 &=\frac{1}{2\pi}\,\frac{\abs{\psi(0)}^2}{N^2}\int_{-\pi}^{\pi} \left|\frac{1-e^{-i\theta}}{1-(1-\lambda)e^{-i\theta}}\right|^2\, \left|1+(-1)^n\alpha e^{-i\theta}\right|^2\,\D\theta \\ &=\frac{\abs{\psi(0)}^2}{N^2}\cdot \frac{1}{2\pi}\int_{-\pi}^{\pi}\Big|F(e^{-i\theta})\Big|^2\,\D\theta, \qquad \text{where}\quad F(z):=\frac{1-z}{1-(1-\lambda)z}\,\big(1+(-1)^n\,\alpha z\big). \end{align*}

Using a geometric-series expansion, 
\begin{align*} F(z) &=\Big(1-\lambda\sum_{k=1}^{\infty}(1-\lambda)^{k-1}z^k\Big)\big(1+(-1)^n\,\alpha z\big) \\ &=1+(-1)^n\,\alpha z-\lambda\sum_{k=1}^{\infty}(1-\lambda)^{k-1}z^k -\lambda (-1)^n\alpha\sum_{k=1}^{\infty}(1-\lambda)^{k-1}z^{k+1}. \end{align*} 

Therefore, writing $F(z)=\sum_{k\ge 0}f_k z^k$, the coefficients are \[ f_0=1,\qquad f_1=(-1)^n\alpha-\lambda,\qquad f_k=-\lambda(1-\lambda)^{k-2}\big(1-\lambda+(-1)^n\alpha\big),\ \ k\ge 2. \] 
Since $F\in\mathbb{H}^2$, Parseval's identity yields \[ \frac{1}{2\pi}\int_{-\pi}^{\pi}\Big|F(e^{-i\theta})\Big|^2\,\D\theta =\sum_{k=0}^{\infty}\abs{f_k}^2 =1+\abs{(-1)^n\,\alpha-\lambda}^2+\sum_{k=2}^{\infty}\lambda^2(1-\lambda)^{2(k-2)}\abs{1-\lambda+(-1)^n\,\alpha}^2 . \] 

Summing the geometric series, we obtain \[ (\tilde\sigma_n^\pi(\lambda))^2 = \frac{\abs{\psi(0)}^2}{N^2}\left[ 1+\abs{(-1)^n\alpha-\lambda}^2+\frac{\lambda}{2-\lambda}\abs{1-\lambda+(-1)^n\alpha}^2 \right]. \]

\section{Heterogeneous Replenishment Lead Times}\label{subsec:leadtime_extension}

Our baseline model assumes that all sellers face identical and zero replenishment lead times, regardless of the fulfillment mode they select. We now extend the framework to allow for heterogeneous lead times that may vary both across sellers and across fulfillment modes. In particular, for each seller $n\in[N]$ let $L^{\ti{FBP}}\ge 0$ and $L^{\ti{FBM}}\ge 0$ denote, respectively, the replenishment lead time when seller $n$ operates under FBP and under FBM. Throughout the discussion, we use $\bar L_n$ to denote the relevant (mode-dependent) lead time faced by seller $n$ under the prevailing fulfillment mode.\smallskip

Define seller $n$’s \emph{lead-time demand} as
\[
\mathbb{D}_{n,t}
:=\sum_{\ell=0}^{\bar L_n} D_{n,t+\ell},
\]
that is, the cumulative demand realized during the replenishment delay.
Just before demand in period $t+1$ is realized, seller $n$ chooses an order-up-to level $S_{nt}$ measurable with
respect to $\mathcal{F}_{nt}$ and orders
\[
Q_{nt}=\bigl(S_{nt}-I_{nt}\bigr)^+,
\]
where $I_{nt}$ denotes net inventory after serving demand in period $t$. As in the baseline model, seller $n$ selects $S_{nt}$ to minimize expected holding and backorder
costs associated with lead-time demand:
\[
S_{nt}
=\argmin_{S}\;
\mathbb{E}\!\left[
\,\bar h_n\, (S-\mathbb{D}_{n,t+1})^+
+ b_n\, (S-\mathbb{D}_{n,t+1})^-
\;\middle|\; \mathcal{F}_{nt}
\right].
\]

Under \cref{assm:Normal}, $\mathbb{D}_{n,t+1}$ is Gaussian conditional on $\mathcal{F}_{nt}$, and the optimal base-stock
policy takes the familiar form
\begin{equation}\label{eq:basestock_leadtime}
S_{nt} \;=\; \bar m_{nt} + \zeta_n\,\bar\sigma_n,
\qquad
\zeta_n:=\Phi^{-1}\!\left(\frac{b_n}{\bar h_n+b_n}\right),
\end{equation}
where $\bar m_{nt}
:=\mathbb{E}\!\left[\mathbb{D}_{n,t+1}\mid \mathcal{F}_{nt}\right]$
is the conditional mean forecast of lead-time demand and
\[
\bar\sigma_n^2
:=\Var\!\left[\mathbb{D}_{n,t+1}-\bar m_{nt}\mid \mathcal{F}_{nt}\right]
\]
is the corresponding MSFE of the lead-time demand. Thus, seller $n$’s safety stock equals $\zeta_n\,\bar\sigma_n$, implying that inventory holdings depend on the
predictability of cumulative lead-time demand rather than one-step-ahead demand alone. Just as in \eqref{eq:sellerprofit_modes}, seller $n$'s expected per-period operating profit can be written as
\begin{equation}\label{eq:sellerprofit_leadtime}
U_n
\;=\;
(r-\rho-\bar f_n)\,\mu_n
\;-\;
K_n\,\bar\sigma_n,
\qquad
K_n
:=
\bar h_n\,\zeta_n
+(\bar h_n+b_n)\,\mathcal{L}(\zeta_n).
\end{equation}
A key difference between \eqref{eq:sellerprofit_modes} and \eqref{eq:sellerprofit_leadtime} is that under lead times the operating mode affects not only the effective cost parameters $(\bar h_n,\bar f_n)$ but also the relevant lead time: FBP and FBM may entail different replenishment delays, so that $\bar L_n\in\{L^{\ti{FBP}},L^{\ti{FBM}}\}$ and, consequently, the corresponding root MSFE $\bar\sigma_n$ (computed for lead-time demand) may differ across modes even under the same demand allocation policy.\smallskip

Let us now extend the result in \cref{lem:msfe_outer_main} to account for a general leadtime.\smallskip

\begin{prop}\label{prop:MSFE-leadtime} Suppose seller $n$ faces a demand process with $z$-transform $\psi_n(z)=\mathcal{O}_n(z)\,\mathcal{I}_n(z)$, where $\mathcal{O}_n\in\mathbb{O}$ and $\mathcal{I}_n\in\mathbb{I}$. Then, the MSFE $\bar{\sigma}^2_n$ of the lead-time demand is given by
\begin{equation}\label{eq:MSFE-leadtime}\bar{\sigma}^2_n=\sum_{\ell=0}^{\bar L_n} \Big(\sum_{k=0}^{\bar L_n-\ell} \theta_{nk}\Big)^2, \qquad \mbox{where}\quad \theta_{nk}:={1 \over 2 \pi}\int_{-\pi}^{\pi} e^{i\,k\,x}\,\C{O}_n(e^{-i\,x})\, \D x.\end{equation}
\end{prop}

The following corollary directly adapts the previous result to our proposed class of demand allocation policies in Propositions~\ref{prop:weaklyeven} and \ref{prop:weaklyodd}.\smallskip

\begin{cor}
Consider the demand allocation $\psi_n(z)=\frac{1}{N}\,\psi(z)\,\mathcal{T}_n(z)$, where the transfer function $\mathcal{T}_n(z)$ is given in Propositions~\ref{prop:weaklyeven} or \ref{prop:weaklyodd}. Then, the MSFE $\bar{\sigma}^2_n$ of seller $n$'s lead-time demand is given by \eqref{eq:MSFE-leadtime}, where the sequence $\{\theta_{nk}\}$ satisfies
\begin{align*}
\mbox{If $N$ is even, or if $N$ is odd and $n\ge 3$:} 
&\qquad 
\theta_{nk}=\frac{1}{N}\,\bigl(\alpha_n\,\psi_k-\psi_{k-1}\bigr),\\
\mbox{If $N$ is odd and $n=1$:} 
&\qquad 
\theta_{nk}=\frac{1}{N}\,\bigl(\alpha_n\,\psi_k+\alpha_n\,\psi_{k-1}+\psi_{k-2}\bigr),\\
\mbox{If $N$ is odd and $n=2$:} 
&\qquad 
\theta_{nk}=\frac{1}{N}\,\bigl(\alpha_n\,\psi_k-\psi_{k-2}\bigr),
\end{align*}
where $\alpha_n:=(-1)^n\,\sigma/\sigma_{\LB}$ and we adopt the convention that $\psi_k=0$ for $k<0$.
\end{cor}\smallskip

Equipped with this corollary, for any given value of $\sigma$ each seller can compute the root MSFE of its lead-time demand under each fulfillment mode by applying \cref{prop:MSFE-leadtime} using the corresponding mode-specific replenishment leadtime. The seller then selects its preferred mode by maximizing its expected operating payoff as characterized in \eqref{eq:sellerprofit_leadtime}. This mode choice also pins down the seller’s average safety stock through the base-stock rule \eqref{eq:basestock_leadtime}. Anticipating these endogenous mode choices, the platform subsequently determines the optimal value of $\sigma$ by maximizing the appropriate lead-time version of \eqref{eq:platform-z}, in which both the set of sellers operating under FBP, $N^{\ti{FBP}}(\sigma)$, and the associated cumulative safety stock $\Gamma^{\ti{FBP}}(\sigma)$ are adjusted to reflect the mode-dependent lead-time root MSFE and safety-stock calculations described above.

{\sc Proof of \cref{prop:MSFE-leadtime}:}
Fix seller $n$ and write the mean--zero demand as
\[
\widetilde D_{nt}:=D_{nt}-\mu_n=\psi_n(\mathcal{B})\,\eps_t,
\]
where $\mathcal{B}$ is the backshift operator and $\{\eps_t\}$ is the Gaussian sequence of shocks appearing in the market-demand Wold representation in \eqref{def:Dt}.\smallskip

Since $\psi_n(z)=\mathcal{O}_n(z)\mathcal{I}_n(z)$ with $\mathcal{O}_n\in\mathbb{O}$ and $\mathcal{I}_n\in\mathbb{I}$, we can write
\[
\widetilde D_{nt}=\mathcal{O}_n(\mathcal{B})\mathcal{I}_n(\mathcal{B})\,\eps_t=\mathcal{O}_n(\mathcal{B})\,\hat{\eps}_{nt},
\qquad
\hat{\eps}_{nt}:=\mathcal{I}_n(\mathcal{B})\,\eps_t.
\]
Because $\mathcal{I}_n$ is an inner function (i.e., an all-pass filter), $\{\hat{\eps}_{nt}\}$ is again a Gaussian white-noise sequence. Thus,
$\widetilde D_{nt}=\mathcal{O}_n(\mathcal{B})\,\hat{\eps}_{nt}$ is the Wold representation of $\widetilde D_{nt}$, i.e.,
\begin{equation}\label{eq:wold_n}
\widetilde D_{nt}=\sum_{k=0}^\infty \theta_k\,\hat{\eps}_{n,t-k},
\end{equation}
where $\{\theta_k\}_{k\ge 0}$ are the Taylor coefficients of the outer factor $\mathcal{O}_n$, namely
$\mathcal{O}_n(z)=\sum_{k=0}^\infty \theta_k z^k$. Since $\mathcal{O}_n\in H^2$, these coefficients can be recovered via the Fourier formula
\[
\theta_k=\frac{1}{2\pi}\int_{-\pi}^{\pi} e^{ikx}\,\mathcal{O}_n(e^{-ix})\,dx, \qquad k\ge 0.
\]

Now consider the lead-time demand
\[
\mathbb{D}_{n,t+1}=\sum_{\ell=0}^{\bar L_n} D_{n,t+1+\ell}
\;=\;
(\bar L_n+1)\mu_n+\sum_{\ell=0}^{\bar L_n}\widetilde D_{n,t+1+\ell}.
\]
Because the constant mean term is $\mathcal{F}_{nt}$-measurable, it does not affect the MSFE, so we work with
$\sum_{\ell=0}^{\bar L_n}\widetilde D_{n,t+1+\ell}$. Substituting \eqref{eq:wold_n} and re-indexing gives
\begin{align*}
\sum_{\ell=0}^{\bar L_n}\widetilde D_{n,t+1+\ell}
&=\sum_{\ell=0}^{\bar L_n}\sum_{k=0}^\infty \theta_k\,\hat{\eps}_{n,t+1+\ell-k}\\
&=\underbrace{\sum_{\ell=0}^{\bar L_n}\sum_{k=\ell+1}^\infty \theta_k\,\hat{\eps}_{n,t+1+\ell-k}}_{\text{$\mathcal{F}_{nt}$-measurable}}
\;+\;
\sum_{\ell=0}^{\bar L_n}\sum_{k=0}^{\ell} \theta_k\,\hat{\eps}_{n,t+1+\ell-k}.
\end{align*}
The first term only involves innovation shocks $\hat{\eps}_{n,s}$ with $s\le t$ and is therefore measurable with respect
to $\mathcal{F}_{nt}=\sigma(\{D_{ns}\}_{s\le t})$. Hence, the forecast error
$\mathbb{D}_{n,t+1}-\e[\mathbb{D}_{n,t+1}\mid\mathcal{F}_{nt}]$ equals the orthogonal (innovation) part
\[
\sum_{\ell=0}^{\bar L_n}\sum_{k=0}^{\ell} \theta_k\,\hat{\eps}_{n,t+1+\ell-k}.
\]
Collecting terms by innovation time $t+1+\ell$ (set $u=\ell-k$) yields
\[
\mathbb{D}_{n,t+1}-\e[\mathbb{D}_{n,t+1}\mid\mathcal{F}_{nt}]
=
\sum_{\ell=0}^{\bar L_n}\Big(\sum_{k=0}^{\bar L_n-\ell}\theta_k\Big)\,\hat{\eps}_{n,t+1+\ell}.
\]
Since $\{\hat{\eps}_{nt}\}$ is white noise with unit variance, the conditional variance (and hence the MSFE) is
\[
\bar\sigma_n^2
=
\Var\!\left(\mathbb{D}_{n,t+1}-\e[\mathbb{D}_{n,t+1}\mid\mathcal{F}_{nt}] \,\middle|\, \mathcal{F}_{nt}\right)
=
\sum_{\ell=0}^{\bar L_n}\Big(\sum_{k=0}^{\bar L_n-\ell}\theta_k\Big)^2,
\]
which proves the claim.
\hfill \qed

\end{appendices}

\end{document}